\documentclass[11pt]{article}

\RequirePackage{amssymb} %
\RequirePackage{amsmath} %
\RequirePackage{latexsym}
\RequirePackage{mathrsfs}
\RequirePackage{varioref} %

\DeclareMathOperator*{\argmin}{argmin}

\DeclareMathOperator*{\argmax}{argmax}

\DeclareMathOperator*{\cov}{cov}
\DeclareMathOperator*{\var}{var}

\DeclareMathOperator*{\supp}{supp}

\newcommand{\avg}{\frac{1}{n}\sum_{i=1}^n}
\newcommand{\ind}{\indicator}

\newcommand{\R}{\ensuremath{\mathbb{R}}}

\newcommand{\Exp}{\ensuremath{\mathbb{E}}} %
\newcommand{\Prob}{\ensuremath{\mathbb{P}}} %

\newcommand{\A}{\ensuremath{\mathbf{A}}} %

\newcommand{\Z}{\ensuremath{\mathbf{Z}}}

\usepackage{bbm}
\newcommand{\indicator}{\ensuremath{\mathbbm{1}}}

\DeclareMathAlphabet{\mathpzc}{OT1}{pzc}{m}{it}

\newcommand{\normal}{\ensuremath{\mathcal{N}}}

\makeatletter
\newcommand*{\indep}{%
  \mathbin{%
    \mathpalette{\@indep}{}%
  }%
}
\newcommand*{\nindep}{%
  \mathbin{%
    \mathpalette{\@indep}{\not}%
  }%
}
\newcommand*{\@indep}[2]{%
  \sbox0{$#1\perp\m@th$}%
  \sbox2{$#1=$}%
  \sbox4{$#1\vcenter{}$}%
  \rlap{\copy0}%
  \dimen@=\dimexpr\ht2-\ht4-.2pt\relax
  \kern\dimen@
  {#2}%
  \kern\dimen@
  \copy0 %
} 
\makeatother

\usepackage[dvipsnames]{xcolor}
\colorlet{mylinkcolor}{RoyalBlue}%

\usepackage{enumitem}

\usepackage{titlesec}

\usepackage{wrapfig}
\usepackage{tikz-cd}
\usepackage{tikz}
\usepackage{pgfplots}
\usepackage{float}
\usepackage{afterpage}

\usepackage{xr-hyper}
\pdfoutput=1
\usepackage{hyperref}
\hypersetup{
    colorlinks=true,
    linkcolor=black,
    citecolor=black,
    filecolor=black,
    urlcolor=black,
}

\usepackage{calc} %
\usepackage{chngpage}

\usepackage{array}
\newcolumntype{L}[1]{>{\raggedright\let\newline\\\arraybackslash\hspace{0pt}}p{#1}}
\newcolumntype{C}[1]{>{\centering\let\newline\\\arraybackslash\hspace{0pt}}p{#1}}
\newcolumntype{R}[1]{>{\raggedleft\let\newline\\\arraybackslash\hspace{0pt}}p{#1}}
\usepackage[labelfont=bf]{caption}
\usepackage{booktabs}

\usepackage{mathtools}

\RequirePackage{amsthm}

\theoremstyle{definition}

\newtheorem{proposition}{Proposition}
\newtheorem{lemma}{Lemma}

\newtheorem{definition}{Definition}
\newtheorem{theorem}{Theorem}

\newcounter{partialIndepSection}
\newcommand{\aAssump}{A\arabic{partialIndepSection}}

\newtheoremstyle{theoremSuppressedNumber}{}{}{}{}{\bfseries}{.}{ }{\thmname{#1}\thmnote{ (\mdseries #3)}}
\theoremstyle{theoremSuppressedNumber}
\newtheorem{partialIndepAssump}{Assumption \aAssump \addtocounter{partialIndepSection}{1}}
\refstepcounter{partialIndepSection}

\usepackage{thmtools}
\declaretheoremstyle[notefont=\bfseries,notebraces={}{},%
    headpunct={},postheadspace=1em]{mystyle}

\usepackage{setspace}
\usepackage[longnamesfirst]{natbib}
\usepackage{ulem}
\usepackage{subfig}
\usepackage{wrapfig}
\usepackage{float}

\usepackage{epstopdf}

\title{\textbf{Assessing Sensitivity to Unconfoundedness: \\Estimation and Inference}\footnote{This paper was presented at the 2018 Western Economic Association International Conference, the 2019 Stata Conference Chicago, the 2020 World Congress of the Econometric Society, the DC-MD-VA Econometrics Workshop 2020, University of Southern California, University of Toronto, and the 2020 SEA Conference. We thank participants at those seminars and conferences, as well as Karim Chalak, Toru Kitagawa, and John Pepper. We thank Paul Diegert for excellent research assistance. Masten thanks the National Science Foundation for research support under Grant No.\ 1943138.}}

\author{Matthew A. Masten\footnote{Department of Economics, Duke University,
\texttt{matt.masten@duke.edu}} 
\qquad
Alexandre Poirier\footnote{ Department of Economics, Georgetown University,
\texttt{alexandre.poirier@georgetown.edu}}
\qquad
Linqi Zhang\thanks{Department of Economics, Boston College,
\texttt{linqi.zhang@bc.edu}}
}
\date{December 31, 2020}

\begin{document}
\maketitle

\begin{abstract}
This paper provides a set of methods for quantifying the robustness of treatment effects estimated using the unconfoundedness assumption (also known as selection on observables or conditional independence). Specifically, we estimate and do inference on bounds on various treatment effect parameters, like the average treatment effect (ATE) and the average effect of treatment on the treated (ATT), under nonparametric relaxations of the unconfoundedness assumption indexed by a scalar sensitivity parameter $c$. These relaxations allow for limited selection on unobservables, depending on the value of $c$. For large enough $c$, these bounds equal the no assumptions bounds. Using a non-standard bootstrap method, we show how to construct confidence bands for these bound functions which are uniform over all values of $c$. We illustrate these methods with an empirical application to effects of the National Supported Work Demonstration program. We implement these methods in a companion Stata module for easy use in practice.
\end{abstract}



\newcommand{\smalleps}{\varepsilon_\text{smaller}}

\onehalfspacing
\normalsize

\newpage
\section{Introduction}\label{sec:intro}

A core goal of causal inference is to identify and estimate effects of a treatment variable on an outcome variable. 
A common assumption used to identify such effects is unconfoundedness, which says that potential outcomes are independent of treatment conditional on covariates. This assumption is also known as conditional independence, selection on observables, ignorability, or exogenous selection; see \cite{Imbens2004} for a survey. 
This assumption is not refutable, meaning that the data alone cannot tell us whether it is true. 
Nonetheless, empirical researchers may wonder: How important is this assumption in their analyses? 
Put differently: How sensitive are their results to failures of the unconfoundedness assumption?

A large literature on sensitivity analysis has developed to answer this question. Moreover, researchers widely acknowledge that answering this question is an important step in empirical research. For example, in their figure 1, \cite{CaliendoKopeinig2008} describe the workflow of a standard analysis using selection on observables. Their fifth and final step in this workflow is to perform sensitivity analysis to the unconfoundedness assumption. \citet[section 6.2]{ImbensWooldridge2009}, \citet[chapter 22]{ImbensRubin2015}, and \cite{AtheyImbens2017} all also recommend that researchers conduct sensitivity analyses to assess the importance of non-refutable identifying assumptions. In particular, \cite{AtheyImbens2017} describe these methods as ``a systematic way of doing the sensitivity analyses that are routinely done in empirical work, but often in an unsystematic way.''

Most of the existing approaches to assessing unconfoundedness rely on strong auxiliary assumptions, however. For example, they often assume treatment effects are homogeneous and that all unobserved confounding arises due to a single unobserved variable whose distribution is parametrically specified, like a binary or normal distribution. They also often assume a parametric functional form for potential outcomes, like a logit model for binary potential outcomes or a linear model for continuous potential outcomes. These assumptions---which are not needed for identification of the baseline model when unconfoundedness holds---raise a new question: Are the findings of these sensitivity analyses \emph{themselves} sensitive to these extra auxiliary assumptions?

In this paper, we provide a set of tools for assessing the sensitivity of the unconfoundedness assumption which do not rely on strong auxiliary assumptions that are not used for the baseline analysis. 
We do this by studying nonparametric relaxations of the unconfoundedness assumption. Specifically, we apply the identification results of \cite{MastenPoirier2018}, who consider a class of assumptions called \emph{conditional $c$-dependence}. 
This class measures relaxations of conditional independence by a single scalar parameter $c \in [0,1]$. 
This parameter $c$ is the largest difference between the propensity score and the probability of treatment conditional on covariates and an unobserved potential outcome.
Hence it has a straightforward interpretation as a deviation from conditional independence, as measured in probability units.
For any positive $c$, conditional independence only partially holds, and so we cannot learn the exact value of our treatment effect parameters, like the average treatment effect (ATE) or the average effect of treatment on the treated (ATT). Instead, we only get bounds. 
\cite{MastenPoirier2018} derive closed-form expressions for these bounds as a function of $c$. 
Setting $c= 0$ yields the baseline model where unconfoundedness holds. Setting $c= 1$ yields the other extreme where no assumptions on selection are made, and hence gives the no assumption bounds as in \cite{Manski1990}. The bounds are monotonic in $c$, so that small values of $c$ give narrow bounds while larger values of $c$ give wider bounds. Just how wide these bounds are---and hence how sensitive one's results are---depends on the data.

While \cite{MastenPoirier2018} studied identification of treatment effects under nonparametric relaxations of unconfoundedness, they did not study estimation or inference. We do that in this paper. First we propose sample analog estimators of the bounds on the conditional quantile treatment effect (CQTE), the conditional average treatment effect (CATE), the ATE, and the ATT. We do this using flexible parametric first step estimators of the propensity score and the conditional quantile function of the observed outcomes given treatment and covariates. Although such parametric restrictions are not required for our identification theory, the analysis of inference is complicated and non-standard even with these parametric first step estimators. Doing inference based on fully nonparametric first step estimators will likely require deriving and applying more general asymptotic theory for non-Hadmard differentiable functionals than currently exists. Hence we leave that to future work. Moreover, note that our approach of using nonparametric identification results paired with flexible parametric estimators is analogous to what is commonly done in the baseline model which imposes unconfoundedness: Identification is shown nonparametrically but many commonly used estimators are based on flexible parametric first step estimators. For example, see chapter 13 in \cite{ImbensRubin2015}. 

We derive the asymptotic distributions of our bound estimators using the delta method for Hadamard directionally differentiable functionals from \cite{FangSantos2014}. We then show consistency of a non-standard bootstrap based on estimating the analytical Hadamard directional derivatives of our bound functionals. This step again involves using the recent results of \cite{FangSantos2014}. We show how to construct confidence bands for the bound functions which are uniform over all values of $c \in [0,1]$. We also provide a sufficient condition on the propensity score and the distribution of the covariates under which we can do inference using the standard nonparametric bootstrap. Finally, we show how to implement our analysis in an empirical illustration to the National Supported Work Demonstration program (MDRC \citeyear{MDRC1983}). Using the techniques developed in this paper, and implemented in an accompanying Stata module, researchers can quantify the robustness of treatment effects estimated using the unconfoundedness assumption.

The rest of this paper is organized as follows. 
In the rest of this section we briefly discuss the related literature.
In section \ref{sec:model} we summarize the identification results from \cite{MastenPoirier2018}.
We also discuss how to use and interpret these results in practice.
Section \ref{sec:estimation} describes the definition of our bound estimators.
Section \ref{sec:asymptotics} provides the corresponding asymptotic estimation and inference theory for these estimators.
Section \ref{sec:bootstrap} describes how to use these inference results to conduct bootstrap based inference.
In section \ref{sec:standardboot} we give sufficient conditions under which standard bootstrap approaches are valid.
Section \ref{sec:empirical} shows how to use our methods in an empirical illustration.
Appendix \ref{sec:prelimest} contains theoretical results and proofs for our first step estimators. Appendices \ref{sec:proofs}, \ref{sec:BootstrapProofs}, and \ref{sec:standardbootProof} have proofs for our main results. 
Appendix \ref{sec:HDDformulas} gives the full expressions for various analytical Hadamard directional derivatives used in our analysis.
Appendix \ref{sec:CQTEandCATEbootstrapResults} provides several additional results.

\subsection*{Related Literature}

We conclude this section with a brief literature review. As mentioned earlier, there is a large existing literature that studies how to relax unconfoundedness. This includes \cite{RosenbaumRubin1983sensitivity}, \cite{Mauro1990}, \cite{Rosenbaum1995, Rosenbaum2002}, \cite{RobinsRotnitzkyScharfstein2000}, \cite{Imbens2003}, \cite{AltonjiElderTaber2005, AltonjiElderTaber2008}, \cite{IchinoMealliNannicini2008}, \cite{HosmanHansenHolland2010}, \cite{Krauth2016}, \cite{KallusMaoZhou2019}, \cite{Oster2019}, and \cite{CinelliHazlett2020}, among others. Here we discuss the most closely related work and several recent papers. For further details about the related literature, see section 1 in \cite{MastenPoirier2018} for identification and Appendix D in \cite{MastenPoirier2020} for estimation and inference.

A key feature of our results is that they are based on the fully nonparametric analysis of \cite{MastenPoirier2018}. There are only a few other alternative nonparametric analyses available in the literature. The first is \cite{IchinoMealliNannicini2008}, who require that all variables are discretely distributed. In contrast, we allow for continuous outcomes, covariates, and unobservables. Their approach requires picking a vector of sensitivity parameters that determines the joint distribution of the discrete observable and unobservable variables. In contrast, our approach uses a scalar sensitivity parameter. Finally, unlike us, they do not provide any formal results for doing estimation or inference. The second is \cite{Rosenbaum1995,Rosenbaum2002}, who proposed a sensitivity analysis for unconfoundedness within the context of doing randomization inference based on the sharp null hypothesis of no unit level treatment effects for all units in the data set. Like our approach, he only uses a scalar sensitivity parameter and also does not rely on a parametric model for outcomes or treatment assignment probabilities. His approach, however, is based on finite sample randomization inference (for more discussion, see chapter 5 of \citealt{ImbensRubin2015}). This approach to inference is conceptually distinct from the approach we use based on repeated sampling from a large population. For this reason, we view these different approaches to inference in sensitivity analyses as complementary. Finally, \cite{KallusMaoZhou2019} study bounds on CATE under the same nonparametric relaxations defined by \cite{Rosenbaum1995,Rosenbaum2002}. Unlike him, however, they take a large population view. They propose sample analog kernel estimators based on an implicit characterization of the identified set using extrema. They show consistency of these estimators, but they do not provide any inference results. As we discuss later, this is a key distinction because inference in this setting is non-standard.

A few recent papers provide methods for assesssing unconfoundedness in parametric linear models. This includes \cite{Oster2019} and \cite{CinelliHazlett2020}. These results rely on the assumption that outcomes are linear functions of treatment and covariates, among other parametric assumptions. In contrast, we build on the selection on observables literature that has emphasized nonparametric identification. That literature emphasizes that identification by functional form is often implausible. Sensitivity analyses that rely on functional form assumptions are subject to the same criticism: Findings that one's results are robust to violations of unconfoundedness can be driven primarily from the parametric functional form restrictions. To address this, our estimation and inference results are based on nonparametric sensitivity analyses that do not require parametric assumptions.

Finally, we discuss the relationship with our own previous work. As noted earlier, our paper provides estimation and inference results for population bounds derived in \cite{MastenPoirier2018}. That paper did not provide any estimation or inference theory. \cite{MastenPoirier2020} builds on those results in several ways: First, they extend the identification analysis to identification of distributional treatment effect parameters, with a focus on assessing the importance of the rank invariance assumption. Second, they provide some asymptotic distributional results for sample analog estimators of the average treatment effect (ATE), the conditional average treatment effect (CATE), and the conditional quantile treatment effect (CQTE), among other results. Those results are limited in a variety of ways, which we discuss next.

Specifically, our paper differs from the results in \cite{MastenPoirier2020} in several important ways: (1) Our paper allows for both discrete and continuous covariates, whereas that paper focused on the case where all covariates are discrete. In particular, to allow for continuous covariates we develop a different estimator of the bound functions. This is important since many empirical applications, like ours in section \ref{sec:empirical}, use continuous covariates. (2) Our results allow for all possible values of $c \in [0,1]$, whereas that paper restricted attention to small values of the sensitivity parameter $c$ (see their assumption A2.1). This is also important for practice and requires a substantial amount of new theoretical work. (3) Our results use the \cite{FangSantos2014} bootstrap based on estimators of analytical Hadamard directional derivatives to do inference. That paper instead used the numerical delta method bootstrap of \cite{HongLi2015}. Our approach allows us to avoid choosing the step size tuning parameter required for the numerical delta method bootstrap, although our estimators of the analytical Hadamard directional derivatives also have tuning parameters. (4) Unlike that paper, we also discuss inference on the average effect of treatment on the treated (ATT). (5) In this paper we provide a new companion Stata module implementing our results.

\section{Population Bounds on Treatment Effects}\label{sec:model}

In this section we describe the model and review standard results on point identification of treatment effects under unconfoundedness. 
We then describe how we relax unconfoundedness. 
Finally, we review the bounds on treatment effects derived by \cite{MastenPoirier2018} when unconfoundedness is relaxed.

\subsection*{Model and Baseline Point Identification Results}

We use the standard potential outcomes model. Let $X \in \{0, 1 \}$ be an observed binary treatment. Let $Y_1$ and $Y_0$ denote the unobserved potential outcomes. The observed outcome is
\begin{equation}\label{eq:potential outcomes}
	Y = X Y_1 + (1-X) Y_0.
\end{equation}
Let $W \in \R^{d_W}$ denote a vector of observed covariates, which may be discrete, continuous, or mixed. Let $\mathcal{W} = \supp(W)$ denote the support of $W$. Let
\[
	p_{x \mid w} = \Prob(X=x \mid W=w)
\]
denote the observed generalized propensity score.

It is well known that the conditional distributions of potential outcomes $Y_1 \mid W$ and $Y_0 \mid W$ are point identified under the following two assumptions:
\begin{itemize}
\item[] Unconfoundedness: $X \indep Y_1 \mid W$ and $X \indep Y_0 \mid W$.

\item[] Overlap: $p_{1 \mid w} \in (0,1)$ for all $w \in \mathcal{W}$.
\end{itemize}
Consequently, any functional of the distributions of $Y_1 \mid W$ and $Y_0 \mid W$ is also point identified. 
We focus on two leading examples: The average treatment effect, $\text{ATE} = \Exp(Y_1 - Y_0)$ and the average treatment effect for the treated, $\text{ATT} = \Exp(Y_1 - Y_0 \mid X=1)$. We also consider the conditional quantile treatment effects $\text{CQTE}(\tau \mid w) = Q_{Y_1 \mid W}(\tau \mid w) - Q_{Y_0 \mid W}(\tau \mid w)$ and the conditional average treatment effect $\text{CATE}(w) = \Exp(Y_1 - Y_0 \mid W=w)$.

\subsection*{Sensitivity Analysis: Relaxing Unconfoundedness}

As discussed in section \ref{sec:intro}, the overlap assumption is refutable and hence can be directly verified from the data. 
The unconfoundedness assumption, however, is not refutable. 
Consequently, like much of the literature reviewed in section \ref{sec:intro}, we perform a sensitivity analysis. 
This entails replacing unconfoundedness with a weaker assumption and investigating how this changes the conclusions we can draw about our parameter of interest. 
Specifically, we define the following class of assumptions, which we call \emph{conditional $c$-dependence} (\citealt{MastenPoirier2018}):

\begin{definition}\label{def:c-dep}
Let $x \in \{ 0, 1 \}$. 
Let $w\in\mathcal{W}$. 
Let $c$ be a scalar between 0 and 1. 
Say $X$ is \emph{conditionally $c$-dependent} with $Y_x$ given $W$ if
\begin{equation}\label{eq:c-indep1}
	\sup_{y_x \in \supp(Y_x \mid W=w)} | \Prob(X=1 \mid Y_x=y_x,W=w) - \Prob(X=1 \mid W=w) | \leq c.
\end{equation}
holds for all $w \in \mathcal{W}$.
\end{definition}

When $c = 0$, conditional $c$-dependence is equivalent to $X \indep Y_x \mid W$. 
For $c > 0$, however, we allow for violations of unconfoundedness by allowing the unobserved conditional probability
\[
	\Prob(X=1 \mid Y_x=y_x, W=w)
\]
to differ from the observed propensity score
\[
	\Prob(X=1 \mid W=w)
\]
by at most $c$. 
Thus we actually allow for some selection on unobservables, since treatment assignment may depend on $Y_x$, but in a constrained manner. 
For sufficiently large $c$, however, conditional $c$-dependence imposes no constraints on the relationship between $Y_x$ and $X$. 
This happens when $c \geq \overline{C}$ where $\overline{C} =\sup_{w \in \mathcal{W}} \max \{ p_{1 \mid w}, p_{0 \mid w} \}$. 
When $c \in (0,\overline{C})$, conditional $c$-dependence imposes some constraints on treatment assignment, but it does not require conditional independence to hold exactly. 
For this reason, we call it a \emph{conditional partial independence} assumption. 
Thus our sensitivity analysis replaces unconfoundedness with
\begin{itemize}
\item[] Conditional Partial Independence: $X$ is conditionally $c$-dependent with $Y_1$ and $Y_0$ given $W$. %
\end{itemize}

\subsection*{Treatment Effect Bounds}

By relaxing conditional independence our main parameters of interest---ATE and ATT---are no longer point identified. Instead they are partially identified: We can bound them from above and from below. As $c$ gets close to zero, however, these bounds collapse to a point. Hence for small $c$ these bounds can be quite narrow. The goal of a sensitivity analysis is to understand how the shape and width of these bounds changes as $c$ varies from 0 to 1.  

These bounds were derived in \cite{MastenPoirier2018}, which we summarize here. Although that paper studied both continuous and binary outcomes, here we only summarize the results for continuous $Y_x$. All of our parameters of interest can be written in terms of bounds on the quantile regressions $Q_{Y_x \mid W}(\tau \mid w)$. Under the conditional partial independence assumption stated above and some regularity conditions, \cite{MastenPoirier2018} showed that $[\underline{Q}^c_{Y_x \mid W}(\tau \mid w), \overline{Q}^c_{Y_x \mid W}(\tau \mid w)]$ are are sharp bounds on this quantile regression, uniformly in $\tau$, $x$, and $w$, where
\begin{align}\label{eq:quantile upperbound}
	\overline{Q}^c_{Y_x \mid W}(\tau \mid w)
	&= Q_{Y \mid X,W}\left( \overline{t}(\tau,x,w) \mid x,w\right) \\[1em]
	\text{where} \qquad \overline{t}(\tau,x,w) &= \min\left\{\tau + \frac{c}{p_{x \mid w}}\min\{\tau,1-\tau\},\frac{\tau}{p_{x \mid w}},1\right\} \notag
\end{align}
and 
\begin{align}\label{eq:quantile lowerbound}
	\underline{Q}^c_{Y_x \mid W}(\tau \mid w)
	&= Q_{Y \mid X,W}\left( \underline{t}(\tau,x,w) \mid x,w\right) \\[1em]
	\text{where} \qquad \underline{t}(\tau,x,w) &= \max\left\{\tau - \frac{c}{p_{x \mid w}}\min\{\tau,1-\tau\},\frac{\tau-1}{p_{x \mid w}} +1,0\right\}. \notag
\end{align}

Taking differences of these bounds for $x=1$ and $x=0$ yields sharp bounds on the conditional quantile treatment effect $\text{CQTE}(\tau \mid w)$, uniformly in $\tau$ and $w$:
\begin{align*}\label{eq:CQTE_bounds}
	&\left[\underline{\text{CQTE}}^c(\tau \mid w), \overline{\text{CQTE}}^c(\tau \mid w)\right] \notag \\
	&\hspace{25mm} \equiv 
	\left[\underline{Q}^c_{Y_1 \mid W}(\tau \mid w) - \overline{Q}^c_{Y_0 \mid W}(\tau \mid w), \,\overline{Q}^c_{Y_1 \mid W}(\tau \mid w) - \underline{Q}^c_{Y_0 \mid W}(\tau \mid w)\right].
\end{align*}
Integrating these bounds over $\tau$ yields sharp bounds on $\text{CATE}(w)$, uniformly in $w$:
\[
	\left[\underline{\text{CATE}}^c(w),\overline{\text{CATE}}^c(w) \right]
	\equiv
	\left[\int_0^1 \underline{\text{CQTE}}^c(\tau \mid w) \; d\tau, \int_0^1 \overline{\text{CQTE}}^c(\tau \mid w) \; d\tau \right].
\]
Further integrating over the marginal distribution of $W$ yields sharp bounds on ATE:
\[
	\left[\underline{\text{ATE}}^c,\overline{\text{ATE}}^c \right]
	\equiv \Big[ \Exp \big( \underline{\text{CATE}}^c(W) \big), \, \Exp \big( \overline{\text{CATE}}^c(W) \big) \Big]
\]
To obtain bounds on ATT, let
\[
	\underline{E}_x^c(w) = \int_0^1 \underline{Q}_{Y_x}^c(\tau \mid w) \; d\tau 
	\qquad \text{and} \qquad
	\overline{E}_x^c(w) = \int_0^1 \overline{Q}_{Y_x}^c(\tau \mid w) \; d\tau
\]
denote bounds on $\Exp(Y_x \mid W=w)$. 
Averaging these over the marginal distribution of $W$ yields bounds on $\Exp(Y_x)$, denoted by
\[
	\underline{E}_x^c = \Exp \big( \underline{E}_x^c(W) \big)
	\qquad \text{and} \qquad
	\overline{E}_x^c = \Exp \big( \overline{E}_x^c(W) \big).
\]
This yields the following bounds on ATT:
\begin{align}
	\left[\Exp(Y \mid X=1) - \frac{\overline{E}_0^c - p_0\Exp(Y \mid X=0)}{p_1}, \, \Exp(Y \mid X=1) - \frac{\underline{E}_0^c - p_0\Exp(Y \mid X=0)}{p_1}\right] \label{eq:ATTbounds}
\end{align}
where $p_x = \Prob(X=x)$ for $x \in \{0,1\}$. Finally, note that all of these bounds are sharp.

\subsection*{Breakdown Points}

So far we've discussed sharp bounds on various parameters of interest as a function of the sensitivity parameter $c$. In addition to the bounds themselves, it is common to analyze \emph{breakdown points} for various conclusions of interest. For example, suppose that under the baseline model ($c=0$) we find that $\text{ATE} > 0$. We then ask: How much can we relax unconfoundedness while still being able to conclude that the ATE is nonnegative? To answer this question, define the breakdown point for the conclusion that the ATE is nonnegative as
\begin{equation}\label{eq:ATEbreakdownPoint}
	c_\textsc{bp} = \sup \{ c \in [0,1]: \left[\underline{\text{ATE}}^c,\overline{\text{ATE}}^c \right] \subseteq [0,\infty)\}.
\end{equation}
This number is a quantitative measure of the robustness of the conclusion that ATE is positive to relaxations of the key identifying assumption of unconfoundedness. Breakdown points can be defined for other parameters and  conclusions as well. See \cite{MastenPoirier2020} for more discussion and additional references.

\subsection*{Interpreting Conditional $c$-Dependence}

We conclude this section by giving some suggestions for how to interpret conditional $c$-dependence in practice. In particular, what values of $c$ are large? What values are small? Here we summarize and extend the discussion on page 321 of \cite{MastenPoirier2018}. We illustrate these interpretations in our empirical analysis in section \ref{sec:empirical}.

Let $W_k$ denote a component of $W$. Denote the propensity score by
\[
	p_{1 \mid W}(w_{-k},w_k) = \Prob(X=1 \mid W=(w_{-k},w_k) ).
\]
Let 
\[
	p_{1 \mid W_{-k}}(w_{-k}) = \Prob(X=1 \mid W_{-k}=w_{-k})
\]
denote the leave-out-variable-$k$ propensity score. This is just the proportion of the population who are treated, conditional on only $W_{-k}$. Consider the random variable
\[
	\Delta_k = | p_{1 \mid W}(W_{-k}, W_k) - p_{1 \mid W_{-k}}(W_{-k}) |.
\]
This difference is a measure of the impact on the observed propensity score of adding $W_k$, given that we already included $W_{-k}$. Conditional $c$-dependence is defined by a similar difference, except there we add the unobservable $Y_x$ given that we already included $W$. Hence we suggest using the distribution of $\Delta_k$ to calibrate values of $c$. For example, you could examine the 50th, 75th, and 90th quantiles of $\Delta_k$, along with the upper bound on the support, $\bar{c}_k = \max \supp(\Delta_k)$. You may also find it useful to plot an estimate of the density of $\Delta_k$. All of these reference values can be compared to the breakdown point $c_\textsc{bp}$ for a specific conclusion of interest. Specifically, if $c_\textsc{bp}$ is larger than the chosen reference value, then the conclusion of interest could be considered robust. In contrast, if $c_\textsc{bp}$ is smaller than the chosen reference value, then the conclusion of interest could be considered sensitive. You may also want to see where $c_\textsc{bp}$ lies relative to the distribution of $\Delta_k$. This can be done by computing $F_{\Delta_k}(c_\textsc{bp})$.

While you could do this for all covariates $k$, it may be helpful to restrict attention to covariates $k$ that have a sufficiently large impact on the baseline point estimates. For example, suppose we are interested in the ATE. Let $\text{ATE}_{-k}$ denote the ATE estimand obtained in the baseline selection on observables model using only the covariates $W_{-k}$. Let $\text{ATE}$ denote the ATE estimand obtained in the baseline model using all the covariates. Then
\[
	\left| \frac{\text{ATE} - \text{ATE}_{-k} }{\text{ATE}} \right|
\]
denotes the effect of omitting covariate $k$ on the ATE point estimand, as a percentage of the baseline estimand that uses all covariates in $W$. You may want to restrict attention to covariates $k$ for which this ratio is relatively large. We illustrate this approach in our empirical analysis in section \ref{sec:empirical}.

\section{Estimation}\label{sec:estimation}

In the previous section we assumed the entire population distribution of $(Y,X,W)$ was known. In practice we only have a finite sample $\{ (Y_i,X_i,W_i) \}_{i=1}^n$ from this distribution. In this section we explain how to use this finite sample data to estimate the population bounds of section \ref{sec:model}. We give the corresponding asymptotic theory in section \ref{sec:asymptotics} where we obtain the joint limiting distribution of treatment effects bounds. We describe how to perform bootstrap based inference on these bounds in section \ref{sec:bootstrap}.

As shown in section \ref{sec:model}, all of our bounds can be constructed from the marginal distribution of $W$ and the bounds on $Q_{Y_x \mid W}$ given in equations \eqref{eq:quantile upperbound} and \eqref{eq:quantile lowerbound}. 
These bounds on $Q_{Y_x \mid W}$, in turn, depend on just two features of the data:
\begin{enumerate}
\item The conditional quantile function $Q_{Y \mid X,W}(\tau \mid x,w)$.

\item The propensity score $p_{x \mid w} = \Prob(X=x \mid W=w)$.
\end{enumerate}
In both cases, we can use parametric, semiparametric, or nonparametric estimation methods. 
In this paper we focus on flexible parametric approaches.
Even in this case the asymptotic distribution theory is non-standard and quite complicated.
We discuss this point further in the conclusion, section \ref{sec:conclusion}.

In section \ref{sec:estimationFirstStep} we describe our first step estimators of these two functions. Given these estimators, we then construct sample analog estimates of our bound functions in a second step. We describe these estimators in section \ref{sec:estimationSecondStep}.

\subsection{First Step Quantile Regression and Propensity Score Estimation}\label{sec:estimationFirstStep}

We estimate $Q_{Y \mid X,W}$ by a linear quantile regression of $Y$ on flexible functions of $(X,W)$ that we denote by $q(X,W) \in \R^{d_q}$. For example, $q(x,w)$ could be $(1,x,w)$, $(1,x, w, x \cdot w)$, or could contain additional interactions between the treatment indicator $X$ and functions of the covariates $W$. For $\tau \in (0,1)$, let 
\begin{align*}
	\widehat{\gamma}(\tau) &= \argmin_{a \in \R^{d_q}} \sum_{i=1}^n \rho_\tau(Y_i - a' q(X_i,W_i))
\end{align*}
be the estimated coefficients from a linear quantile regression of $Y$ on $q(X,W)$ at the quantile $\tau$. Here $\rho_\tau(s) = s (\tau - \indicator(s<0))$ is the check function. Let $\widehat{Q}_{Y \mid X,W}(\tau \mid x,w) = q(x,w)'\widehat{\gamma}(\tau)$ denote this estimator.

We estimate the propensity score by maximum likelihood. In particular, specify the parametric model
\[
	\Prob(X=1 \mid W=w) = F( r(w)'\beta_0)
\]
where $F$ is a known cdf, $r(w)$ is a known vector function, and $\beta_0$ is an unknown constant vector. The functions $r(w)$ could simply be $(1,w)$ or may contain functions of $w$, like squared or interaction terms. For notational simplicity, we will assume throughout the paper that $r(w) = w$. 
Given this assumption, the dimension of $\beta_0$ is $d_W$, the length of $W$. Suppose $\beta_0$ lies in the parameter space $\mathcal{B} \subseteq \R^{d_W}$.

This specification for the propensity score includes the probit and logit estimators as special cases. Those estimators are commonly used in the literature; for example, see chapter 13 of \cite{ImbensRubin2015}.
Let $\widehat{\beta}$ denote the maximum likelihood estimate of $\beta$:
\[
	\widehat{\beta} = \argmax_{\beta \in \mathcal{B}} \sum_{i=1}^n \log L(X_i,W_i'\beta)
\]
where
\[
	L(x,w'\beta) = F(w'\beta)^x(1-F(w'\beta))^{1-x}.
\]
For each $x\in \{0,1\}$, let $\widehat{p}_{x|w} = L(x,w'\widehat\beta)$ denote our propensity score estimator.

\subsection{Second Step Estimation of the Bound Functions}\label{sec:estimationSecondStep}

Given the first step estimators from section \ref{sec:estimationFirstStep}, we obtain the following sample analog estimators of the CQTE bound functions defined in equations \eqref{eq:quantile upperbound} and \eqref{eq:quantile lowerbound}:
\begin{align*}
	\widehat{\overline{Q}}^c_{Y_x \mid W}(\tau \mid w)
	&= \widehat{Q}_{Y \mid X,W}( \widehat{\overline{t}}(\tau,x,w) \mid x,w) \\[1em]
	\text{where} \qquad \widehat{\overline{t}}(\tau,x,w) &= \min\left\{\tau + \frac{c}{\widehat{p}_{x \mid w}}\min\{\tau,1-\tau\},\frac{\tau}{\widehat{p}_{x \mid w}},1\right\}
\end{align*}
and 
\begin{align*}
	\widehat{\underline{Q}}^c_{Y_x \mid W}(\tau \mid w)
	&= \widehat{Q}_{Y \mid X,W}( \widehat{\underline{t}}(\tau,x,w) \mid x,w) \\[1em]
	\text{where} \qquad \widehat{\underline{t}}(\tau,x,w) &= \max\left\{\tau - \frac{c}{\widehat{p}_{x \mid w}}\min\{\tau,1-\tau\},\frac{\tau-1}{\widehat{p}_{x \mid w}} +1,0\right\}.
\end{align*}
As discussed in section \ref{sec:model}, averaging these over $\tau \in (0,1)$ yields sample analog estimates of bounds on $\text{CATE}(w)$, which we can then use to get bounds on $\text{ATE}$. This approach requires estimation of extremal quantiles, however---estimation for $\tau$'s close to 0 or 1. This is well known to be a delicate problem (see \citealt{ChernozhukovFernandezValKaji2017} for details). So in this paper we use a common solution: Fixed trimming of the extremal quantiles. We do this by modifying the quantile bound estimators to ensure that the quantile index lies in $[\varepsilon,1-\varepsilon]$ for some fixed and known $\varepsilon \in (0,0.5)$. Specifically, this yields the trimmed estimators of the quantile bounds
\begin{equation}
	\widehat{\overline{Q}}^c_{Y_x \mid W}(\tau \mid w)
	= \widehat{Q}_{Y \mid X,W}\left( \max\{ \min \{ \widehat{\overline{t}}(\tau, x, w), 1- \varepsilon \}, \varepsilon \} \mid x,w\right)\label{eq:CQupperB}
\end{equation}
and 
\begin{equation}
	\widehat{\underline{Q}}^c_{Y_x \mid W}(\tau \mid w)
	=
	\widehat{Q}_{Y \mid X,W}\left( \max \{ \min \{ \widehat{\underline{t}} (\tau,x,w), 1 - \varepsilon \}, \varepsilon \} \mid x,w\right). \label{eq:CQlowerB}
\end{equation}
We use these estimators for the rest of the paper. Common choices of $\varepsilon$ are $0.05$ or $0.01$. In our asymptotic analysis we hold $\varepsilon$ fixed with sample size. In principle we could generalize the results to allow $\varepsilon \rightarrow 0$ as $n \rightarrow \infty$, but this would complicate the analysis of inference, which is already non-standard for other reasons. 
Since we fix $\varepsilon$ throughout, we omit $\varepsilon$ from the notation for brevity, except when necessary.

We next estimate the CQTE bounds by taking differences of the quantile bound estimators:
\begin{align*}
	&\left[\widehat{\underline{\text{CQTE}}}^c(\tau \mid w), \widehat{\overline{\text{CQTE}}}^c(\tau \mid w)\right] \notag \\
	&\hspace{25mm} \equiv 
	\left[\widehat{\underline{Q}}^c_{Y_1 \mid W}(\tau \mid w) - \widehat{\overline{Q}}^c_{Y_0 \mid W}(\tau \mid w), \,\widehat{\overline{Q}}^c_{Y_1 \mid W}(\tau \mid w) - \widehat{\underline{Q}}^c_{Y_0 \mid W}(\tau \mid w)\right].
\end{align*}
Since our CATE bounds are simply the integral of the CQTE bounds over all the quantiles $\tau$, we can estimate them by
\[
	\left[\widehat{\underline{\text{CATE}}}^c(w),\widehat{\overline{\text{CATE}}}^c(w) \right]
	\equiv
	\left[\int_0^1 \widehat{\underline{\text{CQTE}}}^c(\tau \mid w) \; d\tau, \int_0^1 \widehat{\overline{\text{CQTE}}}^c(\tau \mid w) \; d\tau \right].
\]
A second integration over $w$ with respect to the marginal distribution of $W$ yields bounds on $\text{ATE}$. Like much of the literature, we use the empirical distribution of $W$ to estimate the marginal distribution of $W$. This yields the following estimator of our $\text{ATE}$ bounds:
\[
	\left[ \widehat{\underline{\text{ATE}}}^c, \widehat{\overline{\text{ATE}}}^c \right] 
	= \left[ \avg \widehat{\underline{\text{CATE}}}^c(W_i), \, \avg \widehat{\overline{\text{CATE}}}^c(W_i) \right].
\]
Next consider the estimation of the $\text{ATT}$ bounds. Let
\[
	\widehat{\underline{E}}_0^c = \avg \int_0^1 \widehat{\underline{Q}}_{Y_0}^c(\tau \mid W_i) \; d\tau 
	\qquad \text{and} \qquad
	\widehat{\overline{E}}_0^c = \avg \int_0^1 \widehat{\overline{Q}}_{Y_0}^c(\tau \mid W_i) \; d\tau.
\]
For $x \in \{0,1\}$ let
\[
	\widehat{\Exp}(Y \mid X=x) = \frac{\sum_{i=1}^n Y_i \indicator(X_i=x)}{\sum_{i=1}^n \indicator(X_i=x)}
	\qquad \text{and} \qquad
	\widehat{p}_x = \avg \indicator(X_i = x).
\]
We can then estimate the ATT bounds by replacing the population quantities in \eqref{eq:ATTbounds} with their estimators that we just defined.

For $c = 0$, our estimated upper and lower bounds are equal and give point estimates of the various parameters of interest. For $c > 0$, our bounds have positive width. To use these bounds in a sensitivity analysis, we recommend producing the following plot: Pick a grid $\{ c_1,\ldots,c_K \} \subseteq [0,1]$ of values for $c$. Compute our bound estimates on this grid and plot them against these values of $c$. Then compute and plot confidence bands for these bound estimates against $c$ as well; we describe how to compute these bands in section \ref{sec:bootstrap}. We illustrate all of these steps in our empirical analysis of section \ref{sec:empirical}.

\section{Asymptotic Theory}\label{sec:asymptotics}

In this section we provide formal results on the consistency and limiting distributions of the estimators we described in section \ref{sec:estimation}. In section \ref{sec:bootstrap} we show how to use these results to do inference based on a non-standard bootstrap. In section \ref{sec:standardboot} we provide sufficient conditions under which standard bootstrap inference is valid.

\subsection{Convergence of the First Step Estimators}

Throughout this paper we assume that we observe a random sample.

\begin{partialIndepAssump}[Random Sample]\label{assn:iid}
$\{(Y_i,X_i,W_i)\}_{i=1}^n$ are iid.
\end{partialIndepAssump}

Our first step estimators are standard in the literature. Hence we only briefly review the main assumptions and results for these estimators. For completeness, we provide a formal analysis in appendix \ref{sec:prelimest}.

We assume that both the propensity score and quantile regression functions are correctly specified:
\[
	\Prob(X=x\mid W=w) = L(x,w'\beta_0)
\]
and
\[
	Q_{Y \mid X,W}(\tau \mid x,w) 
	= q(x,w)'\gamma_0(\tau)
\]
for all $\tau \in [\varepsilon,1-\varepsilon]$. 

Since the first step estimators consist of linear quantile regression and maximum likelihood estimation, their $\sqrt{n}$-convergence to Gaussian elements can be shown under standard assumptions and arguments. For example, see \cite{NeweyMcFadden1994}. Moreover, the convergence of $\widehat{\gamma}(\tau)$ to $\gamma_0(\tau)$ is uniform over $\tau \in [\varepsilon,1-\varepsilon]$. Formally, as we show in appendix \ref{sec:prelimest} lemma \ref{lemma:prelim estimators}, 
\[
	\sqrt{n}
	\begin{pmatrix}
		\widehat{\beta} - \beta_0\\
		\widehat{\gamma}(\tau) - \gamma_0(\tau)
	\end{pmatrix}
	\rightsquigarrow \mathbf{Z}_1(\tau),
\]
where $\mathbf{Z}_1(\cdot)$ is a mean-zero Gaussian process in $\R^{d_W} \times \ell^\infty([\varepsilon,1-\varepsilon],\R^{d_q})$ with continuous paths. The covariance kernel of this process is defined in appendix \ref{sec:prelimest}, equation \eqref{eq:prelimest_covkernel}. Also see appendix \ref{sec:prelimest} for the formal assumptions under which this result holds.

\subsection{Convergence of the Second Step Estimators}\label{sec:SecondStepEstConvergence}

Next we consider the limiting distribution of our various second step estimators.

\subsubsection*{The CATE Bounds}

We start with equations \eqref{eq:CQupperB} and \eqref{eq:CQlowerB}, our estimators of the conditional quantile bounds. 
The population conditional quantile bounds, equations \eqref{eq:quantile upperbound} and \eqref{eq:quantile lowerbound}, are known functions of $\theta_0 = (\beta_0,\gamma_0)$. 
Define
\begin{align*}
	\overline{\Gamma}_1(x,w,\tau,\theta)
	&= q(x,w)'\gamma\left(\max\left\{\min\left\{\tau + \frac{c}{L(x,w'\beta)}\min\{\tau,1-\tau\},\frac{\tau}{L(x,w'\beta)},1-\varepsilon\right\},\varepsilon\right\}\right)\\
	\underline{\Gamma}_1(x,w,\tau,\theta)
	&= q(x,w)'\gamma\left(\min\left\{\max\left\{\tau - \frac{c}{L(x,w'\beta)}\min\{\tau,1-\tau\},\frac{\tau-1}{L(x,w'\beta)} +1,\varepsilon\right\},1-\varepsilon\right\}\right)
\end{align*}
Throughout the paper, we let $\Gamma_j = (\overline{\Gamma}_j, \underline{\Gamma}_j)$ for $j\geq 1$.  Evaluating these at $\theta_0$ gives the trimmed population conditional quantile bounds. Evaluating these at $\widehat{\theta}$ gives their sample analog estimators. Define
\[
	\overline{\Gamma}_2(x,w,\theta) = \int_0^1 \overline{\Gamma}_1(x,w,\tau,\theta) \;d\tau
	\qquad \text{ and } \qquad
	\underline{\Gamma}_2(x,w,\theta) = \int_0^1 \underline{\Gamma}_1(x,w,\tau,\theta) \;d\tau.
\]
Then
\[
	\left[\underline{\text{CATE}}_\varepsilon^c(w),\overline{\text{CATE}}_\varepsilon^c(w) \right]
	=
	\Big[ \underline{\Gamma}_2(1,w,\theta_0)- \overline{\Gamma}_2(0,w,\theta_0), \,
	\overline{\Gamma}_2(1,w,\theta_0) - \underline{\Gamma}_2(0,w,\theta_0) \Big]
\]
are the trimmed population CATE bounds. We estimate them by
\[
	\left[\widehat{\underline{\text{CATE}}}^c(w),\widehat{\overline{\text{CATE}}}^c(w) \right]
	\equiv
	\left[\underline{\Gamma}_2(1,w,\widehat\theta)- \overline{\Gamma}_2(0,w,\widehat\theta), \,
	\overline{\Gamma}_2(1,w,\widehat\theta) - \underline{\Gamma}_2(0,w,\widehat\theta)\right].
\]
If these mappings were Hadamard differentiable in $\theta$ at $\theta_0$, we could use the functional delta method to show that the above estimators have limiting Gaussian distributions and converge at $\sqrt{n}$ rates. Because they depend on the $\min$ and $\max$ functions these mappings are not Hadamard differentiable. They are, however, Hadamard directionally differentiable (HDD); see definition \ref{def:HDD} in appendix \ref{sec:proofs}. It turns out that this weaker version of differentiability is sufficient to establish their (non-Gaussian) limiting distribution. 

To formally derive the limiting distribution of the CATE estimators, we show that the mapping
\[
	\Gamma_2(x,w,\cdot): \R^{d_W} \times \ell^\infty([\varepsilon,1-\varepsilon],\R^{d_q}) \to \R^2
\]
is Hadamard directionally differentiable at $\theta_0$ tangentially to $\R^{d_W} \times \mathscr{C}([\varepsilon,1-\varepsilon],\R^{d_q})$. Here $\mathscr{C}(A,B)$ is the set of continuous functions from $A$ to $B$.

As a technical assumption, we restrict the complexity of the space that the quantile regression coefficient $\gamma_0(\cdot)$ lives in. Specifically, we assume that it is in a H\"older ball. To define this parameter space precisely, let $\mathscr{C}_m(\mathcal{D})$ denote the set of $m$-times continuously differentiable functions $f : \mathcal{D} \rightarrow \R$, where $m$ is an integer and $\mathcal{D}$ be an open subset of $\R^{d_q}$. Denote the differential operator by
\[
	\nabla^\lambda = \frac{\partial^{|\lambda |}}{\partial x_1^{\lambda_1} \cdots \partial x_{d_q}^{\lambda_{d_q}}}
\]
where $\lambda = (\lambda_1,\ldots,\lambda_{d_q})$ is a $d_q$-tuple of nonnegative integers and $| \lambda | = \lambda_1 + \cdots + \lambda_{d_q}$. Let $\nu \in (0,1]$. Let $\| \cdot \|$ without any subscripts denote the $\R^{d_q}$-Euclidean norm. Define the H\"older norm of $f: \mathcal{D} \rightarrow \R$ by
\[
	\| f \|_{m,\infty,\nu} = \max_{| \lambda | \leq m} \sup_{x \in \text{int}(\mathcal{D})} | \nabla^{\lambda} f(x) | + \max_{| \lambda | = m} \sup_{x,y \in \text{int} (\mathcal{D}), x \neq y} \frac{ | \nabla^\lambda f(x) - \nabla^\lambda f(y) |}{\| x - y \|^\nu}.
\]
For any $B > 0$, let $\mathscr{C}_{m,\nu}^B(\mathcal{D}) = \{ f \in \mathscr{C}_m(\mathcal{D}) : \| f \|_{m, \infty, \nu} \leq B \}$ denote a H\"older ball.

\begin{partialIndepAssump}[Quantile regression regularity]\label{assn:quant reg regularity}
\medskip
Let $m$ be an integer with $m \geq 3$ and $\nu \in (0,1]$. Let $B > 0$. Then $\gamma_0 \in \mathcal{G}$ where $\mathcal{G} \subseteq \mathscr{C}_{m,\nu}^B([\smalleps,1-\smalleps])^{d_q}$ for some $\smalleps \in (0,\varepsilon)$.
\end{partialIndepAssump}

In this assumption we assume $m \geq 3$ to obtain bounded third derivatives of $\gamma_0$. 
In appendix \ref{sec:prelimest} we state several additional standard regularity conditions that we use to obtain asymptotic normality of the first step estimators; see assumptions A\ref{assn:prop score}--A\ref{assn:quant reg} starting on page \pageref{assn:prop score}. We continue to maintain these assumptions here. As a first preliminary result, we use these assumptions to derive the limiting distribution of the CQTE bound estimators; see proposition \ref{prop:CQTE convergence} in appendix \ref{sec:proofs}. Using that result, we can then derive the limiting distribution of the CATE bound estimators.

\medskip

\begin{proposition}[CATE convergence]\label{prop:CATE convergence}
Fix $w \in \mathcal{W}$. Suppose A\ref{assn:iid}, A\ref{assn:quant reg regularity}, and A\ref{assn:prop score}--A\ref{assn:quant reg} hold.  Fix $c \in [0,1]$. Then
\[
	\sqrt{n}
	\begin{pmatrix}
		\widehat{\overline{\text{CATE}}}^c( w) - \overline{\text{CATE}}_\varepsilon^c(w)\\
		\widehat{\underline{\text{CATE}}}^c(w) - \underline{\text{CATE}}_\varepsilon^c(w)
	\end{pmatrix}
	\xrightarrow{d} \mathbf{Z}_{\text{CATE}}(w),
\]
where $\mathbf{Z}_\text{CATE}(w)$ is a random vector in $\R^2$ whose distribution is characterized in the proof.
\end{proposition}

In the statement of this result we deferred the full characterization of $\mathbf{Z}_\text{CATE}(w)$ to the proof. To get a brief idea of what it looks like, however, consider the first component. It is
\[
	\textbf{Z}_{\text{CATE}}^{(1)}(w) = \overline{\Gamma}_{2,\theta_0}'(1,w,\mathbf{Z}_1) - \underline{\Gamma}_{2,\theta_0}'(0,w,\mathbf{Z}_1)
\]
where $\overline{\Gamma}_{2,\theta_0}'(x,w,\mathbf{Z}_1)$ is the Hadamard directional derivative of $\overline{\Gamma}_{2}$ evaluated at $\mathbf{Z}_1$, the limiting distribution of the first step estimators. See page \pageref{eq:defOfGamma2HDD} for the expression for $\overline{\Gamma}_{2,\theta_0}'$. Likewise, $\underline{\Gamma}_{2,\theta_0}'(x,w,\mathbf{Z}_1)$ is the Hadamard directional derivative of $\underline{\Gamma}_2$ evaluated at $\mathbf{Z}_1$. Although $\mathbf{Z}_1$ is Gaussian, the HDDs are continuous but generally nonlinear functionals. Hence the distribution of $\mathbf{Z}_\text{CATE}(w)$ is non-Gaussian. In section \ref{sec:bootstrap} we show how to use a non-standard bootstrap to approximate its distribution.

\subsubsection*{The ATE Bounds}

Next we derive the limiting distribution of our ATE bound estimators. Let
\[
	\overline\Gamma_3(x,\theta) = \int_\mathcal{W} \overline{\Gamma}_2(x,w,\theta) \; dF_W(w) \quad \text{ and } \qquad
	\underline\Gamma_3(x,\theta) 	= \int_\mathcal{W} \underline{\Gamma}_2(x,w,\theta) \; dF_W(w).
\]
Then
\[
	[\underline{\text{ATE}}_\varepsilon^c, \overline{\text{ATE}}_\varepsilon^c ] = \left[ \underline{\Gamma}_3(1,\theta_0) - \overline{\Gamma}_3(0,\theta_0), \,\overline{\Gamma}_3(1,\theta_0) - \underline{\Gamma}_3(0,\theta_0) \right]
\]
are the trimmed population ATE bounds. We estimate them by
\[
	\left[\widehat{\underline{\text{ATE}}}^c,\widehat{\overline{\text{ATE}}}^c \right]
	\equiv
	\left[\avg\left(\underline{\Gamma}_2(1,W_i,\widehat\theta)- \overline{\Gamma}_2(0,W_i,\widehat\theta)\right), \,
	\avg\left(\overline{\Gamma}_2(1,W_i,\widehat\theta) - \underline{\Gamma}_2(0,W_i,\widehat\theta)\right)\right].
\]
Unlike $\Gamma_2$, the $\Gamma_3$ mapping depends on $F_W$, which is unknown. Here we estimate it by the empirical distribution of $W$. 

Next, let $\delta > 0$ and define
\[
	\mathcal{B}_\delta = \{ \beta \in \mathcal{B} : \|\beta - \beta_0\| \leq \delta\}
	\qquad \text{and} \qquad
	 L_\beta(x,w'\beta) = \frac{\partial}{\partial\beta}L(x,w'\beta).
\]
The following assumption bounds the inverse ratio of the squared propensity score to its derivative with respect to the parameter $\beta$. This assumption holds under common parametric specifications for the propensity score, like logit or probit. It also holds if strong overlap holds. Moreover, under our other assumptions, note that strong overlap holds when $W$ has finite support.

\begin{partialIndepAssump}\label{assn:propscoremoment}
There is a $\delta > 0$ such that
\[
	\Exp \left( \sup_{ \beta\in\mathcal{B}_\delta}\left\|\frac{L_\beta(x,W'\beta)}{L(x,W'\beta)^2}\right\|^4\right) < \infty
\]
for each $x \in \{ 0,1\}$. 
\end{partialIndepAssump}

Under these assumptions, we show the following result.

\begin{theorem}[ATE convergence]\label{thm:ATE convergence}
Suppose A\ref{assn:iid}--A\ref{assn:quant reg} hold. Then
\[
	\sqrt{n}
	\begin{pmatrix}
		\widehat{\overline{\text{ATE}}}^c - \overline{\text{ATE}}_\varepsilon^c \\
		\widehat{\underline{\text{ATE}}}^c - \underline{\text{ATE}}_\varepsilon^c
	\end{pmatrix}
	\xrightarrow{d} \mathbf{Z}_{\text{ATE}},
\]
where $\mathbf{Z}_\text{ATE}$ is a random vector in $\R^2$ whose distribution is characterized in the proof.
\end{theorem}

Like the CATE bound estimators, the limiting distribution of the ATE bound estimators is non-Gaussian. To understand this limiting distribution, first recall that we denote our bounds on the means $\Exp(Y_x)$ by
\[
	\overline{E}_{x,\varepsilon}^c = \overline{\Gamma}_3(x,\theta_0) = \Exp[ \overline{\Gamma}_2(x,W,\theta_0)]
	\qquad \text{and} \qquad	
	\underline{E}_{x,\varepsilon}^c = \underline{\Gamma}_3(x,\theta_0) = \Exp[ \underline{\Gamma}_2(x,W,\theta_0)].
\]
We estimate them by
\[
	\widehat{\overline{E}}_x^c = \avg \overline{\Gamma}_2(x,W_i,\widehat{\theta}) \quad \text{ and } \quad
	\widehat{\underline{E}}_x^c  = \avg \underline{\Gamma}_2(x,W_i,\widehat{\theta}).
\]
In the proof of theorem \ref{thm:ATE convergence}, we show that the following asymptotic expansion holds: 
\begin{align}\label{eq:twoPieceRepresentationForPotMeanEsts}
	\sqrt{n}\begin{pmatrix}
		\widehat{\overline{E}}_x^c - \overline{E}_{x,\varepsilon}^c \\
		\widehat{\underline{E}}_x^c - \underline{E}_{x,\varepsilon}^c
	\end{pmatrix}
	&= \Gamma_{3,\theta_0}'(x, \sqrt{n} (\widehat{\theta} - \theta_0)) + \frac{1}{\sqrt{n}} \sum_{i=1}^n \Big( \Gamma_2(x,W_i, \theta_0) - \Exp[ \Gamma_2(x,W,\theta_0)] \Big) + o_p(1),
\end{align}
where $\Gamma_{3,\theta_0}'$ is the Hadamard directional derivative of $\Gamma_3$, which we define in the proof of theorem \ref{thm:ATE convergence}. The first term in this expansion comes from the sample variation in the first step estimators: the propensity score $\widehat{p}_{x|w} = L(x,w'\widehat{\beta})$ and quantile function $\widehat{Q}(\tau \mid x,w) = p(x,w)'\widehat{\gamma}(\tau)$. The functional $\Gamma_{3,\theta_0}'(x,\cdot)$ is nonlinear in $\widehat{\beta}$. Therefore, since $\sqrt{n}(\widehat{\beta}-\beta_0)$ converges in distribution to a Gaussian limiting process, the limiting distribution of this functional is non-Gaussian. If $\beta_0$ was known---and hence the propensity score was known---then this component would follow a Gaussian distribution since the remaining component $\widehat{\gamma}$ is asymptotically Gaussian and enters $\Gamma_{3,\theta_0}'(x,\cdot)$ linearly.

The second term in this expansion comes from the variation of the CATE bounds over the values of the covariates $W$. It follows a limiting Gaussian distribution by the central limit theorem. This term is asymptotically independent of the sampling variation in the first step estimators $\widehat{\theta}$ since the influence function of $\widehat{\theta}$ is mean independent of $W$. Overall, we see that the limiting distribution of the ATE bounds is the sum of two independent random vectors, one Gaussian and one non-Gaussian. We approximate the distribution of these two random vectors using two separate bootstraps in section \ref{sec:bootstrap}.

\subsubsection*{The ATT Bounds}

Finally we study the limiting properties of our ATT bound estimators. Our trimmed population ATT bounds are
\[
	[\underline{\text{ATT}}_\varepsilon^c, \overline{\text{ATT}}_\varepsilon^c] =
		\left[\Exp(Y \mid X=1) - \frac{\overline{E}_{0,\varepsilon}^c - p_0 \Exp(Y \mid X=0)}{p_1}, \, \Exp(Y \mid X=1) - \frac{\underline{E}_{0,\varepsilon}^c - p_0\Exp(Y \mid X=0)}{p_1}\right].
\]
We estimate them by
\[
	[\widehat{\underline{\text{ATT}}}^c, \widehat{\overline{\text{ATT}}}^c] = \left[\widehat\Exp (Y \mid X=1) - \frac{\widehat{\overline{E}}_0^c - \widehat{p}_0 \widehat\Exp (Y \mid X=0)}{\widehat{p}_1}, \,
	\widehat\Exp (Y \mid X=1) - \frac{\widehat{\underline{E}}_0^c - \widehat{p}_0 \widehat\Exp (Y \mid X=0)}{\widehat{p}_1} \right], 
\]
where $\widehat{\Exp}(Y \mid X=0)$ and $\widehat{p}_0$ are defined in section \ref{sec:estimation}.

\begin{proposition}[ATT convergence]\label{prop:ATT convergence}
Suppose the assumptions of theorem \ref{thm:ATE convergence} hold. Suppose further that $\var(Y\indicator(X=x)) <\infty$ for each $x\in\{0,1\}$. Then
\[
	\sqrt{n}
	\begin{pmatrix}
		\widehat{\overline{\text{ATT}}}^c - \overline{\text{ATT}}_\varepsilon^c\\
		\widehat{\underline{\text{ATT}}}^c - \underline{\text{ATT}}_\varepsilon^c
	\end{pmatrix}
	\xrightarrow{d} \mathbf{Z}_{\text{ATT}},
\]
where $\mathbf{Z}_\text{ATT}$ is a random vector in $\R^2$ whose distribution is characterized in the proof.
\end{proposition}

$\widehat{\Exp}(Y \mid X=0)$ and $\widehat{p}_0$ are asymptotically Gaussian. Like our analysis of the ATE bounds, however, $\widehat{\overline{E}}_0^c$ and $\widehat{\underline{E}}_0^c$ have non-Gaussian asymptotic distributions. Overall, $\mathbf{Z}_\text{ATT}$, the asymptotic distribution of our ATT bound estimators, is a linear combination of Gaussian and non-Gaussian random variables.

\section{Bootstrap Inference}\label{sec:bootstrap}

We now show how to conduct inference on our bounds for CATE, ATE, and ATT. Earlier we noted that these bounds are generally not ordinary Hadamard differentiable mappings of the underlying parameters $\theta_0$. By corollary 3.1 in \cite{FangSantos2014}, this implies that standard bootstrap approaches cannot be used for these bounds. We instead use the non-standard bootstrap approach developed by \cite{FangSantos2014}. For brevity we focus on ATE and ATT in this section. We provide analogous results for CQTE and CATE in lemmas \ref{prop:CQTE analytical boot} and \ref{prop:CATE analytical boot} in appendix \ref{sec:CQTEandCATEanalyticalboot}. 

\subsection{Inference on Potential Outcome Means}

The bounds for ATE and ATT can be written in terms of bounds on $\Exp(Y_x)$. In this section we describe how to do inference on bounds for these means. We'll then use these results to do inference on our ATE and ATT bounds in the next subsection. Recall that our bounds on $\Exp(Y_x)$ can be written as a functional of $\theta$. This functional is Hadamard directionally differentiable in $\theta$, but it is generally not ordinary Hadamard differentiable. Theorem 3.1 of \cite{FangSantos2014} shows how to do bootstrap inference by consistently estimating the Hadamard directional derivative (HDD). This can be done by using analytical estimators or by using a numerical derivative as described in \cite{HongLi2015}. Here we use analytical estimates of the HDD. This approach explicitly uses the functional form of the HDD to estimate it. It allows us to avoid picking the numerical derivative step size, although other tuning parameters are used to estimate the HDDs analytically.

\subsubsection*{Setup}

Next we define some general notation. Let $Z_i = (Y_i,X_i,W_i)$ and $Z^n = \{Z_1,\ldots,Z_n\}$. Let $\vartheta_0$ denote some parameter of interest and let $\widehat{\vartheta}$ be an estimator of $\vartheta_0$ based on the data $Z^n$. Let $\mathbf{A}_n^*$ denote $\sqrt{n}(\widehat{\vartheta}^* - \widehat{\vartheta})$ where $\widehat{\vartheta}^*$ is a draw from the nonparametric bootstrap distribution of $\widehat{\vartheta}$. Suppose $\A$ is the tight limiting process of $\sqrt{n} ( \widehat{\vartheta} - \vartheta_0)$. Denote bootstrap consistency by $\mathbf{A}_n^* \overset{P}{\rightsquigarrow} \mathbf{A}$ where $\overset{P}{\rightsquigarrow}$ denotes weak convergence in probability, conditional on the data $Z^n$. Weak convergence in probability conditional on $Z^n$ is defined as
\[
	\sup_{h\in \text{BL}_1} \left| \Exp [h(\mathbf{A}_n^*) \mid Z^n] - \Exp [h(\mathbf{A})] \right| = o_p(1)
\]
where $\text{BL}_1$ denotes the set of Lipschitz functions into $\R$ with Lipschitz constant no greater than 1. We leave the domain of these functions and its associated norm implicit.

We focus on the choices $\vartheta_0 = \theta_0$ and $\widehat{\vartheta} = \widehat{\theta}$. For these choices, let $\mathbf{Z}_n^* = \sqrt{n}(\widehat{\theta}^* - \widehat{\theta})$. Let $\mathbf{Z}_1$ denote the limiting distribution of $\sqrt{n}(\widehat{\theta} - \theta_0)$; see lemma \ref{lemma:prelim estimators} in appendix \ref{sec:prelimest}. Theorem 3.6.1 of \cite{VaartWellner1996} implies that $\mathbf{Z}_n^* \overset{P}{\rightsquigarrow} \mathbf{Z}_1$. Our parameters of interest are all functionals $\Gamma$ of $\theta_0$. In particular, in section \ref{sec:asymptotics} we showed that
\[
	\sqrt{n}(\Gamma(\widehat{\theta}) - \Gamma(\theta_0)) \rightsquigarrow \Gamma'_{\theta_0}(\mathbf{Z}_1)
\]
for a variety of functionals $\Gamma$. To do inference on these functionals, we therefore want to estimate the distribution of $\Gamma_{\theta_0}'(\mathbf{Z}_1)$. \cite{FangSantos2014} show that
\[
	\widehat{\Gamma}_{\theta_0}'(\mathbf{Z}_n^*) \overset{P}{\rightsquigarrow} \Gamma'_{\theta_0}(\mathbf{Z}_1)
\]
where $\widehat{\Gamma}_{\theta_0}'$ is a suitable estimator of the Hadamard directional derivative $\Gamma'_{\theta_0}$. In this section we construct the estimators $\widehat{\Gamma}_{\theta_0}'$ and show that they can be used in this bootstrap.

\subsubsection*{Main Result}

Next, recall the asymptotic expansion in equation \eqref{eq:twoPieceRepresentationForPotMeanEsts} on page \pageref{eq:twoPieceRepresentationForPotMeanEsts}. As we will show, the second term in this expansion can be approximated using standard bootstrap approaches and replacing $\theta_0$ by $\widehat{\theta}$. The first term requires estimating the HDDs $\overline{\Gamma}_{3,\theta_0}'$ and $\underline{\Gamma}_{3,\theta_0}'$. The formulas for our estimators of these HDDs are long, and so we describe them in appendix \ref{sec:HDDformulas}. Denote these estimators by $\widehat{\overline{\Gamma}}_{3,\theta_0}'$ and $\widehat{\underline{\Gamma}}_{3,\theta_0}'$. They require choosing two scalar tuning parameters, $\kappa_n$ and $\eta_n$. $\kappa_n$ is a slackness parameter and $\eta_n$ is a step size parameter used to compute numerical derivatives of $\widehat{\gamma}(\cdot)$. Although not used in our proof, the asymptotic independence of the two components implies that approximating their respective marginal distributions is sufficient to obtain their joint distribution. 

As we just mentioned, we'll use the standard nonparametric bootstrap to approximate the second term of equation \eqref{eq:twoPieceRepresentationForPotMeanEsts}. To formalize this, let $\mathbb{G}_n^*$ denote the nonparametric bootstrap empirical process:
\[\mathbb{G}_n^* = \frac{1}{\sqrt{n}}\sum_{i=1}^n (M_{n,i} - 1) \delta_{Z_i}\]
where $(M_{n,1},\ldots,M_{n,n})$ are multinomially distributed with parameters $(1/n,\ldots,1/n)$ independently of $Z^n$, and where $\delta_{Z_i}$ is a distribution which assigns probability one to the value ${Z_i} \in \R^{2+d_W}$. Then for any function $g$,
\[
	\mathbb{G}_n^* g(Z) = \sqrt{n}\left(\frac{1}{n} \sum_{i=1}^n g(Z_{i}^*) - \overline{g(Z)}\right)
\]
where $Z_{i}^*$, $i=1,\ldots,n$ are drawn independently with replacement from $\{Z_1,\ldots,Z_n\}$ and $\overline{g(Z)} = \avg g(Z_i)$. In particular, we'll study the asymptotic distribution of
\[
	\mathbb{G}_n^* \Gamma_2(x,W,\widehat{\theta}) = \sqrt{n}\left(\frac{1}{n} \sum_{i=1}^n \Gamma_2(x,W_i^*, \widehat{\theta}) - \frac{1}{n} \sum_{i=1}^n \Gamma_2(x,W_i,\widehat{\theta}) \right).
\]

The following proposition is our main bootstrap consistency result.

\begin{proposition}[Analytical Bootstrap for Mean Potential Outcomes]\label{prop:ATE analytical boot}
Suppose the assumptions of theorem \ref{thm:ATE convergence} hold. Let $\kappa_n \to 0$, $n\kappa_n^2\to \infty$, $\eta_n \to 0,$ and $n\eta_n^2 \to \infty$ as $n\to\infty$. Then
\[
	\widehat{\Gamma}_{3,\theta_0}'(x, \sqrt{n}(\widehat{\theta}^* - \widehat{\theta})) + \mathbb{G}_n^* \Gamma_2(x,W,\widehat{\theta})
	\overset{P}{\rightsquigarrow}
	\Z_4(x),
\]
where $\Z_4(\cdot)$ is the limiting process of the expression given in equation \eqref{eq:twoPieceRepresentationForPotMeanEsts}, as characterized in the proof.
\end{proposition}

This result shows how to use the bootstrap to approximate the joint limiting distribution of upper and lower bounds of $\Exp(Y_x)$, $x\in\{0,1\}$. As we show in section \ref{sec:ATEboundsInference} below, these approximations can be used to conduct pointwise or uniform-in-$c$ inference on the ATE bounds. As part of the proof, we show that $\widehat{\Gamma}_{3,\theta_0}'(x, \sqrt{n}(\widehat{\theta}^* - \widehat{\theta}))$ weakly converges in probability conditional on the data to $\Gamma_{3,\theta_0}'(x,\Z_1)$, a non-Gaussian vector which reflects the sample variation in the first step estimators. We also show weak convergence in probability conditional on the data of $\mathbb{G}_n^* \Gamma_2(x,W,\widehat{\theta})$ to $\mathbb{G} \Gamma_2(x,W,\theta_0) \sim \normal(0, \var(\Gamma_2(x,W,\theta_0))$, a bivariate Gaussian vector which reflects the variation of the CATE bounds over $W$. This variation can be approximated using the standard nonparametric bootstrap. Hence the bounds' limiting distribution is approximated by a combination of standard and non-standard bootstraps. Note that the two bootstrap distributions can be computed from a unique sequence of draws $Z_i^*$ and so has the same computational burden as a single bootstrap.

\subsection{Inference on the ATE Bounds}\label{sec:ATEboundsInference}

Next we show how to use proposition \ref{prop:ATE analytical boot} to do inference on our ATE bounds $[\underline{\text{ATE}}_\varepsilon^c, \overline{\text{ATE}}_\varepsilon^c]$. We first consider inference pointwise in $c$. We then construct confidence bands that are uniform over $c$.

\subsubsection{Pointwise in $c$ Confidence Sets}\label{sec:pointwiseInCconfidencesetsATE}

An immediate corollary of proposition \ref{prop:ATE analytical boot} is
\begin{align}\label{eq:bootATEexpr}
	&\sqrt{n}\begin{pmatrix}
	\widehat{\overline{\text{ATE}}}^{c,*} - \widehat{\overline{\text{ATE}}}^c \\
	\widehat{\underline{\text{ATE}}}^{c,*} - \widehat{\underline{\text{ATE}}}^c 
	\end{pmatrix} \\
	&\qquad =
	\begin{pmatrix}
	\Big( \widehat{\overline{\Gamma}}_{3,\theta_0}'(1, \sqrt{n}(\widehat{\theta}^* - \widehat{\theta})) + \mathbb{G}_n^* \overline{\Gamma}_2(1,W,\widehat{\theta}) \Big)
	-
	\Big(
	\widehat{\underline{\Gamma}}_{3,\theta_0}'(0, \sqrt{n}(\widehat{\theta}^* - \widehat{\theta})) + \mathbb{G}_n^* \underline{\Gamma}_2(0,W,\widehat{\theta})
	\Big) \\[1em]
	\Big( \widehat{\underline{\Gamma}}_{3,\theta_0}'(1, \sqrt{n}(\widehat{\theta}^* - \widehat{\theta})) + \mathbb{G}_n^* \underline{\Gamma}_2(1,W,\widehat{\theta}) \Big)
	-
	\Big(
	\widehat{\overline{\Gamma}}_{3,\theta_0}'(0, \sqrt{n}(\widehat{\theta}^* - \widehat{\theta})) + \mathbb{G}_n^* \overline{\Gamma}_2(0,W,\widehat{\theta})
	\Big)
	\end{pmatrix}
	\notag \\
	&\qquad 
	\overset{P}{\rightsquigarrow} \textbf{Z}_\text{ATE}.\notag
\end{align}
Thus we can also use this specific bootstrap to approximate the asymptotic distribution of our ATE bounds estimators. Given this result, we can construct a $100(1-\alpha)$\%  confidence set for the ATE identified set under $c$-dependence as follows. Let
\[
	\text{CI}_\text{ATE}^c(1-\alpha)=
	\left[\widehat{\underline{\text{ATE}}}^c - \frac{\widehat{d}_\alpha}{\sqrt{n}}, 
	\,
	\widehat{\overline{\text{ATE}}}^c + \frac{\widehat{d}_\alpha}{\sqrt{n}}\right]
\]
where
\[
	\widehat{d}_\alpha = \inf\left\{z\in\R: \Prob \left(\sqrt{n} (\widehat{\underline{\text{ATE}}}^{c,*} - \widehat{\underline{\text{ATE}}}^c ) \leq -z
	\; \text{ and } \;
	\sqrt{n} (\widehat{\overline{\text{ATE}}}^{c,*} - \widehat{\overline{\text{ATE}}}^c ) \geq z \mid Z^n \right) \geq 1-\alpha \right\}.
\]
The probability in this expression can be approximated by taking a large number of bootstrap draws according to equation \eqref{eq:bootATEexpr}. Proposition \ref{prop:ATE analytical boot} then implies that
\[
	\liminf_{n \rightarrow \infty} \Prob \Big( \text{CI}_\text{ATE}^c(1-\alpha) \supseteq [\underline{\text{ATE}}_\varepsilon^c,\overline{\text{ATE}}_\varepsilon^c] \Big) \geq 1-\alpha.
\]

Let  
\[d_\alpha = \inf\left\{z\in\R: \Prob \left(\textbf{Z}_\text{ATE}^{(2)} \leq -z	\; \text{ and } \; \textbf{Z}_\text{ATE}^{(1)} \geq z\right) \geq 1-\alpha\right\}.\]
If $\Prob (\textbf{Z}_\text{ATE}^{(2)} \leq -z \; \text{ and } \; \textbf{Z}_\text{ATE}^{(1)} \geq z)$ is continuous and strictly increasing in a neighborhood of $d_\alpha$, corollary 3.2 in \cite{FangSantos2015workingPaper} yields $\widehat{d}_\alpha = d_\alpha + o_p(1)$ and hence
 \[
	\lim_{n \rightarrow \infty} \Prob \Big( \text{CI}_\text{ATE}^c(1-\alpha) \supseteq [\underline{\text{ATE}}_\varepsilon^c,\overline{\text{ATE}}_\varepsilon^c] \Big) = 1-\alpha.
\]

\subsubsection{Uniform over $c$ ATE Bands}\label{sec:uniformATEbands}

We just described how to use proposition \ref{prop:ATE analytical boot} to do inference on the ATE bounds for any fixed $c$. Those results can be immediately extended to do inference the ATE bounds for any finite grid of $c$'s. In this section we show how to construct confidence bands that are uniform over all $c \in [0,1]$. We do this by using monotonicity of the ATE bound functions in $c$. This lets us extrapolate bands that are uniform on a finite grid in such a way that they have uniform coverage. A related procedure is described in corollary 1 of \cite{MastenPoirier2020}.

Although $\overline{\text{ATE}}_\varepsilon^c$ is nondecreasing in $c$, its estimate $\widehat{\overline{\text{ATE}}}^c$ may be nonmonotonic in $c$ because of the quantile crossing problem with linear quantile regression. In that case, we could monotonize the estimated ATE bound function by using the rearrangement procedure of \cite{ChernozhukovFernandez-ValGalichon2010}, for example. As they show, the rearrangement operator is Hadamard directionally differentiable, and thus can be accommodated in our inferential results. Likewise, $\underline{\text{ATE}}_\varepsilon^c$ is nonincreasing in $c$ and $\widehat{\underline{\text{ATE}}}^c$ is also nonincreasing after applying a suitable rearrangement. From here on we assume our bound estimators have been monotonized. Note that this does not affect the asymptotic distribution of the estimators, by corollary 1 of \cite{ChernozhukovFernandez-ValGalichon2010}.

Next consider a grid of values $\mathcal{C} = \{c_1,\ldots,c_K\}$ such that $0 = c_1 <\cdots < c_K = 1$. All of our analysis works with $0 < c_1$ and $c_K < 1$, but typically researchers will want to include the endpoints of $[0,1]$ so we do that from here on. Using methods similar to those in section \ref{sec:pointwiseInCconfidencesetsATE}, for all $c \in \mathcal{C}$ let
\[
	\text{CI}_\text{ATE}^c(1-\alpha) = \left[\widehat{\underline{\text{ATE}}}^c - \frac{\widehat{d}_\alpha(c)}{\sqrt{n}},\widehat{\overline{\text{ATE}}}^c + \frac{\widehat{d}_\alpha(c)}{\sqrt{n}}\right]
\]
be $100(1-\alpha)$\% confidence sets for $[\underline{\text{ATE}}_\varepsilon^c,\overline{\text{ATE}}_\varepsilon^c]$ where the critical values $\widehat{d}_\alpha(c)$ are chosen such that these sets are uniform over $c$ in the finite grid $\mathcal{C}$. That is, 
\[
	\Prob\left( \text{CI}_\text{ATE}^c(1-\alpha) \supseteq [\underline{\text{ATE}}_\varepsilon^c,\overline{\text{ATE}}_\varepsilon^c] \text{ for all $c \in \mathcal{C}$} \right) \rightarrow 1-\alpha
\]
as $n \rightarrow \infty$.
Finally, let $\min(c) = \inf \{ c_k \in \mathcal{C} : c \leq c_k \}$ denote the smallest element in the grid $\mathcal{C}$ that is still larger than $c$. For $c \in [0,1]$ define
\[
	\widehat{\text{UB}}(c) = \widehat{\overline{\text{ATE}}}^{\min(c)} + \frac{\widehat{d}_\alpha(\min(c))}{\sqrt{n}}
	\qquad \text{and} \qquad
	\widehat{\text{LB}}(c) = \widehat{\underline{\text{ATE}}}^{\min(c)} - \frac{\widehat{d}_\alpha(\min(c))}{\sqrt{n}}.
\]
$\widehat{\text{UB}}(c)$ is the greatest monotonic interpolation of the upper bounds of the confidence intervals on the grid $\mathcal{C}$. $\widehat{\text{LB}}(c)$ is the least monotonic interpolation of the lower bounds of the confidence intervals on the grid $\mathcal{C}$. By the definition of these interpolated bands and by monotonicity of the population ATE bounds,
\begin{align*}
	&\Prob \left( [\widehat{\text{LB}}(c), \widehat{\text{UB}}(c)] \supseteq [\underline{\text{ATE}}_\varepsilon^c,\overline{\text{ATE}}_\varepsilon^c] \text{ for all $c \in [0,1]$} \right) \\
	&= \Prob\left( [\widehat{\text{LB}}(c), \widehat{\text{UB}}(c)] \supseteq [\underline{\text{ATE}}_\varepsilon^c,\overline{\text{ATE}}_\varepsilon^c] \text{ for all $c \in \mathcal{C}$} \right) \\
	&\to 1-\alpha
\end{align*}
as $n \rightarrow \infty$.

In this subsection we've shown that, although we cannot obtain the limiting distribution of the ATE bounds uniformly over $c \in [0,1]$, the fact that these bounds are monotonic lets us nonetheless do inference on them uniformly over $[0,1]$. This monotonicity comes from the nested nature of $c$-dependence: $c_1$-dependence implies $c_2$-dependence when $c_1 \leq c_2$. This kind of monotonicity is common in many other approaches to sensitivity analysis and hence greatest and least monotonic interpolations can likely be used more broadly to construct uniform confidence bands.

\subsection{Inference on the ATT bounds}

Bootstrap inference on the ATT bounds is quite similar to that on the ATE bounds. By examining the ATT bounds' limiting distribution (see the proof of proposition \ref{prop:ATT convergence}) we see that it depends on two types of terms:
\begin{enumerate}
\item One term comes from the limiting distribution of 
\[
	\sqrt{n}
	\begin{pmatrix}
	\widehat{\overline{E}}_0^c - \overline{E}_{0,\varepsilon}^c \\
	\widehat{\underline{E}}_0^c - \underline{E}_{0,\varepsilon}^c
\end{pmatrix},
\]
which is non-standard. We'll approximate this term distribution by using the non-standard bootstrap of proposition \ref{prop:ATE analytical boot}. 

\item The other terms are due to the limiting distributions of
\[
	\sqrt{n}
	\begin{pmatrix}
		\widehat{\Exp}(Y \mid X=x) - \Exp(Y \mid X=x) \\
		\widehat{p}_x - p_x
	\end{pmatrix},
\]
which are standard and Gaussian. The distribution of these terms can be approximated by the nonparametric bootstrap. For example, standard arguments show that the limiting distribution of $\sqrt{n}(\widehat{\Exp}(Y \mid X=x) - \Exp(Y \mid X=x))$ is approximated by 
\[
	\mathbb{Z}_{\Exp(Y \mid X=x)}^* \equiv \frac{1}{\widehat{p}_x} \left(\mathbb{G}_n^* Y \indicator(X=x) - \widehat{\Exp}(Y \mid X=x) \cdot  \mathbb{G}_n^* \indicator(X=x) \right).
\]
Similarly, $\mathbb{Z}_{p_x}^* \equiv \mathbb{G}_n^* \indicator(X=x)$ converges weakly in probability conditional on $Z^n$ to the limiting distribution of $\sqrt{n}(\widehat{p}_x - p_x)$. 
\end{enumerate}
Combining all the terms gives
\[
	\begin{pmatrix}	
	\begin{array}{l}
	\displaystyle
		\mathbb{Z}_{\Exp(Y \mid X=1)}^* - \frac{\widehat{\underline{\Gamma}}_{3,\theta_0}'(0,\sqrt{n}(\widehat{\theta}^* - \widehat{\theta})) + \mathbb{G}_n^*\underline{\Gamma}_{2}(0,W,\widehat{\theta})}{\widehat{p}_1} \\[1em]
	\displaystyle	
		\qquad + \frac{\widehat{p}_0}{\widehat{p}_1}\mathbb{Z}^*_{\Exp(Y \mid X=0)} + \frac{\widehat{\Exp}(Y \mid X=0)}{\widehat{p}_1}\mathbb{Z}^*_{p_0} + \frac{\widehat{\underline{E}}_0^c - \widehat{\Exp}(Y \mid X=0)\widehat{p}_0}{\widehat{p}_1^2}\mathbb{Z}^*_{p_1}\\[3em]
	\displaystyle
		\mathbb{Z}_{\Exp(Y \mid X=1)}^* - \frac{\widehat{\overline{\Gamma}}_{3,\theta_0}'(0,\sqrt{n}(\widehat{\theta}^* - \widehat{\theta})) + \mathbb{G}_n^*\overline{\Gamma}_{2}(0,W,\widehat{\theta})}{\widehat{p}_1} \\[1em]
	\displaystyle
		\qquad + \frac{\widehat{p}_0}{\widehat{p}_1}\mathbb{Z}^*_{\Exp(Y \mid X=0)} + \frac{\widehat{\Exp}(Y \mid X=0)}{\widehat{p}_1}\mathbb{Z}^*_{p_0} + \frac{\widehat{\overline{E}}_0^c - \widehat{\Exp}(Y \mid X=0)\widehat{p}_0}{\widehat{p}_1^2}\mathbb{Z}^*_{p_1}
	\end{array}
	\end{pmatrix} 
	\overset{P}{\rightsquigarrow} \textbf{Z}_\text{ATT}.
\]
We can use this result to construct pointwise confidence sets for the ATT bounds for a fixed $c$, or to construct confidence bands that are uniform on a finite grid $\mathcal{C}$. Like the ATE bounds, the ATT bounds are monotonic in $c$. Thus a similar interpolation can be used to construct confidence bands for the ATT bounds that are uniform over $c \in [0,1]$.

\subsection{Inference on Breakdown Points}\label{sec:breakdownPointInference}

We conclude this section be showing how to use the confidence bands we just described to do inference on breakdown points. For brevity we focus on the breakdown point for the conclusion that the ATE is nonnegative, which we defined earlier in equation \eqref{eq:ATEbreakdownPoint}. Inference on other breakdown points for other conclusions can be done similarly.

Let $\text{CI}_\text{ATE}^c(1-\alpha)$ be a pointwise-in-$c$ confidence band for the ATE bounds, as described in section \ref{sec:pointwiseInCconfidencesetsATE}. Define
\[
	c_L = \sup \{ c \in [0,1] : \text{CI}_\text{ATE}^c(1-\alpha) \subseteq [0,\infty) \}.
\]
This is simply the value at which the confidence band first intersects the horizontal line at zero. By proposition S.2 in Appendix D of \cite{MastenPoirier2020},
\[
	\lim_{n \rightarrow \infty} \Prob( c_L \leq c_\textsc{bp} ) \geq 1 - \alpha.
\]
Thus $[c_L, 1]$ is a valid one-sided lower confidence interval for the breakdown point $c_\textsc{bp}$.

\section{Sufficient Conditions for Standard Inference}\label{sec:standardboot}

In the previous section we showed how to use a non-standard bootstrap method to conduct inference on the CATE, ATE, and ATT bounds. The key technical problem was that these bounds are not necessarily Hadamard differentiable functionals of the first step estimators; they are only Hadamard directionally differentiable. In this section, we provide simple sufficient conditions on the propensity score under which the CATE, ATE, and ATT bounds are in fact Hadamard differentiable. Under this condition, the methods in section \ref{sec:bootstrap} are still valid, but so is the standard nonparametric bootstrap. After stating the formal result, we discuss when this sufficient condition holds and when it does not.

First consider the average treatment effect. Recall that the ATE bounds depend on the functional $\Gamma_3(x,\theta)$. We will show that its Hadamard directional derivative $\Gamma_{3,\theta_0}'(x, h)$ is linear in $h$ under a condition on the value of $c$, the propensity score $p_{1 | w}$, and the distribution of $W$. By proposition 2.1 in \cite{FangSantos2014}, this linearity is equivalent to Hadamard differentiability. By theorem 3.9.11 in \cite{VaartWellner1996}, this linearity also implies that the bootstrap process $\sqrt{n}(\Gamma_{3,\theta_0}(x, \widehat{\theta}^*) - \Gamma_{3,\theta_0}(x, \widehat{\theta}))$ converges weakly in probability conditional on the data to $\Gamma_{3,\theta_0}'(x, \mathbf{Z}_1)$, a Gaussian vector. In other words, we can conduct inference on the ATE bounds using the standard nonparametric bootstrap: Take $n$ independent draws from the data with replacement, compute the bound estimates in this bootstrap sample, and then use the distribution of these bound estimates across many such bootstrap samples to approximate the sampling distribution of the bound estimators. The following theorem provides the explicit sufficient condition for validity of this bootstrap. In this result, let $p_{1 \mid W} = \Prob(X=1 \mid W)$ denote the random variable obtained by evaluating the propensity score at the random vector $W$.

\begin{theorem}\label{thm:standardbootATE}
Suppose the assumptions of theorem 1 hold. Suppose $\Prob(p_{1|W} \in \{c,1-c\}) = 0$. Then
\[
	\sqrt{n}
	\begin{pmatrix}
		\widehat{\overline{\text{ATE}}}^{c \; \ast} - \widehat{\overline{\text{ATE}}}^c \\
		\widehat{\underline{\text{ATE}}}^{c \; \ast} - \widehat{\underline{\text{ATE}}}^c
	\end{pmatrix}
	\overset{P}{\rightsquigarrow} \mathbf{Z}_{\text{ATE}},
\]
where $(\widehat{\underline{\text{ATE}}}^{c \; \ast}, \widehat{\overline{\text{ATE}}}^{c \; \ast})$ are drawn from the nonparametric bootstrap distribution of $(\widehat{\underline{\text{ATE}}}^{c}, \widehat{\overline{\text{ATE}}}^{c})$.
\end{theorem}

In the proof of this result we show that when the propensity score does not contain a point mass on either $c$ or $1-c$, the mapping $\Gamma_3(x,\theta_0)$ is Hadamard differentiable for $x \in \{0,1\}$. Hence the nonparametric bootstrap is valid. Although it is not formally stated in the theorem, we conjecture that our sufficient condition for validity of the standard bootstrap is also a necessary condition. That is, we expect the standard bootstrap to be invalid when $c$ or $1-c$ are point masses of the propensity score's distribution. From the proof of theorem \ref{thm:standardbootATE}, we see when this condition on the propensity score fails, $\Gamma'_{1,\theta_0}(x,w,\tau,h)$ is nonlinear in $h$ on a set of $(\tau,w)$ values of positive measure. Since the HDD of $\Gamma_3(x,\cdot)$ is the integral over $(\tau,w)$ of the HDD of $\Gamma_1(x,w,\tau,\cdot)$, we expect that $\Gamma_{3,\theta_0}'(x,h)$ will also be nonlinear in $h$, a failure of Hadamard differentiability.

By further examining the proof of theorem \ref{thm:standardbootATE}, we can also show that the CATE bounds are Hadamard differentiable at covariate values $w$ and sensitivity parameter values $c$ such that $p_{1|w} \notin \{c,1-c\}$.
Finally, the following proposition gives a similar result for the ATT, using slightly weaker assumptions.

\begin{proposition}\label{prop:standardbootATT}
Suppose the assumptions of theorem 1 hold. Suppose $\var(Y\ind(X=x)) <\infty$ for each $x\in\{0,1\}$. Suppose $\Prob(p_{1|W} = c) = 0$. Then,
\[
	\sqrt{n}
	\begin{pmatrix}
		\widehat{\overline{\text{ATT}}}^{c \; \ast} - \widehat{\overline{\text{ATT}}}^c \\
		\widehat{\underline{\text{ATT}}}^{c \; \ast} - \widehat{\underline{\text{ATT}}}^c
	\end{pmatrix}
	\overset{P}{\rightsquigarrow} \mathbf{Z}_{\text{ATT}},
\]
where $(\widehat{\underline{\text{ATT}}}^{c \; \ast}, \widehat{\overline{\text{ATT}}}^{c \; \ast})$ are drawn from the nonparametric bootstrap distribution of $(\widehat{\underline{\text{ATT}}}^{c}, \widehat{\overline{\text{ATT}}}^{c})$.
\end{proposition}

The ATT bounds only depend on our bounds for $\Exp(Y_0)$, and not our bounds for $\Exp(Y_1)$. Hence we only need to examine $\Gamma_3(x,\theta_0)$ for $x=0$. So the proof of this proposition proceeds by showing that $\Gamma_3(0,\theta_0)$ is Hadamard differentiable when the propensity score does not have a point mass at $c$.

The sufficient conditions in theorem \ref{thm:standardbootATE} and proposition \ref{prop:standardbootATT} depend on the support of the propensity score $p_{1 \mid W}$. If the propensity score's distribution is absolutely continuous, it contains no point masses and therefore these condition holds. This holds when one covariate $W_k$ has nonzero coefficient $\beta_{0,k}$ and has a continuous distribution conditional on the other covariates $W_{-k}$. However, if all covariates are discrete or mixed, the support of the propensity score will generally contain point masses.  The nonparametric bootstrap may not be valid whenever $c$ coincides with these points. Even when $p_{1|W}$ has point masses, the nonparametric bootstrap is valid for $c$ outside of these points. To use this bootstrap, one could in principle estimate the support of $p_{1|W}$ to determine at which values of $c$ inference might be invalid, and select sensitivity parameters outside of this support.

The nonparametric bootstrap has the advantage of being computationally simple and does not require the choice of tuning parameters. While more involved, the bootstrap technique detailed in section \ref{sec:bootstrap} is valid regardless of the support of the propensity score. 
For example, in our empirical analysis in section \ref{sec:empirical}, all of the covariates are either mixed or discrete. Given our analysis above, we therefore use the non-standard bootstrap in our empirical analysis since the standard bootstrap may fail at some values of $c$. 

\section{Empirical Illustration}\label{sec:empirical}

In this section we illustrate our methods using data on the National Supported Work (NSW) demonstration project studied by \cite{LaLonde1986}. 
Since this is a highly studied and well-known program, we only briefly summarize it here. 
See, for example, \cite{HLS1999} for further details. We use LaLonde's data as reconstructed by \cite{DehejiaWahba1999}.

The NSW demonstration project randomly assigned participants to either receive a guaranteed job for 9 to 18 months along with frequent counselor meetings or to be left in the labor market by themselves. 
We use the \cite{DehejiaWahba1999} sample, which are all males in LaLonde's NSW dataset and where earnings are observed in 1974, 1975, and 1978.
This dataset has 445 people: 185 in the treatment group and 260 in the control group.
Like \cite{Imbens2003}, we use this experimental sample primarily as an illustration; in experiments where treatment was truly randomized it is not necessary to assess sensitivity to unconfoundedness. Our results may be useful for assessing the impact of randomization failure in experiments, but that is not our focus here.

\begin{table}[tb]
\caption[]{Summary statistics.
\label{summaryStats}}

\begin{adjustwidth}{-0.25in}{-0.25in}

\setlength{\linewidth}{.1cm}
\newcommand{\contents}{
\centering
\begin{tabular}{lccc} \hline \hline
 & \multicolumn{2}{c}{Experimental dataset} & Observational dataset \\
 & Control  & Treatment  & Control  \\  \hline 
Married &      0.15 &      0.19 &      0.78 \\  
&      (0.36) &      (0.39) &      (0.42) \\  
[0.5em]
Age &     25.05 &     25.82 &     38.61 \\  
&      (7.06) &      (7.16) &     (11.45) \\  
[0.5em]
Black &      0.83 &      0.84 &      0.27 \\  
&      (0.38) &      (0.36) &      (0.44) \\  
[0.5em]
Hispanic &      0.11 &      0.06 &      0.04 \\  
&      (0.31) &      (0.24) &      (0.20) \\  
[0.5em]
Education &     10.09 &     10.35 &     11.37 \\  
&      (1.61) &      (2.01) &      (3.40) \\  
[0.5em]
Earnings in 1974 &   2107.03 &   2095.57 &    765.75 \\  
&   (5687.91) &   (4886.62) &   (1399.79) \\  
[0.5em]
Earnings in 1975 &   1266.91 &   1532.06 &    650.54 \\  
&   (3102.98) &   (3219.25) &   (1332.89) \\  
[0.5em]
Positive earnings in 1974 &      0.25 &      0.29 &      0.29 \\  
&      (0.43) &      (0.46) &      (0.46) \\  
[0.5em]
Positive earnings in 1975 &      0.32 &      0.40 &      0.25 \\  
&      (0.47) &      (0.49) &      (0.43) \\  
[0.5em]
Sample size & 260 & 185 & 242 \\
\hline \hline 
\multicolumn{4}{p{\linewidth}}{\scriptsize Variable mean is shown in each cell, with that variable's standard deviation in parentheses.}\\
\end{tabular}
}

\setbox0=\hbox{\contents}
\setlength{\linewidth}{\wd0-2\tabcolsep-.25em}
\contents

\end{adjustwidth}
\end{table}

In addition to this experimental sample, we construct a sample using observational data.
This sample combines the 185 people in the NSW treatment group with 2490 people in a control group constructed from the Panel Study of Income Dynamics (PSID). 
This control group, called PSID-1 by LaLonde, consists of all male household heads observed in all years between 1975 and 1978 who were less than 55 years old and who did not classify themselves as retired.
We further drop observations with earnings above \$5,000 in 1974, 1975, or both. This leaves 148 treated units (out of 185) and 242 untreated units (out of 2490).
This observational sample was also considered by \cite{Imbens2003}.

The outcome of interest is earnings in 1978. There are also nine covariates: Earnings in 1974, earnings in 1975, years of education, age, indicators for race (Black, Hispanic, other), an indicator for marriage, an indicator for having a high school degree, and an indicator for treatment. All earnings variables are measured in 1982 dollars. Table \ref{summaryStats} shows the summary statistics, as reported in table 1 of \cite{Imbens2003}.

\begin{table}[H]
\caption[]{\label{baselineEsts} Baseline treatment effect estimates (in 1982 dollars).}

\begin{adjustwidth}{-0.25in}{-0.25in}

\setlength{\linewidth}{.1cm}
\newcommand{\contents}{
\centering
\begin{tabular}{lccc} \hline \hline
 & ATE  & ATT & Sample size \\  \hline 
Experimental dataset &      1633 &      1738 & 445\\  
 &       (650) &       (689) \\
[0.5em]
Observational dataset &      3337 &      4001 & 390\\  
 &       (769) &       (762) \\
\hline \hline
\multicolumn{4}{p{\linewidth}}{\scriptsize Standard errors in parentheses.}\\
\end{tabular}
}

\setbox0=\hbox{\contents}
\setlength{\linewidth}{\wd0-2\tabcolsep-.25em}
\contents

\end{adjustwidth}
\end{table}

\subsection*{Baseline Estimates}

Table \ref{baselineEsts} shows the baseline point estimates of both ATE and ATT under the unconfoundedness assumption in the two samples we consider. 
These estimates are all computed by inverse probability weighting (IPW) using a parametric logit propensity score estimator. 
We do not consider other estimators, since our goal is to illustrate sensitivity to identifying assumptions, rather than finite sample sensitivity to the choice of estimator.

\begin{figure}[H]
\centering
\includegraphics[width=100mm]{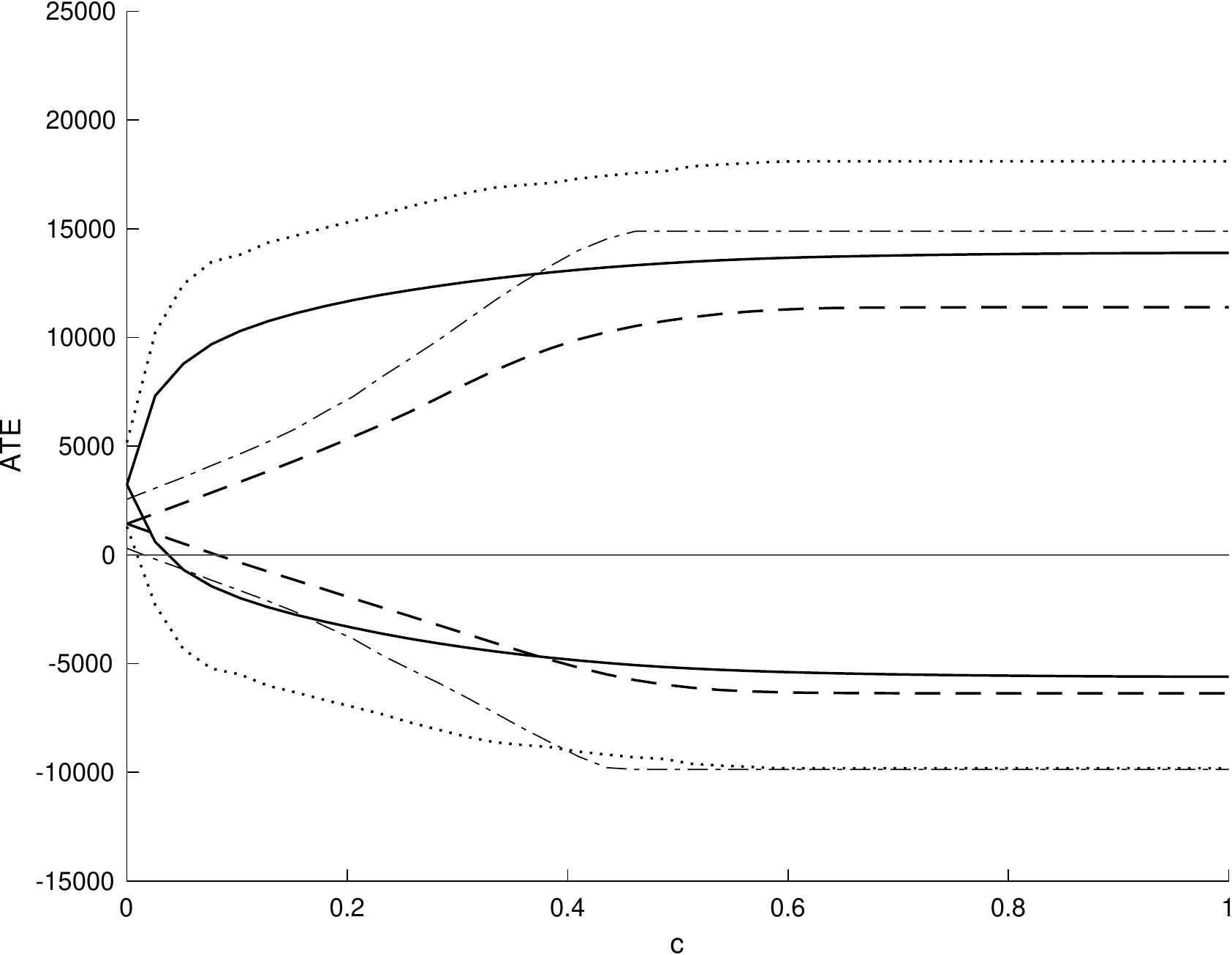} \\
\vspace{5mm}
\includegraphics[width=100mm]{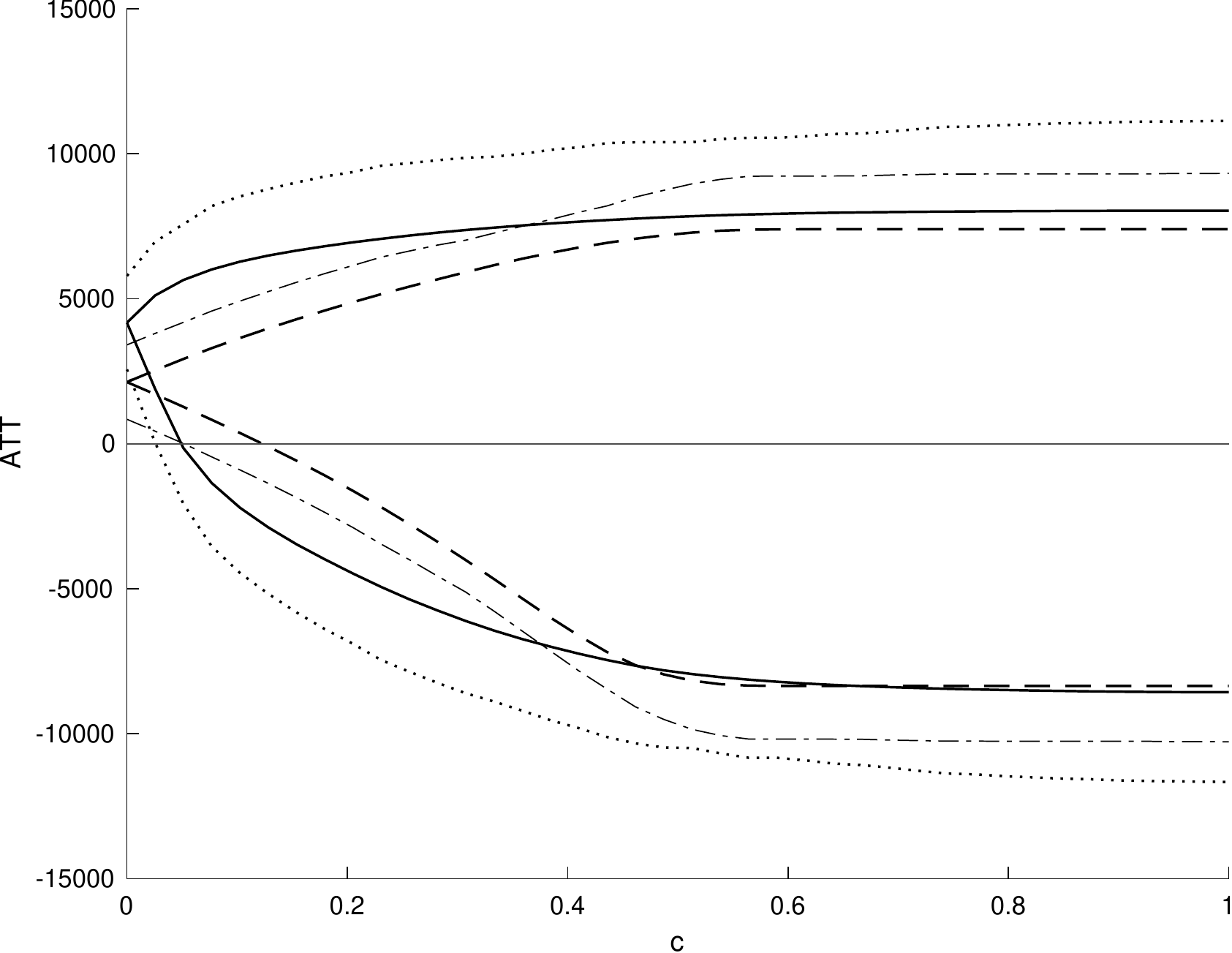}
\caption{Sensitivity of ATE (top) and ATT (bottom) estimates to relaxations of the selection on observables assumption. The solid lines are bounds computed using the observational dataset while the dashed lines are bounds computed using the experimental dataset. The light dotted lines are confidence bands for the observational dataset while the light dashed-dotted lines are confidence bands for the experimental dataset.}
\label{LaLondePlots}
\end{figure}

\subsection*{Relaxing Unconfoundedness}

Figure \ref{LaLondePlots} shows our main results. 
These are estimated treatment effect bounds under $c$-dependence, along with corresponding pointwise confidence bands, as described in sections \ref{sec:model}--\ref{sec:bootstrap}. 
The top plot shows bounds on ATE while the bottom plot shows bounds on ATT. 
The solid lines are bounds for the observational dataset while the dashed lines are bounds for the experimental dataset.
The light dotted lines are confidence bands for the observational dataset while the light dashed-dotted lines are confidence bands for the experimental dataset. 
These bands are constructed to have nominal 95\% coverage probability pointwise in $c$ based on our non-standard bootstrap results in section \ref{sec:bootstrap}.
For the tuning parameters we use $\varepsilon = 0.05$, $\eta_n = 0.05 n^{-1/4}$, and $\kappa_n = n^{-1/3}$.
Note that the sufficient conditions for validity of the standard bootstrap that we gave in section \ref{sec:standardboot} do not apply here, since the distribution of the propensity score variable $p_{1 \mid W}$ has point masses. 
This occurs because seven of the nine covariates are discrete, while the other two mixed discrete-continuous. 
The mixed variables are earnings in 1974 and earnings in 1975, which have point masses at zero since many people in the sample did not work in those years.

For both datasets, at $c=0$ the bounds collapse to the baseline point estimate. 
When $c>0$, we allow for some selection on unobservables. 
Comparing the shape of the bounds for both datasets we see that the experimental data are substantially more robust to relaxations from the baseline assumptions than the observational data. 
Specifically, for most values of $c$ the bounds for the experimental data are substantially tighter than the bounds for the observational data. 
Even the no assumptions bounds ($c = 1$) are tighter for the experimental data than for the observational data.

A second way to measure robustness uses \emph{breakdown points}. \cite{MastenPoirier2020} discuss these in detail and give additional references. 
In the current context, the breakdown point is simply the largest value of $c$ such that we can no longer draw a specific conclusion about some parameter. 
Specifically, in the next two subsections we consider two conclusions: The conclusion that ATE is nonnegative, and the conclusion that ATT is less than the per participant program cost.

\subsection*{Breakdown Points for Nonnegative ATE}

\newcommand*{\ATEexpBP}{0.082} %
\newcommand*{\ATEobsBP}{0.037} %
\newcommand*{\ATTexpBP}{0.123} %
\newcommand*{\ATTobsBP}{0.049} %

\newcommand*{\ATEexpBPconf}{0.0156} 
\newcommand*{\ATEobsBPconf}{0.009} 
\newcommand*{\ATTexpBPconf}{0.052} 
\newcommand*{\ATTobsBPconf}{0.026} 

First consider the conclusion that ATE is nonnegative. 
Our point estimates support this conclusion, but does it still hold if the baseline unconfoundedness assumption fails? 
In the experimental dataset, the estimated breakdown point is $\ATEexpBP$. 
This is simply the value of $c$ such that the lower bound function in figure \ref{LaLondePlots} intersects the horizontal axis. 
For all $c \leq \ATEexpBP$, the estimated identified sets for ATE only contains nonnegative values.
For $c > \ATEexpBP$, the estimated identified sets contain both positive and negative values. 
Hence, for such relaxations of unconfoundedness, we cannot be sure that the average treatment effect is positive.

For the observational dataset, the estimated breakdown point for the conclusion that ATE is nonnegative is $\ATEobsBP$. 
This is more than twice as small as the breakdown point for the experimental dataset. 
Hence again we see that conclusions about ATE from the experimental dataset are substantially more robust than the observational dataset. 
The same conclusion holds for ATT: The point estimates in both datasets suggest that it is positive. 
But how robust is that conclusion? 
The estimated breakdown point for the conclusion that ATT is nonnegative in the experimental data is $\ATTexpBP$ while it is $\ATTobsBP$ for the observational dataset. 
By this measure, the conclusion that ATT is positive is more than twice as robust using the experimental data compared to the observational data.

Thus far we have compared the robustness of results obtained from the experimental data with results obtained from the observational data. 
Next we discuss whether either of these results are robust in an absolute sense. 
To do this, we use the leave-out-variable $k$ propensity score analysis discussed in section \ref{sec:model}.

\begin{table}[htbp]
\caption{\label{tableCkInfo_d1} Variation in leave-out-variable-$k$ propensity scores, experimental data.}\centering
\begin{tabular}{lcccc} \hline \hline
 & p50  & p75  & p90  & $\bar{c}_k$  \\  \hline 
Earnings in 1975 &     0.001 &     0.004 &     0.008 &     0.053 \\  
[0.2em]
Black &     0.007 &     0.009 &     0.014 &     0.082 \\  
[0.2em]
Positive earnings in 1974 &     0.002 &     0.010 &     0.018 &     0.034 \\  
[0.2em]
Education &     0.012 &     0.022 &     0.031 &     0.087 \\  
[0.2em]
Married &     0.006 &     0.012 &     0.032 &     0.042 \\  
[0.2em]
Age &     0.015 &     0.024 &     0.034 &     0.099 \\  
[0.2em]
Earnings in 1974 &     0.002 &     0.011 &     0.035 &     0.209 \\  
[0.2em]
Positive earnings in 1975 &     0.013 &     0.017 &     0.062 &     0.082 \\  
[0.2em]
Hispanic &     0.007 &     0.017 &     0.099 &     0.124 \\  
\hline \hline \end{tabular}
\end{table}

First consider table \ref{tableCkInfo_d1}, which uses data from the experimental sample. 
For each variable $k$, listed in the rows of this table, we compute four summary statistics from the estimated distribution of
\[
	\Delta_k = | p_{1 \mid W}(W_{-k},W_k) - p_{1 \mid W_{-k}}(W_{-k}) |.
\]
Specifically, we estimate the 50th, 75th, and 90th percentiles of $\Delta_k$, along with the maximum observed value, denoted $\bar{c}_k$. 
As discussed in section \ref{sec:model}, these quantities tell us about the marginal impact of covariate $k$ on treatment assignment. 
$c$-dependence constrains the maximum value of the marginal impact of the \emph{unobserved} potential outcome on treatment assignment, above and beyond the \emph{observed} covariates. 
Thus the values in table \ref{tableCkInfo_d1} can help us calibrate $c$. Specifically, we will compare the breakdown point to the values in this table. These values could be interpreted as upper bounds on the magnitude of selection on unobservables that we might think is present. Thus, for a given reference value from this table, if the breakdown point is larger than the reference value, we could consider the conclusion of interest to be robust to failure of unconfoundedness. In contrast, if the breakdown point is smaller than the reference value, we could consider the conclusion of interest to be sensitive to failure of unconfoundedness.

Recall that the estimated breakdown point for the conclusion that ATE is nonnegative is $\ATEexpBP$. 
This is larger than three of the $\bar{c}_k$ values and on the same order of magnitude as four more. 
If we look at a less stringent comparison, the 90th percentile, we see that the estimated breakdown point is now larger than all but one of the rows, corresponding to the indicator for Hispanic. Let's examine this variable more closely. Figure \ref{propScoreVar_d1} plots the density $\Delta_k$ for $k=$ Hispanic indicator. 
Here we see that there is a small proportion of mass who have values larger than $\ATEexpBP$, but most people have values well below the breakdown point. Next suppose we weaken the criterion even more by considering the 75th percentile column in table \ref{tableCkInfo_d1}. The breakdown point is larger than all values in this column.

\begin{figure}[ht]
\centering
\includegraphics[width=100mm]{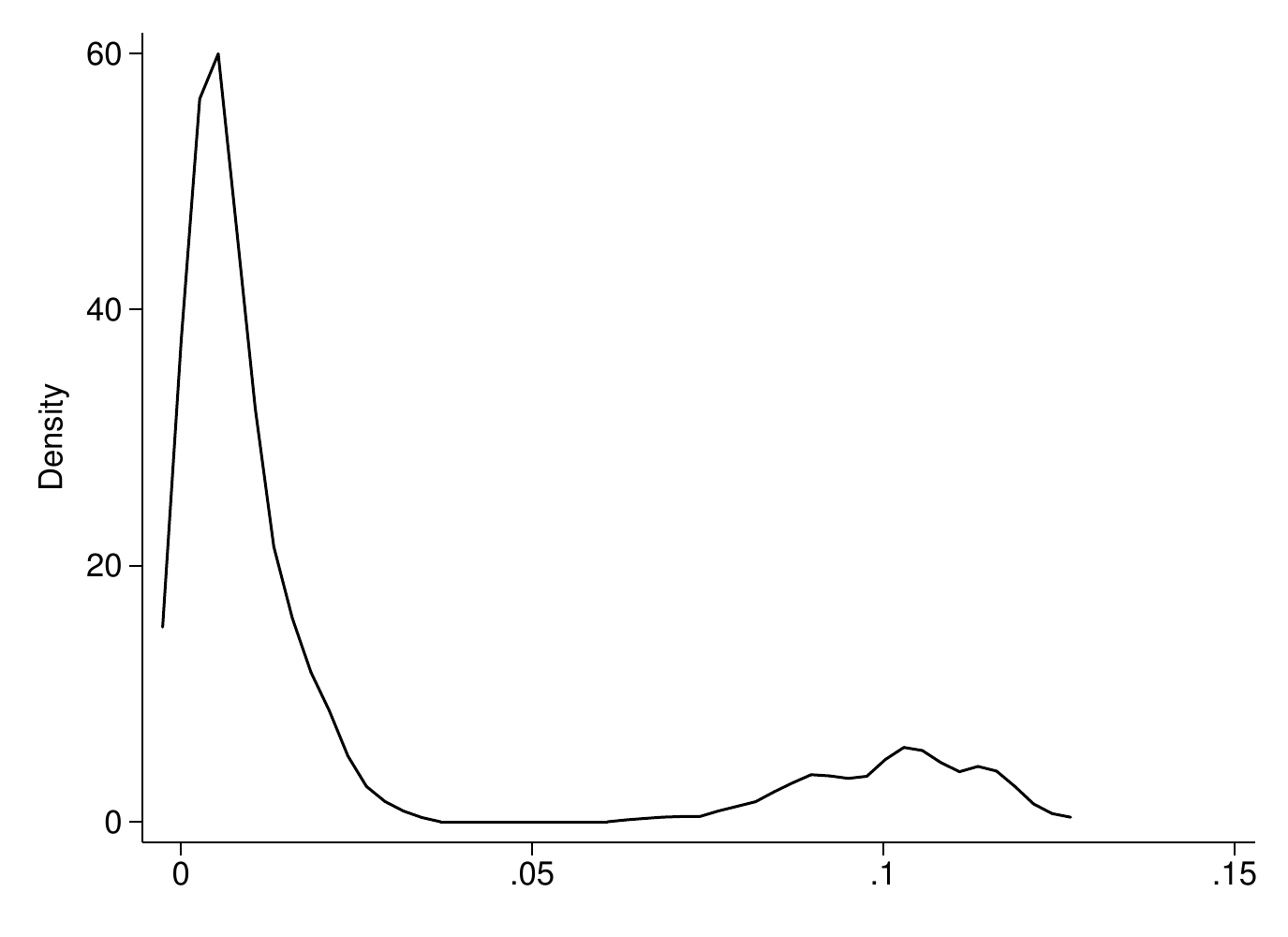}
\caption{Kernel density estimate of $\Delta_k$, the absolute difference between propensity score and leave-out-variable $k$ propensity score, for $k=$ Hispanic indicator, in the experimental dataset.}
\label{propScoreVar_d1}
\end{figure}

The leave-out-variable $k$ propensity score analysis focuses on the relationship between observed covariates and treatment assignment. 
It does not use data on outcomes. 
A less conservative analysis is to only worry about covariates $k$ which have large values in table \ref{tableCkInfo_d1} \emph{and} which also affect our outcomes in some way. 
Specifically, we next consider leave-out-variable $k$ IPW estimates of ATE under the baseline unconfoundedness assumption. 
Table \ref{tableATEdiffs} shows the effect of leaving out a single variable on the ATE point estimates for both datasets. 
Continue to consider just the experimental dataset. 
Here we first see that omitting any covariate at most changes the point estimate by 5.4\%. 
Moreover, recall the main variable we were concerned about before: the indicator for Hispanic. Omitting this variable only changes the ATE point estimate by 1.5\%. 

Overall, the leave-out-variable $k$ analysis suggests that, on an absolute scale, the conclusion that ATE is nonnegative using the experimental data is quite robust. 
A similar analysis applies to conclusions about ATT.

\begin{table}[t]
\caption{\label{tableATEdiffs} Magnitude of the effect of omitting a single variable on ATE point estimates (as a percentage of the baseline estimate).}\centering
\begin{tabular}{lcc} \hline \hline
 & Experimental dataset  & Observational dataset  \\  \hline 
Earnings in 1975 &      0.07 &      0.02 \\  
[0.1em]
Married &      0.21 &     14.27 \\  
[0.1em]
Positive earnings in 1974 &      1.35 &     10.20 \\  
[0.1em]
Hispanic &      1.51 &      1.01 \\  
[0.1em]
Black &      2.91 &     14.11 \\  
[0.1em]
Positive earnings in 1975 &      3.32 &      0.64 \\  
[0.1em]
Age &      3.36 &      6.49 \\  
[0.1em]
Earnings in 1974 &      3.90 &      0.34 \\  
[0.1em]
Education &      5.39 &      1.84 \\  
\hline \hline \end{tabular}
\end{table}

Next consider the observational data. 
Table \ref{tableCkInfo_d3} shows the leave-out-variable $k$ propensity score analysis. 
Recall that the estimated breakdown point for the conclusion that ATE is nonnegative in the observational dataset is \ATEobsBP. 
By \emph{any} of these measures the conclusion that ATE is nonnegative is not robust.
Suppose we only consider variables which also substantially change the point estimates, as shown in table \ref{tableATEdiffs}. 
Even then we still find that the results are sensitive. 
For example, the indicator for Black changes the ATE point estimate by 14\% and also has substantial marginal impact on the propensity score, with its 50th percentile in table \ref{tableCkInfo_d3} about 1.5 times as large as the estimated ATE breakdown point. 
Thus, using these as absolute measures of robustness, we find that the conclusion that ATE is positive using the observational data is not robust.

\begin{table}[htbp]
\caption{\label{tableCkInfo_d3} Variation in leave-out-variable-$k$ propensity scores, observational data.}\centering
\begin{tabular}{lcccc} \hline \hline
 & p50 & p75 & p90 & $\bar{c}_k$  \\  \hline 
Earnings in 1974 &     0.000 &     0.001 &     0.009 &     0.065 \\  
[0.2em]
Hispanic &     0.003 &     0.011 &     0.024 &     0.214 \\  
[0.2em]
Education &     0.006 &     0.017 &     0.042 &     0.127 \\  
[0.2em]
Earnings in 1975 &     0.002 &     0.010 &     0.057 &     0.276 \\  
[0.2em]
Positive earnings in 1975 &     0.007 &     0.019 &     0.076 &     0.295 \\  
[0.2em]
Positive earnings in 1974 &     0.012 &     0.028 &     0.099 &     0.423 \\  
[0.2em]
Married &     0.028 &     0.079 &     0.172 &     0.314 \\  
[0.2em]
Age &     0.035 &     0.093 &     0.205 &     0.508 \\  
[0.2em]
Black &     0.053 &     0.143 &     0.266 &     0.477 \\  
\hline \hline \end{tabular}
\end{table}

This conclusion that findings based on the observational dataset are not robust contrasts with the sensitivity analysis of \cite{Imbens2003}, who finds that the same observational dataset yields relatively robust results. 
Imbens' analysis relied importantly on fully parametric assumptions about the joint distribution of the observables and unobservables.
In particular, he assumed outcomes were normally distributed, that the treatment effect is homogeneous, and that any selection on unobservables arises due to an omitted binary variable.
Our identification analysis does not require any of these assumptions.
As discussed in section \ref{sec:estimation}, we do impose some parametric assumptions to simplify estimation, but even these assumptions are substantially weaker than those used by Imbens.
Given that we are making weaker auxiliary assumptions, it is not surprising that our analysis shows the findings to be more sensitive than the analysis in \cite{Imbens2003}.
Nonetheless, even with these weaker assumptions, we continue to find that conclusion from the experimental dataset remain robust.

Finally, note that all of our discussion thus far has focused on the point estimates of the breakdown points. In section \ref{sec:breakdownPointInference} we showed that the value at which the pointwise confidence band intersects the horizontal axis is a valid one-sided lower confidence interval for the breakdown point. For the ATE with experimental data, this gives a confidence set of $[\ATEexpBPconf,1]$, with a point estimate of $\ATEexpBP$. For the ATE with observational data, this gives a confidence set of $[\ATEobsBPconf,1]$, with a point estimate of $\ATEobsBP$. Thus the lower bound of the confidence interval for the experimental data is almost twice as large as the lower bound for the observational data. So the relative comparison of the two datasets continues to hold once we account for sampling uncertainty. Unfortunately, the lower bound of $\ATEexpBPconf$ for the experimental dataset is quite small, if we compare it to the variation in the leave-out-variable-$k$ propensity scores. This is not surprising though, given that there is a substantial amount of sampling uncertainty---even the lower bound of the confidence intervals for the baseline estimates are quite close to zero.

\subsection*{Can Selection on Unobservables Help the Program Pass a Cost-Benefit Analysis?}

In the previous subsection we studied the sensitivity of the conclusion that the ATE is nonnegative. In practice, however, this is not necessarily the most policy relevant conclusion. For example, \citet[section 3]{HeckmanSmith1998} give a model where the socially optimal decision whether to continue a small scale program or to shut it down can be computed by comparing the ATT with the program's per participant cost. In this subsection we show how our sensitivity analysis can be used in these kinds of cost-benefit analyses. Specifically, we consider the conclusion that the ATT is less than the per participant program cost. Under the model in \cite{HeckmanSmith1998}, the program should be shut down when this conclusion holds.

Chapter 8 of MDRC \citeyearpar{MDRC1983} reports NSW per participant program costs. For males, total costs ranges between \$4,637 and \$5,218 (tables 8-2, 8-3, and 8-4) in 1976 dollars. Our treatment effect estimates are in 1982 dollars. Adjusting these reported costs to 1982 dollars (CPI-U series) gives a range of \$7,865 to \$8,850.

First consider the experimental dataset. The conclusion of interest holds for the baseline estimate: The ATT of \$1,738 is far less than the per participant cost. Suppose, however, that a supporter of the program claims that this baseline estimate is implausible due to selection on unobservables. How strong does selection on unobservables need to be to allow for the possibility that the program is cost effective? Formally, what is the smallest $c$ such that the identified set for ATT includes values that are larger than the per participant cost? If we look at the bounds' point estimates, there are \emph{no} values of $c$ under which the program is cost effective. Accounting for sampling uncertainty by examining the confidence bands, we need $c$ to be at least about 0.5 before it is possible that the program is cost effective. As we argued earlier, these are very large values, so it is unlikely that selection on unobservables is this strong.

Next consider the observational dataset. Here again the conclusion of interest holds for the baseline estimate: The ATT of \$4,001 is smaller than the per participant cost. Next consider the breakdown point for this conclusion. Since the bounds for the observational dataset are larger than those for the experimental dataset, we need less selection on unobservables to allow for possibly large values of the ATT. Despite this, there are still no values of $c$ under which the program is cost effective, based on the bounds' point estimates. This is largely because the uncertainty due to the impact of selection on unobservables is asymmetric in this example: The lower bound grows much faster in $c$ than the upper bound does. Hence conclusions about the largest possible value of the ATT are more robust to relaxations of unconfoundedness than conclusions about the smallest possible value of the ATT. If we account for sampling uncertainty by examining the confidence bands, then we need $c$ to be at least about 0.08 before the confidence intervals contain ATT values larger than the per participant costs. This is a relatively large value, although it is smaller than a decent number of the leave-out-variable $k$ propensity score values in table \ref{tableCkInfo_d3}. That, however, likely just reflects the large amount of sampling uncertainty in this data.

Overall, our analysis suggests that the program does not pass a cost-benefit analysis, even if we allow for a large amount of selection on unobservables. Hence the conclusion that the ATT is less than the per participant cost, and hence that the program should be shut down, is quite robust to failures of unconfoundedness. 

Finally, note that our analysis here is primarily illustrative. A more comprehensive cost-benefit analysis would require examining many other program outcomes besides just short run post-program earnings. For example, see the analysis in chapter 8 of MDRC \citeyearpar{MDRC1983} and section 10 of \cite{HLS1999}. Note, however, that given data on these additional outcomes, our methods could then be used to analyze the sensitivity of total program impacts to failures of unconfoundedness.

\section{Conclusion}\label{sec:conclusion}

Identification, estimation, and inference on treatment effects under unconfoundedness has been widely studied and applied. This approach uses two assumptions: Unconfoundedness and Overlap. The overlap assumption is refutable, and many tools have been developed for checking this assumption in practice. For example, Stata's built-in package \texttt{teffects} has commands for checking overlap. In this paper, we provide a complementary suite of tools for assessing the unconfoundedness assumption. There are two key distinctions between our results and the previous literature. First, we begin from fully nonparametric bounds. In contrast, most of the previous literature relies on parametric assumptions for their identification analysis. Second, we provide tools for inference. This is important because, just like baseline estimators, sensitivity analyses are also subject to sampling uncertainty.

\subsection{Extensions and Future Work}\label{sec:additionalEstimators}

We conclude by discussing several extensions and directions for future work. As we just mentioned, a key distinguishing feature of our sensitivity analysis is that we begin from fully nonparametric bounds. We then estimated these bounds using flexible parametric estimators of the propensity score and the quantile regression of outcomes on treatment and covariates. These estimators can include quadratic terms, cubic terms, and interactions, for example, but they are not fully nonparametric. We restricted attention to parametric estimators for one reason: Even in this case, the asymptotic distribution theory is non-standard, complicated, and at the frontier of current research. This difficulty comes from the fact that our estimands are not Hadamard differentiable. Extending our analysis to first step nonparametric estimators is an important next step, but doing so will likely require both deriving and applying more general asymptotic theory for non-Hadamard differentiable functionals than currently exists. Hence we leave that analysis to future work.

A second extension is to consider additional parameters of interest. In this paper we focus on estimation and inference on the ATE and ATT bounds. We also developed analogous results for the CQTE and CATE. The conditional average treatment effect for the treated, $\text{CATT}(w) = \Exp(Y_1 - Y_0 \mid X=1, W=w)$, can be studied with the same tools we use in section \ref{sec:asymptotics}. We omit that analysis for brevity. \cite{MastenPoirier2018} also derive sharp bounds on unconditional quantile treatment effects (QTEs). Estimation and inference on the QTE bounds is more complicated than the ATE and ATT bounds. The reason is identical to the explanation \citet[page 307]{Vaart2000} gives when discussing inference on unconditional sample quantiles: ``to derive the asymptotic normality of even a single quantile estimator $\widehat{F}_n^{-1}(p)$, we need to know that the estimators $\widehat{F}_n$ are asymptotically normal as a process, in a neighborhood of $F^{-1}(p)$.'' In our case, performing inference on the QTE bounds requires showing convergence of corresponding bounds on the unconditional potential outcome cdfs as a process in a neighborhood of the quantile of interest.\footnote{\cite{MastenPoirier2020} prove some results along these lines; see their lemma 1. Those results are only valid for sufficiently small values of $c$ and with discrete $W$, which substantially simplifies the analysis.} For this reason, we leave estimation and inference on the QTE bounds to a separate paper.

\singlespacing
\bibliographystyle{econometrica}
\bibliography{Cdep_estimation}

\allowdisplaybreaks

\appendix

\section{Asymptotics for the First Step Estimators} \label{sec:prelimest}

In this appendix we formally state assumptions that ensure asymptotic normality of our first step estimators. 

\subsection{Assumptions}

We begin with the propensity score, which we estimate by maximum likelihood.

\begin{partialIndepAssump}[Propensity Score]\label{assn:prop score}
\hfill
\begin{enumerate}
\item (Correct specification) Let $\mathcal{B} \subseteq \R^{d_W}$ be compact. There is a $\beta_0 \in \text{int}(\mathcal{B})$ such that
\begin{align*}
	\Prob(X=x \mid W=w)
	&= F(w'\beta_0)^x(1-F(w'\beta_0))^{1-x} \\
	&\equiv L(x,w'\beta_0)
\end{align*}
for all $x \in \{0,1\}$ and $w \in \mathcal{W}$.

\item (Sufficient variation) There is no proper linear subspace $A$ of $\R^{d_W}$ such that $\Prob(W\in A) = 1$.

\item (Regularity of link function) $F: \R \to (0,1)$ is strictly increasing and twice continuously differentiable with uniformly bounded derivative.
\end{enumerate}
\end{partialIndepAssump}

This assumption requires our propensity score specification to be correct. It also imposes some standard assumptions on parameter space $\mathcal{B}$, the link function $F(\cdot)$, and the distribution of the covariates $W$. Next, let
\[
	\ell(x,w'\beta) = \log L(x,w'\beta)
\]
denote the log likelihood function. Let
\[
	\ell_\beta(x,w'\beta) = \frac{\partial}{\partial\beta} \ell(x,w'\beta)
	\qquad \text{and} \qquad
	\ell_{\beta\beta}(x,w'\beta) = \frac{\partial^2}{\partial\beta\partial\beta'} \ell(x,w'\beta)
\]
denote its vector of derivatives and second derivative matrix, respectively. Recall from section \ref{sec:asymptotics} that $\mathcal{B}_\delta = \{ \beta \in \mathcal{B} : \|\beta - \beta_0\| \leq \delta\}$. We impose the following assumptions on the propensity score as well.

\begin{partialIndepAssump}[Propensity Score Regularity]\label{assn:prop score regularity}
For each $x\in\{0,1\}$,
\begin{enumerate}
\item We have
\[
	\Exp \left( \sup_{\beta \in \mathcal{B}}|\ell(x,W'\beta)| \right) < \infty.
\]

\item For some $\delta >0$,
\[
	\int_{\mathcal{W}} \sup_{\beta \in \mathcal{B}_\delta} \left\| \frac{\partial}{\partial \beta} L(x,w'\beta)\right\| \; dw < \infty
	\qquad \text{and} \qquad
	\int_{\mathcal{W}} \sup_{\beta \in \mathcal{B}_\delta} \left\| \frac{\partial^2}{\partial \beta\partial\beta'} L(x,w'\beta)\right\| \; dw < \infty.
\]

\item The matrix $V_\beta = \Exp [\ell_\beta(X,W'\beta_0)\ell_\beta(X,W'\beta_0)']$ exists and is nonsingular.

\item For some $\delta >0$,
\[
	\Exp \left(\sup_{\beta \in \mathcal{B}_\delta}\|\ell_{\beta\beta}(x,W'\beta)\|\right)<\infty.
\] 
\end{enumerate}
\end{partialIndepAssump}

These conditions are standard for maximum likelihood estimators. For example, see theorem 3.3. in \cite{NeweyMcFadden1994} along with their discussion. Note that dominance conditions A\ref{assn:prop score regularity}.1, A\ref{assn:prop score regularity}.3, and A\ref{assn:prop score regularity}.4 hold in standard parametric models like logit and probit. Although not necessary, it also holds when $\Exp (\|W\|^2)<\infty$ and strong overlap holds; that is, when there exists $0<\underline{p}<\overline{p}<1$ such that $p_{1|w} \in [\underline{p},\overline{p}]$ for all $w \in \mathcal{W}$.

Besides the propensity score, the other first step estimator is the conditional quantile function of $Y$ given $(X,W)$. We consider here a linear quantile regression of $Y$ on $q(X,W)$, a set of flexible functions of $(X,W)$. We make the following assumptions.

\begin{partialIndepAssump}[Quantile Regression]\label{assn:quant reg}
There exists an $\smalleps \in (0,\varepsilon)$ such that
\begin{enumerate}
\item There is some $\gamma_0 \in \mathscr{C}( [\smalleps,1-\smalleps],\R^{d_q})$ such that
\[
	Q_{Y \mid X,W}(\tau \mid x,w) = q(x,w)'\gamma_0(\tau)
\]
for every $\tau \in [\smalleps,1-\smalleps]$.

\item The conditional density $f_{Y \mid q(X,W)}(y \mid q(x,w))$ exists and is bounded and uniformly continuous in $y$, uniformly in $q(x,w) \in \supp(q(X,W))$. 

\item The matrix
\[
	J(\tau) = \Exp \left[ f_{Y \mid q(X,W)} \big( q(X,W)'\gamma_0(\tau) \mid q(X,W) \big) q(X,W)q(X,W)' \right]
\]
is positive definite for all $\tau \in [\smalleps,1-\smalleps]$.

\item $\Exp( \|q(x,W)\|^4 ) < \infty$ for $x\in\{0,1\}$.
\end{enumerate}
\end{partialIndepAssump}

These are standard assumptions for obtaining limiting distributions of quantile regression processes indexed by $\tau \in [\smalleps,1-\smalleps]$. For example, see theorem 3 in \cite{AngristChernozhukovFernandez-Val2006}.

\subsection{Convergence Results}

We next prove two convergence results. The first is joint asymptotic normality of the first step estimators.

\begin{lemma}[First step estimators]\label{lemma:prelim estimators}
Suppose A\ref{assn:iid} and A\ref{assn:prop score}--A\ref{assn:quant reg} hold. Then
\[
	\sqrt{n}
	\begin{pmatrix}
		\widehat{\beta} - \beta_0\\
		\widehat{\gamma}(\tau) - \gamma_0(\tau)
	\end{pmatrix}
	\rightsquigarrow \mathbf{Z}_1(\tau),
\]
where $\mathbf{Z}_1(\cdot)$ is a mean-zero Gaussian process in $\R^{d_W} \times \ell^\infty([\varepsilon,1-\varepsilon],\R^{d_q})$ with uniformly continuous paths. Moreover, its covariance kernel can be written in block form as
\begin{align}\label{eq:prelimest_covkernel}
	\Exp [\mathbf{Z}_1(\tau_1)\mathbf{Z}_1(\tau_2)']
	=
	\begin{pmatrix}
		V_\beta & 0 \\
		0 & V_\gamma(\tau_1,\tau_2)	
	\end{pmatrix}
\end{align}
where
\[
	V_\beta = \Exp \left[\frac{F'(W'\beta_0)^2}{F(W'\beta_0)(1-F(W'\beta_0))}WW'\right]^{-1}
\]
and
\[
	V_\gamma(\tau_1,\tau_2) =  J(\tau_1)^{-1}(\min \{\tau_1, \tau_2 \} - \tau_1 \tau_2)\Exp [q(X,W)q(X,W)'] J(\tau_2)^{-1}.
\]
\end{lemma}

Next we provide a convergence result for estimates of the derivatives of the quantile regression coefficients. We estimate these derivatives as follows. Let $\tau \in [\varepsilon, 1-\varepsilon]$. Then
\begin{equation}\label{eq:QRderivativeEstimator}
	\widehat{\gamma}'(\tau)
	= \frac{\widehat{\gamma}(\tau + \eta_n) - \widehat{\gamma}(\tau -\eta_n)}{2\eta_n}
\end{equation}
where $\eta_n>0$ is a tuning parameter that is chosen to be small enough such that $[\tau - \eta_n, \tau + \eta_n]$ is contained in $(\smalleps, 1-\smalleps)$. The next result shows that these estimators are uniformly consistent. %

\begin{lemma}[Convergence of QR derivatives]\label{applemma:QRderivatives}
Let $\eta_n \to 0$ and $n\eta_n^2 \to \infty$ as $n\to\infty$. Suppose the assumptions of lemma \ref{lemma:prelim estimators} hold. Suppose A\ref{assn:quant reg regularity} holds. Then
\[
	\sup_{\tau\in[\varepsilon,1-\varepsilon]} \|\widehat{\gamma}'(\tau) - \gamma'_0(\tau)\| = o_p(1).
\]
\end{lemma}

\subsection{Proofs}

We begin by examining each of the first step estimators separately.

\begin{lemma}[Propensity score estimation] \label{applemma:prop score estimation}
Suppose A\ref{assn:iid} and A\ref{assn:prop score}--A\ref{assn:prop score regularity} hold. Then
\[
	\sqrt{n}(\widehat{\beta} - \beta_0) = \frac{1}{\sqrt{n}}\sum_{i=1}^n V_\beta \frac{F'(W_i'\beta_0)(X_i - F(W_i'\beta_0))W_i}{F(W_i'\beta_0)(1-F(W_i'\beta_0))} + o_p(1)
\]
and hence
\[
	\sqrt{n}(\widehat{\beta} - \beta_0)
	\xrightarrow{d}
	\normal \left(0, V_\beta\right).
\]
\end{lemma}

\begin{proof}[Proof of lemma \ref{applemma:prop score estimation}]
This result follows from theorem 3.3 (asymptotic normality of MLEs) in \cite{NeweyMcFadden1994}. So it suffices to verify that their assumptions hold.
\begin{enumerate}
\item Their theorem 3.3 begins by supposing the assumptions of their theorem 2.5 (consistency of MLEs) holds. So we verify those assumptions first. By A\ref{assn:prop score}.1 $\ell(x,w'\beta) = \ell(x,w' \widetilde{\beta})$ for all $(x,w) \in \supp(X,W)$ implies that $w'\beta = w' \widetilde{\beta}$ for all $w \in \supp(W)$. By A\ref{assn:prop score}.2 this implies that $\beta = \widetilde{\beta}$. So assumption (i) of their theorem 2.5 holds. We directly assume that their assumptions (ii), (iii), and (iv) hold (via our A\ref{assn:prop score} and A\ref{assn:prop score regularity}.1). Finally, note that A\ref{assn:iid} is our assumption that $\{ (Y_i,X_i,W_i) \}_{i=1}^n$ are iid. Thus all assumptions of their theorem 2.5 hold. %

\item Next we consider the additional assumptions imposed in their theorem 3.3, (i)--(v). These are directly implied by our A\ref{assn:prop score} and A\ref{assn:prop score regularity}.
\end{enumerate}
Thus all assumptions of their theorem 3.3 hold. This gives us $\sqrt{n}(\widehat{\beta} - \beta_0) \xrightarrow{d} \normal \left(0, V_\beta\right)$. The asymptotic linear representation holds by arguments in the proof of their theorem 3.1 and the discussion on pages 2142--2143.
\end{proof}

\begin{lemma}[Quantile regression estimation]\label{applemma:quantile estimation}
Suppose A\ref{assn:iid} and A\ref{assn:quant reg} hold. Then
\begin{align*}
	\sqrt{n}(\widehat{\gamma}(\tau) - \gamma_0(\tau))
	&= J(\tau)^{-1} \frac{1}{\sqrt{n}}
	\sum_{i=1}^n \Big( \tau - \mathbbm{1} \big( Y_i \leq q(X_i,W_i)'\gamma_0(\tau) \big) \Big) q(X_i,W_i) + o_p(1)\\
	&\rightsquigarrow J(\tau)^{-1} \mathbf{Z}_\gamma(\tau),
\end{align*}
where $\mathbf{Z}_\gamma(\cdot)$ is a mean-zero Gaussian process in $\ell^{\infty}([\smalleps,1-\smalleps],\R^{d_q})$ with continuous paths and covariance kernel equal to $\Sigma(\tau_1, \tau_2) = (\min \{ \tau_1, \tau_2 \} - \tau_1 \tau_2) \Exp[q(X,W)q(X,W)']$.
\end{lemma}

\begin{proof}[Proof of lemma \ref{applemma:quantile estimation}]
By Minkowski's inequality,
\begin{align*}
	\Exp(\|q(X,W)\|^4)^{1/4} &= \Exp( \|Xq(1,W) + (1-X)q(0,W)\|^4 )^{1/4}\\
	&\leq \Exp( |X| \cdot \| q(1,W)\|^4)^{1/4} + \Exp( |1-X| \cdot \|q(0,W)\|^4)^{1/4}.
\end{align*}	
By A\ref{assn:quant reg}.4 and $X\in\{0,1\}$, it follows that $\Exp (\|q(X,W)\|^4) < \infty$.  The result then follows directly from theorem 3 in \cite{AngristChernozhukovFernandez-Val2006}.
\end{proof}

\begin{proof}[Proof of lemma \ref{lemma:prelim estimators}]
Note that the influence functions in lemmas \ref{applemma:prop score estimation} and \ref{applemma:quantile estimation} are Donsker. Therefore we can stack them to obtain joint weak convergence to $\mathbf{Z}_1(\tau)$. Next, note that this result holds over $\tau \in [\smalleps,1-\smalleps]$. This is a strict superset of $[\varepsilon,1-\varepsilon]$, so it holds on that set too.

Finally, we note that the diagonal element of the covariance kernel is zero. This diagonal element is
\[
	C_{\beta,\gamma}(\tau) =
	\cov\left(V_\beta \frac{F'(W'\beta_0)(X - F(W'\beta_0))W}{F(W'\beta_0)(1-F(W'\beta_0))},  \Big( \tau - \indicator(Y \leq q(X,W)'\gamma_0(\tau) ) \Big) q(X,W)'J(\tau)^{-1}\right).
\]
By iterated expectations,
\begin{align*}
	&\Exp \left[ \Big( \tau - \indicator(Y \leq q(X,W)'\gamma_0(\tau) ) \Big) q(X,W)'J(\tau)^{-1} \right] \\
	&= \Exp \left[ \Big( \tau - \Exp \left[ \indicator(Y \leq q(X,W)'\gamma_0(\tau) ) \mid X,W \right] \Big) q(X,W)'J(\tau)^{-1} \right] \\
	&= 0
\end{align*}
since $\Prob(Y \leq q(x,w)' \gamma_0(\tau) \mid X=x, W=w) = \tau$ by correct specification of the conditional quantile function. Also,
\begin{align*}
	&\Exp \left[
	V_\beta \frac{F'(W'\beta_0)(X - F(W'\beta_0))W}{F(W'\beta_0)(1-F(W'\beta_0))} \Big( \tau - \indicator(Y \leq q(X,W)'\gamma_0(\tau) ) \Big) q(X,W)'J(\tau)^{-1}
	\right] = 0
\end{align*}
by a similar argument, using iterated expectations conditional on $(X,W)$, by correct specification of the conditional quantile function, and since the first term is deterministic conditional on $(X,W)$. Thus $C_{\beta,\gamma}(\tau) = 0$ by definition of the covariance.
\end{proof}

\begin{proof}[Proof of lemma \ref{applemma:QRderivatives}]
Without loss of generality, consider the convergence of $\widehat{\gamma}'_{(1)}$ to $\gamma'_{0,(1)}$ the first component of $\gamma'_0$. Since $\eta_n \to 0$, let $\eta_n$ be small enough such that $\eta_n \in (0, \varepsilon - \smalleps)$. Then
\begin{align*}
	&\sup_{\tau\in[\varepsilon,1-\varepsilon]}|\widehat{\gamma}'_{(1)}(\tau) - \gamma'_{0,(1)}(\tau)|\\
	&\leq \sup_{\tau\in[\varepsilon,1-\varepsilon]}\left| \frac{\widehat{\gamma}_{(1)}(\tau + \eta_n) - \widehat{\gamma}_{(1)}(\tau -\eta_n)}{2\eta_n} -  \frac{\gamma_{0,(1)}(\tau + \eta_n) - \gamma_{0,(1)}(\tau -\eta_n)}{2\eta_n}\right|\\
	& \quad + \sup_{\tau\in[\varepsilon,1-\varepsilon]}\left| \frac{\gamma_{0,(1)}(\tau + \eta_n) - \gamma_{0,(1)}(\tau -\eta_n)}{2\eta_n} - \gamma_{0,(1)}'(\tau)\right|\\
	&\leq \frac{1}{2\eta_n}\left(\sup_{\tau\in[\varepsilon,1-\varepsilon]}\left| \widehat{\gamma}_{(1)}(\tau + \eta_n)  - \gamma_{0,(1)}(\tau + \eta_n)\right| + \sup_{\tau\in[\varepsilon,1-\varepsilon]}\left| \widehat{\gamma}_{(1)}(\tau- \eta_n)  - \gamma_{0,(1)}(\tau - \eta_n)\right|\right) \\
	& \quad + \sup_{\tau\in[\varepsilon,1-\varepsilon]}\left|\frac{\left(\gamma_{0,(1)}(\tau) + \gamma_{0,(1)}'(\tau)\eta_n + \gamma_{0,(1)}^{\prime\prime}(\tau)\frac{\eta_n^2}{2} + \gamma_{0,(1)}^{\prime\prime\prime}(\tau_{n1}^*)\frac{\eta_n^3}{6}\right)}{2\eta_n} \right. \\
	&\hspace{30mm} -\left. \frac{\left(\gamma_{0,(1)}(\tau) - \gamma_{0,(1)}'(\tau)\eta_n + \gamma_{0,(1)}^{\prime\prime}(\tau)\frac{\eta_n^2}{2} - \gamma_{0,(1)}^{\prime\prime\prime}(\tau_{n2}^*)\frac{\eta_n^3}{6}\right)}{2\eta_n} - \gamma_{0,(1)}'(\tau)\right| \\
	&\leq \frac{1}{2\eta_n}
	\left(
	\sup_{\tau\in[\smalleps,1- \smalleps]}
	\left| \widehat{\gamma}_{(1)}(\tau )  - \gamma_{0,(1)}(\tau)\right| + \sup_{\tau\in[\smalleps,1- \smalleps]}\left| \widehat{\gamma}_{(1)}(\tau)  - \gamma_{0,(1)}(\tau)\right|\right) \\
	&\quad + \sup_{\tau\in[\varepsilon,1-\varepsilon]}\left|\gamma_{0,(1)}^{\prime\prime\prime}(\tau_{n1}^*)\frac{\eta_n^2}{12} + \gamma_{0,(1)}^{\prime\prime\prime}(\tau_{n2}^*)\frac{\eta_n^2}{12} \right|.
\end{align*}
The first inequality follows by the triangle inequality and the definition of $\widehat{\gamma}_{(1)}'$. The second inequality follows by taking two third order Taylor expansions of $\gamma_{0,(1)}$, where $\tau_{n1}^* \in (\tau, \tau + \eta_n)$ and $\tau_{n2}^* \in (\tau-\eta_n,\tau)$. There we use A\ref{assn:quant reg regularity} with $m \geq 3$. The last inequality follows for two reasons: In the first term, $\smalleps < \varepsilon$ and $\eta_n$ is chosen such that $[\tau - \eta_n, \tau + \eta_n ] \subset (\smalleps, 1-\smalleps)$. So we are taking the supremum over a larger set in this term in the last line. In the second term, the 0th, 1st, and 2nd derivatives of $\gamma_{0,(1)}$ all cancel, leaving only the third order derivatives remaining.

Finally,
\begin{align*}
	\sup_{\tau\in[\varepsilon,1-\varepsilon]}|
	\widehat{\gamma}'_{(1)}(\tau) - \gamma'_{0,(1)}(\tau)|
	&= \frac{1}{2\eta_n} O_p\left(\frac{1}{\sqrt{n}}\right) + \frac{B}{6} \eta_n^2\\
	&= o_p(1).
\end{align*}
The first line follows by lemma \ref{applemma:quantile estimation}, which shows that the first term is $(1/2\eta_n) O_p(1/\sqrt{n}) = O_p(1/\sqrt{n \eta_n^2})$. In the second term $B > 0$ is a constant that doesn't depend on $n$. This constant comes from A\ref{assn:quant reg regularity}, which implies the function $\gamma_0$ has uniformly bounded third derivatives. The last line follows by $\eta_n \to 0$, $n\eta_n^2 \to \infty$ as $n\to\infty$. 

Repeating this argument across all components of $\widehat{\gamma}'(\tau) - \gamma_0'(\tau)$ shows that $\sup_{\tau\in[\varepsilon,1-\varepsilon]} \|\widehat{\gamma}'(\tau) - \gamma'_0(\tau)\| = o_p(1)$, as desired.
\end{proof}

\section{Proofs for Section \ref{sec:asymptotics}}\label{sec:proofs}

In this section we give the proofs for the results in section \ref{sec:asymptotics}. We start with a preliminary result on the asymptotic distribution of the CQTE bound estimators. We use this for all of our later results. We then state and prove a useful lemma. Finally we give the proofs for our CATE, ATE, and ATT bound estimators.

All of these results rely on proving Hadamard directional differentiability of various functionals. For that reason, it is helpful to recall its definition.

\begin{definition}\label{def:HDD}
Let $\phi: \mathbb{D}_\phi \to \mathbb{E}$ where $\mathbb{D}$, $\mathbb{E}$ are Banach spaces and $\mathbb{D}_\phi \subseteq \mathbb{D}$. Say $\phi$ is \emph{Hadamard directionally differentiable} at $\theta \in \mathbb{D}_\phi$ tangentially to $\mathbb{D}_0 \subseteq \mathbb{D}$ if there is a continuous map $\phi_\theta':\mathbb{D}_0 \to \mathbb{E}$ such that
\[
	\lim_{m\to\infty}\left\| \frac{\phi(\theta + t_m h_m) - \phi(\theta)}{t_m} - \phi_\theta'(h)\right\|_\mathbb{E} = 0
\]
for all sequences $\{h_m\}\subset \mathbb{D}$, $t_m > 0$ such that $t_m \to 0$, $h_m \to h \in \mathbb{D}_0$ as $m\to\infty$ and $\theta + t_m h_m \in \mathbb{D}_\phi$ for all $m$. 
\end{definition}

By proposition 2.1 in \cite{FangSantos2014}, the mapping $\phi$ is Hadamard differentiable at $\theta$ tangentially to $\mathbb{D}_0$ if and only if it is Hadamard directionally differentiable at $\theta$ tangentially to $\mathbb{D}_0$ and the mapping $\phi'_\theta$ is linear.

Throughout the proofs we let $\|\gamma\|_\infty = \sup_{\tau \in [\varepsilon,1-\varepsilon]} \|\gamma(\tau)\|$ denote the sup-norm in $\ell^\infty([\varepsilon,1-\varepsilon],\R^{d_q})$.

\subsection{The CQTE Bounds}

We start with a preliminary result for our estimates of the CQTE bounds. All our other bounds are built from these, so it is helpful to understand them first.

\begin{proposition}[CQTE convergence]\label{prop:CQTE convergence}
Suppose A\ref{assn:iid}, A\ref{assn:quant reg regularity}, and A\ref{assn:prop score}--A\ref{assn:quant reg} hold. Fix $\varepsilon > 0$, $w \in \mathcal{W}$, $c \in [0,1]$, and $\tau \in (0,1)$. Then
\[
	\sqrt{n}
	\begin{pmatrix}
		\widehat{\overline{\text{CQTE}}}^c(\tau \mid w) - \overline{\text{CQTE}}_\varepsilon^c(\tau \mid w)\\
		\widehat{\underline{\text{CQTE}}}^c(\tau \mid w) - \underline{\text{CQTE}}_\varepsilon^c(\tau \mid w)
	\end{pmatrix}
	\xrightarrow{d}\mathbf{Z}_{\text{CQTE}}(w,\tau),
\]
where $\mathbf{Z}_{\text{CQTE}}$ is a random vector in $\R^2$ whose distribution is characterized in the proof.
\end{proposition}

\begin{proof}[Proof of proposition \ref{prop:CQTE convergence}]
\hfill

\medskip

\textbf{Part 1: The upper bound is HDD}. Recall from section \ref{sec:SecondStepEstConvergence} that $\overline{\Gamma}_1(x,w,\tau,\theta)$ denotes our trimmed population conditional quantile upper bound. We write this parameter as a function of a few different pieces:
\[
	\overline{\Gamma}_1(x,w,\tau, \theta) = q(x,w)'\gamma (S_2(x,w,\tau, \beta))
\]
where
\[
	S_2(x,w,\tau, \beta) = \max\{S_1(x,w,\tau, \beta), \varepsilon\}
\]
and
\[
	S_1(x,w,\tau, \beta) = \min\left\{\tau + \frac{c}{L(x,w'\beta)}\min\{\tau,1-\tau\}, \ \frac{\tau}{L(x,w'\beta)}, \ 1-\varepsilon\right\}.
\]
For simplicity, we leave the dependence on $\varepsilon$ and $c$ implicit in our notation for $S_1$ and $S_2$. There are now three steps: We show Hadamard directional differentiability (HDD) of $\overline{\Gamma}_1(x,w,\tau,\cdot)$ at $\theta_0$ tangentially to $\R^{d_W} \times \mathscr{C}([\varepsilon,1-\varepsilon],\R^{d_q})$ by examining the two pieces $S_1$ and $S_2$ separately. We then combine these to show HDD of $\overline{\Gamma}_1$.

\bigskip

\emph{\textbf{Step 1: HDD of $S_1$}}. We first show $S_1(x,w,\tau, \cdot)$ is HDD at $\beta_0$. Let $t_m \searrow 0$ and $h_{1m} \to h_1 \in \R^{d_W}$ as $m \rightarrow \infty$. Define the secant line
\[
	T_1^m(x,w,\tau, \beta_0,h_{1m})
	= \frac{S_1(x,w,\tau, \beta_0 + t_m h_{1m}) - S_1(x,w,\tau, \beta_0)}{t_m}.
\]
We will show that
\begin{align*}
	T_1^m(x,w,\tau, \beta_0,h_{1m})
		&\to T_1(x,w,\tau,\beta_0,h_1,0) \\
		&= \sum_{j=1}^7 T_{1,j}(x,w,\tau, \beta_0,h_{1}) \indicator_{1,j}(x,w,\tau, \beta_0,0)
\end{align*}
as $m \to \infty$, where $T_{1,j}$ are defined in appendix \ref{sec:HDDformulas} below. To see this, we consider the seven cases associated with $\indicator_{1,j}$ for $j=1,\ldots,7$. 

First suppose $\indicator_{1,1}(x,w,\tau,\beta_0,0) = 1$. Then
\[
	S_1(x,w,\tau,\beta_0) = \tau + \frac{c}{L(x,w'\beta_0)}\min\{\tau,1-\tau\}.
\]
Moreover, for $m$ large enough and by continuity of $S_1$ in $\beta$, $\indicator_{1,1}(x,w,\tau,\beta_0 + t_m h_{1m},0) = 1$. Hence
\[
	S_1(x,w,\tau,\beta_0 + t_m h_{1m}) = \tau + \frac{c}{L(x,w'(\beta_0 + t_m h_{1m}))}\min\{\tau,1-\tau\}.
\]
So for $m$ large enough
\begin{align*}
	T_{1}^m(x,w,\tau, \beta_0,h_{1m}) 
	&= c \min\{\tau,1-\tau\} \frac{1}{t_m} \left(\frac{1}{L(x,w'(\beta_0 + t_m h_{1m}))}- \frac{1}{L(x,w'\beta_0)}\right) \\
	&\to T_{1,1}(x,w,\tau,\beta_0,h_1)
\end{align*}
by the definition of the directional derivative of $1/L(x,w'\beta)$ with respect to $\beta$ in the direction $h_1$ at $\beta_0$.

Similarly, if $\indicator_{1,2}(x,w,\tau,\beta_0,0) = 1$ then
\begin{align*}
	T_{1}^m(x,w,\tau, \beta_0,h_{1m})
		&= \tau \frac{1}{t_m} \left(\frac{1}{L(x,w'(\beta_0 + t_m h_{1m}))}- \frac{1}{L(x,w'\beta_0)}\right) \\
		&\rightarrow T_{1,2}(x,w,\tau,\beta_0,h_1)
\end{align*}
where the first line holds for $m$ large enough and the second line holds as $m \rightarrow \infty$. Likewise, if $\indicator_{1,3}(x,w,\tau,\beta_0,0) = 1$ then
\begin{align*}
	T_{1}^m(x,w,\tau, \beta_0,h_{1m})
		&= 0 \\
		&\rightarrow 0 \\
		&= T_{1,3}(x,w,\tau,\beta_0,h_1).
\end{align*}
where the first line holds for $m$ large enough and the convergence in the second line is as $m \rightarrow \infty$.

If $\indicator_{1,4}(x,w,\tau,\beta_0,0) = 1$ then
\[
	\tau + \frac{c}{L(x,w'\beta_0)} \min\{\tau,1-\tau\} = \frac{\tau}{L(x,w'\beta_0)} < 1-\varepsilon.
\]
For $m$ large enough,
\[
	S_1(x,w,\tau,\beta_0 + t_m h_{1m}) = \min\left\{\tau + \frac{c}{L(x,w'(\beta_0 + t_m h_{1m}))}\min\{\tau,1-\tau\}, \ \frac{\tau}{L(x,w'(\beta_0 + t_m h_{1m}))}\right\}.
\]
Hence
\begin{align*}
	T_{1}^m(x,w,\tau, \beta_0,h_{1m}) &= \min\left\{ c\min\{\tau,1-\tau\} \frac{1}{t_m}  \left(\frac{1}{L(x,w'(\beta_0 + t_m h_{1m}))}- \frac{1}{L(x,w'\beta_0)}\right), \right. \\[1em]
	&\hspace{40mm} \left.
	\tau \frac{1}{t_m} \left(\frac{1}{L(x,w'(\beta_0 + t_m h_{1m}))}- \frac{1}{L(x,w'\beta_0)}\right) \right\}
\end{align*}
for $m$ large enough. Similarly,
\begin{align*}
T_{1}^m(x,w,\tau, \beta_0,h_{1m}) &= \min\left\{c\min\{\tau,1-\tau\} \frac{1}{t_m} \left(\frac{1}{L(x,w'(\beta_0 + t_m h_{1m}))}- \frac{1}{L(x,w'\beta_0)}\right), 0 \right\} \\[1em]
	T_{1}^m(x,w,\tau, \beta_0,h_{1m}) &= \min\left\{ \tau \frac{1}{t_m} \left(\frac{1}{L(x,w'(\beta_0 + t_m h_{1m}))}- \frac{1}{L(x,w\beta_0)}\right),0\right\} \\[1em]
	T_{1}^m(x,w,\tau, \beta_0,h_{1m})
	&= \min\left\{c\min\{\tau,1-\tau\} \frac{1}{t_m} \left(\frac{1}{L(x,w'(\beta_0 + t_m h_{1m}))}- \frac{1}{L(x,w'\beta_0)}\right), \right. \\
	&\hspace{35mm} \left. \tau \frac{1}{t_m} \left(\frac{1}{L(x,w'(\beta_0 + t_m h_{1m}))}- \frac{1}{L(x,w'\beta_0)}\right),0\right\}
\end{align*}
for $m$ large enough when $\indicator_{1,j} = 1$ for $j=5,6,7$ respectively. Letting $t_m \searrow 0$, $h_{1m} \rightarrow h_1$, and by examining $T_{1}^m$  for $\indicator_{1,j} = 1$, $j=4,5,6,7$, we see that
\[
	\frac{S_1(x,w,\tau, \beta_0 + t_m h_{1m}) - S_1(x,w,\tau, \beta_0)}{t_m} = T_{1}^m(x,w,\tau, \beta_0,h_{1m})
	\rightarrow 
	T_1(x,w,\tau,\beta_0, h_1, 0).
\]
Hence $S_1(x,w,\tau, \cdot)$ is Hadamard directionally differentiable at $\beta_0$.

\bigskip

\emph{\textbf{Step 2: HDD of $S_2$}}. Recall that
\[
	S_2(x,w,\tau, \beta) = \max\{S_1(x,w,\tau, \beta), \varepsilon\}.
\]
As before, let $t_m \searrow 0$ and $h_{1m} \rightarrow h_1 \in \R^{d_W}$ as $m \rightarrow \infty$. Define
\[
	T_2^m(x,w,\tau,\beta_0,h_{1m}) = \frac{S_2(x,w,\tau, \beta_0 + t_m h_{1m}) - S_2(x,w,\tau, \beta_0)}{t_m}.
\]
Substituting the functional form for $S_1$ into the definition of $S_2$ gives
\begin{align*}
	&T_2^m(x,w,\tau,\beta_0,h_{1m}) \\
	&= \frac{1}{t_m} \max\left\{\min\left\{\tau + \frac{c}{L(x,w'(\beta_0 + t_mh_{1m}))}\min\{\tau,1-\tau\},\frac{\tau}{L(x,w'(\beta_0 + t_mh_{1m}))},1-\varepsilon\right\},\varepsilon\right\} \\
	&\qquad- \frac{1}{t_m} \max\left\{\min\left\{\tau + \frac{c}{L(x,w'\beta_0)}\min\{\tau,1-\tau\},\frac{\tau}{L(x,w'\beta_0)},1-\varepsilon\right\},\varepsilon\right\}.
\end{align*}
As in step 1, we next characterize the value of this secant line by splitting it into three different cases.
\begin{enumerate}
\item If
\[
	\min\left\{\tau + \frac{c}{L(x,w'\beta_0)}\min\{\tau,1-\tau\}, \frac{\tau}{L(x,w'\beta_0)},1-\varepsilon\right\} > \varepsilon
\]
then
\begin{align*}
	&T_2^m(x,w,\tau,\beta_0,h_{1m}) \\
	&=  \frac{1}{t_m} \min\left\{\tau + \frac{c}{L(x,w'(\beta_0 + t_m h_{1m}))}\min\{\tau,1-\tau\},\frac{\tau}{L(x,w'(\beta_0 + t_m h_{1m}))},1-\varepsilon\right\} \\
	 &\qquad - \frac{1}{t_m} \min\left\{\tau + \frac{c}{L(x,w'\beta_0)}\min\{\tau,1-\tau\},\frac{\tau}{L(x,w'\beta_0)},1-\varepsilon\right\}
\end{align*}
for large enough $m$.

\item If
\[
	 \min\left\{\tau + \frac{c}{L(x,w'\beta_0)}\min\{\tau,1-\tau\}, \frac{\tau}{L(x,w'\beta_0)},1-\varepsilon\right\} < \varepsilon
\]
then
\[
	T_2^m(x,w,\tau,\beta_0,h_{1m}) = 0
\]
for large enough $m$.

\item If
\[
	\min\left\{\tau + \frac{c}{L(x,w'\beta_0)}\min\{\tau,1-\tau\}, \frac{\tau}{L(x,w'\beta_0)},1-\varepsilon\right\} = \varepsilon
\]
then
\begin{align*}
	&T_2^m(x,w,\tau,\beta_0,h_{1m}) \\
	&= \max\left\{ \frac{1}{t_m} \min\left\{\tau + \frac{c}{L(x,w'(\beta_0 + t_m h_{1m}))}\min\{\tau,1-\tau\},\frac{\tau}{L(x,w'(\beta_0 + t_m h_{1m}))},1-\varepsilon\right\} \right. \\
	&\qquad \qquad \left. - \frac{1}{t_m} \min\left\{\tau + \frac{c}{L(x,w'\beta_0)}\min\{\tau,1-\tau\},\frac{\tau}{L(x,w'\beta_0)},1-\varepsilon\right\} ,0\right\}
\end{align*}
for large enough $m$.
\end{enumerate}
Using similar arguments as in step 1, by examining each of the three cases we see that
\[
	T_2^m(x,w,\tau,\beta_0,h_{1m}) \rightarrow T_2(x,w,\tau,\beta_0, h_1, 0)
\]
as $m \rightarrow \infty$. Hence $S_2(x,w,\tau, \cdot)$ is Hadamard directionally differentiable at $\beta_0$.

\bigskip

\emph{\textbf{Step 3: HDD of $\overline{\Gamma}_1$}}. Next we show that $\overline{\Gamma}_1(x,w,\tau, \cdot)$ is HDD at $\theta_0$ tangentially to $\R^{d_W} \times \mathscr{C}([\varepsilon,1-\varepsilon],\R^{d_q})$. Let $t_m \searrow 0$, $h_{1m} \to h_1 \in \R^{d_W}$ and $h_{2m} \to h_2 \in \mathscr{C}([\varepsilon,1-\varepsilon],\R^{d_q})$ endowed with the sup norm, as $m \rightarrow \infty$. Let $h_m = (h_{1m}, h_{2m})$ and $h = (h_1,h_2)$. Then
\begin{align*}
	&\frac{\overline{\Gamma}_1(x,w,\tau, \theta_0 + t_m h_m) - \overline{\Gamma}_1(x,w,\tau, \theta_0)}{t_m}\\
	 &= \frac{q(x,w)'[\gamma_0 + t_m h_{2m}](S_2(x,w,\tau, \beta_0 + t_m h_{1m})) - q(x,w)'\gamma_0(S_2(x,w,\tau, \beta_0))}{t_m}\\
	 &= \frac{q(x,w)'(\gamma_0(S_2(x,w,\tau, \beta_0 + t_m h_{1m})) - \gamma_0(S_2(x,w,\tau, \beta_0)))}{t_m}  + q(x,w)'h_{2m}(S_2(x,w,\tau, \beta_0 + t_m h_{1m})).
\end{align*}	
Consider the first term. By A\ref{assn:quant reg regularity}, $\gamma_0(u)$ is differentiable for any $u\in [\varepsilon,1-\varepsilon]$. By the chain rule,
\begin{multline*}
	\frac{q(x,w)'\left[\gamma_0(S_2(x,w,\tau, \beta_0 + t_m h_{1m})) - \gamma_0(S_2(x,w,\tau, \beta_0))\right]}{t_m} \\
	\to q(x,w)'\gamma_0'(S_2(x,w,\tau, \beta_0)) T_2(x,w,\tau,\beta_0, h_1, 0)
\end{multline*}
as $m \rightarrow \infty$. Next consider the second term. We have
\begin{align*}
		&|q(x,w)'h_{2m}(S_2(x,w,\tau, \beta_0 + t_m h_{1m})) - q(x,w)'h_2(S_2(x,w,\tau, \beta_0))|\\
		&\leq |q(x,w)'h_{2m}(S_2(x,w,\tau, \beta_0 + t_m h_{1m})) - q(x,w)'h_{2}(S_2(x,w,\tau, \beta_0 + t_m h_{1m}))| \\
		&\quad + |q(x,w)'h_{2}(S_2(x,w,\tau, \beta_0 + t_m h_{1m})) - q(x,w)'h_{2}(S_2(x,w,\tau, \beta_0))|\\
		&\leq \|q(x,w)\| \cdot \|h_{2m} - h_2\|_\infty + |q(x,w)'h_{2}(S_2(x,w,\tau, \beta_0 + t_m h_{1m})) - q(x,w)'h_{2}(S_2(x,w,\tau, \beta_0))|.
\end{align*}
The first inequality follows by the triangle inequality. Consider the second inequality. By continuity of $h_2(\cdot)$ and of $S_2(x,w,\tau, \cdot)$, the second term converges to zero as $m \rightarrow \infty$. By sup-norm convergence of $h_{2m}$ to $h_2$, the first term also converges to zero as $m \rightarrow \infty$. Thus
\[
	q(x,w)'h_{2m}(S_2(x,w,\tau, \beta_0 + t_m h_{1m})) \rightarrow q(x,w)'h_2(S_2(x,w,\tau, \beta_0))
\]
as $m \rightarrow \infty$. Putting the two terms together gives
\begin{align*}
	&\frac{\overline{\Gamma}_1(x,w,\tau, \theta_0 + t_m h_m) - \overline{\Gamma}_1(x,w,\tau, \theta_0)}{t_m} \\
	&\qquad\to q(x,w)'\gamma_0'(S_2(x,w,\tau, \beta_0))T_2(x,w,\tau,\beta_0, h_1, 0) + q(x,w)'h_{2}(S_2(x,w,\tau, \beta_0)) \\
	  &\qquad\equiv \overline{\Gamma}_{1,\theta_0}'(x,w,\tau, h).
\end{align*}
Thus $\overline{\Gamma}_1(x,w,\tau, \cdot)$ is HDD at $\theta_0$ tangentially to $\R^{d_W} \times \mathscr{C}([\varepsilon,1-\varepsilon],\R^{d_q})$.

\bigskip

\textbf{Part 2: The lower bound is HDD}. An analogous argument applies to the lower bound $\underline{\Gamma}_1(x,w,\tau, \cdot)$. This gives
\begin{align*}
	&\frac{\underline{\Gamma}_1(x,w,\tau, \theta_0 + t_m h_m) - \underline{\Gamma}_1(x,w,\tau, \theta_0)}{t_m} \\
	&\qquad \to q(x,w)'\gamma_0'(S_4(x,w,\tau, \beta_0))T_4(x,w,\tau,\beta_0, h_1, 0) + q(x,w)'h_{2}(S_4(x,w,\tau, \beta_0)) \\
	&\qquad \equiv \underline{\Gamma}_{1,\theta_0}'(x,w,\tau, h)
\end{align*}
where
\begin{align*}
	S_3(x,w,\tau, \beta)
		&= \max\left\{\tau - \frac{c}{L(x,w'\beta)}\min\{\tau,1-\tau\},\frac{\tau-1}{L(x,w'\beta)}+1,\varepsilon\right\} \\
	S_4(x,w,\tau, \beta) &= \min\{S_3(x,w,\tau, \beta), 1-\varepsilon\}
\end{align*}
and
\begin{align*}
	T_3(x,w,\tau,\beta_0, h_1, 0)
	&= \lim_{m\to\infty} \frac{S_3(x,w,\tau, \beta_0 + t_m h_{1m}) - S_3(x,w,\tau, \beta_0)}{t_m} \\
	T_4(x,w,\tau,\beta_0, h_1,0)
	&= \lim_{m\to\infty} \frac{S_4(x,w,\tau, \beta_0 + t_m h_{1m}) - S_4(x,w,\tau, \beta_0)}{t_m}.
\end{align*}
We give explicit expressions for these limits in appendix \ref{sec:HDDformulas}.

\bigskip

\textbf{Part 3: Apply the delta method}. We've shown that $\overline{\Gamma}_1(x,w,\tau, \cdot)$ and $\underline{\Gamma}_1(x,w,\tau, \cdot)$ are HDD at $\theta_0$. Moreover, by A\ref{assn:iid} and A\ref{assn:prop score}--A\ref{assn:quant reg}, lemma \ref{lemma:prelim estimators} gives $\sqrt{n} (\widehat{\theta} - \theta_0) \rightsquigarrow \mathbf{Z}_1$. Thus the delta method for HDD functionals (theorem 2.1 in \citealt{FangSantos2014}) gives
\begin{align*}
	\sqrt{n}
	\begin{pmatrix}
		\overline{\Gamma}_1(x,w,\tau, \widehat{\theta}) - \overline{\Gamma}_1(x,w,\tau, \theta_0)\\
		\underline{\Gamma}_1(x,w,\tau, \widehat{\theta}) - \underline{\Gamma}_1(x,w,\tau, \theta_0)
	\end{pmatrix}
	& \xrightarrow{d}
	\begin{pmatrix}
		\overline{\Gamma}_{1,\theta_0}'(x,w,\tau, \textbf{Z}_1)\\
		\underline{\Gamma}_{1,\theta_0}'(x,w,\tau, \textbf{Z}_1)
	\end{pmatrix}
	\equiv \mathbf{Z}_2(x,w,\tau).
\end{align*}
This convergence in uniform in $x\in\{0,1\}$.

Finally, the CQTE bounds are just the difference between certain conditional quantile function bounds. Thus we immediately get
\begin{align*}
	\sqrt{n}
	\begin{pmatrix}
		\widehat{\overline{\text{CQTE}}}^c(\tau \mid w) - \overline{\text{CQTE}}_\varepsilon^c(\tau \mid w)\\
		\widehat{\underline{\text{CQTE}}}^c(\tau \mid w) - \underline{\text{CQTE}}_\varepsilon^c(\tau \mid w)
	\end{pmatrix}
	&\xrightarrow{d}
	\begin{pmatrix}
		\mathbf{Z}_2^{(1)}(1,w,\tau) - \mathbf{Z}_2^{(2)}(0,w,\tau)\\
		\mathbf{Z}_2^{(2)}(1,w,\tau) - \mathbf{Z}_2^{(1)}(0,w,\tau)
	\end{pmatrix}
	\equiv \textbf{Z}_{\text{CQTE}}(w,\tau).
\end{align*}
\end{proof}

\subsection{A Useful Lemma}

The following is a technical lemma that we will use a few times in the upcoming proofs.

\begin{lemma}[Min and Max are Lipschitz]\label{applemma:minmaxlipschitz}
	The following hold for any $(x_1,\ldots,x_n),(y_1,\ldots,y_n) \in \R^n$:
	\begin{align*}
		|\min\{x_1,\ldots,x_n\} - \min\{y_1,\ldots,y_n\}| &\leq \sum_{i=1}^n |x_i - y_i|\\
		|\max\{x_1,\ldots,x_n\} - \max\{y_1,\ldots,y_n\}| &\leq \sum_{i=1}^n |x_i - y_i|.
	\end{align*}
\end{lemma}

\begin{proof}[Proof of lemma \ref{applemma:minmaxlipschitz}]
We proceed by induction over $n \geq 1$. The inequalities trivially hold for $n=1$. First, consider the minimum function and let $n=2$. Consider the case where $x_1 \leq x_2$ and $y_1 \leq y_2$. Then
\begin{align*}
	|\min\{x_1,x_2\} - \min\{y_1,y_2\}| &= |x_1 - y_1|\leq |x_1 -y_1| + |x_2 - y_2|.
\end{align*}
Now consider the case where $x_1 \leq x_2$ and $y_1 \geq y_2$. Then,
\begin{align*}
	\min\{x_1,x_2\} - \min\{y_1,y_2\} &= x_1 - y_2 \leq x_2 - y_2 \leq |x_1 - y_1| + |x_2 - y_2| 
\end{align*}
and
\begin{align*}
	\min\{x_1,x_2\} - \min\{y_1,y_2\} &= x_1 - y_2 \geq x_1 - y_1 \geq -|x_1 - y_1| -|x_2 - y_2|.
\end{align*}
Hence
\[
	|\min\{x_1,x_2\} - \min\{y_1,y_2\}| \leq |x_1 -y_1| + |x_2 - y_2|.
\]
To exhaust all cases, we also consider cases where $(x_1 \geq x_2, y_1 \geq y_2)$ and where $(x_1 \geq x_2, y_1 \leq y_2)$. By symmetry across cases, the Lipschitz inequality for the minimum holds when $n=2$. Now suppose it holds for $n-1$. Then,
\begin{align*}
	|\min\{x_1,\ldots,x_n\} - \min\{y_1,\ldots,y_n\}| &= |\min\{\min\{x_1,\ldots,x_{n-1}\},x_n\} - \min\{\min\{y_1,\ldots,y_{n-1}\},y_n\}| \\
	&\leq |\min\{x_1,\ldots,x_{n-1}\} - \min\{y_1,\ldots,y_{n-1}\}| + |x_n - y_n| \\
	&\leq \sum_{i=1}^{n-1}|x_i - y_i| + |x_n - y_n|\\
	&= \sum_{i=1}^n |x_i - y_i|.
\end{align*}
Therefore it holds for all $n \geq 1$. Noting that $\max\{x_1,\ldots,x_n\} = - \min\{-x_1,\ldots,-x_n\}$, this inequality applies to the maximum as well.
\end{proof}

\subsection{The CATE Bounds}

\begin{proof}[Proof of proposition \ref{prop:CATE convergence}]
\hfill

\medskip

\textbf{Part 1: The upper bound is HDD.} We first show that the mapping
\[
	\overline{\Gamma}_2(x,w, \cdot):\R^{d_W} \times \ell^\infty([\varepsilon,1-\varepsilon],\R^{d_q}) \to \R
\]
is HDD at $\theta_0$ tangentially to $\R^{d_W} \times \mathscr{C}([\varepsilon,1-\varepsilon],\R^{d_q})$. Recall its definition:
\[
	\overline{\Gamma}_2(x,w, \theta) = \int_0^1 \overline{\Gamma}_1(x,w,\tau, \theta) \; d\tau.
\]
We will use the dominated convergence theorem to show that
\begin{equation}\label{eq:defOfGamma2HDD}
	\overline{\Gamma}'_{2,\theta_0}(x,w, h) = \int_0^1 \overline{\Gamma}'_{1,\theta_0}(x,w,\tau, h) \; d\tau.
	\notag
\end{equation}
For $\delta > 0$ let
\[
	\mathcal{G}_\delta = \{\gamma \in \mathcal{G} : \|\gamma - \gamma_0\|_\infty \leq \delta\}
	\qquad \text{and} \qquad
	\Theta_\delta = \mathcal{B}_\delta \times \mathcal{G}_\delta.
\]
To show dominated convergence can be applied, we first show that the mapping $\overline{\Gamma}_1(x,w,\tau, \theta)$ is Lipschitz in $\theta \in  \Theta_\delta$ for some $\delta > 0$. To see this, let $\widetilde{\theta},\theta \in \Theta_\delta$. Then
	\begin{align*}
		&|\overline{\Gamma}_1(x,w,\tau, \widetilde{\theta}) - \overline{\Gamma}_1(x,w,\tau, \theta)|\\
		 &= \left|q(x,w)'\left(\widetilde{\gamma}(S_2(x,w,\tau, \widetilde{\beta})) - \gamma(S_2(x,w,\tau, \widetilde{\beta}))\right) + q(x,w)'\left(\gamma(S_2(x,w,\tau, \widetilde{\beta})) - \gamma(S_2(x,w,\tau, \beta))\right)\right|\\
		&\leq \|q(x,w)\| \cdot \|\widetilde{\gamma} - \gamma\|_\infty + \|q(x,w)'\gamma'(\bar{S}_2)\| \cdot \|S_2(x,w,\tau, \widetilde{\beta}) - S_2(x,w,\tau, \beta)\|.
\end{align*}
The last line follows by a Taylor expansion, where $\bar{S}_2$ is on the line segment connecting $S_2(x,w,\tau, \widetilde{\beta})$ and $S_2(x,w,\tau, \beta)$. By $\gamma \in \mathscr{C}_{m,\nu}^B([\varepsilon,1-\varepsilon])^{d_q}$ for $m \geq 3$ we have $\|\gamma' \|_\infty \leq B$. Hence
\[
	\|q(x,w)'\gamma'(\bar{S}_2))\| \leq \|q(x,w)\| \cdot \|\gamma'(\bar{S}_2)\| \leq \|q(x,w)\| B.
\]
Next,
\begin{align*}
		&|S_2(x,w,\tau, \widetilde{\beta}) - S_2(x,w,\tau, \beta)| \\
		&= |\max\{S_1(x,w,\tau, \widetilde{\beta}),\varepsilon\} - \max\{S_1(x,w,\tau, \beta),\varepsilon\}|\\
		&\leq |S_1(x,w,\tau, \widetilde{\beta}) - S_1(x,w,\tau, \beta)|\\
		&\leq \left|\tau + c\min\{\tau,1-\tau\} \frac{1}{L(x,w'\widetilde{\beta})} - \left(\tau + c\min\{\tau,1-\tau\} \frac{1}{L(x,w'\beta)}\right)\right|
		 + \left|\frac{\tau}{L(x,w'\widetilde{\beta})} - \frac{\tau}{L(x,w'\beta)}\right|\\
		&= (\tau + c \min\{\tau,1-\tau\}) \left|\frac{1}{L(x,w'\widetilde{\beta})} - \frac{1}{L(x,w'\beta)}\right|\\
		&\leq (\tau + c \min\{\tau,1-\tau\})\sup_{ \beta \in \mathcal{B}_\delta}\left\|\frac{L_\beta(x,w'\beta)}{L(x,w'\beta)^2}\right\|\|\widetilde{\beta} - \beta\|.
\end{align*}
The second and third lines follow from lemma \ref{applemma:minmaxlipschitz}. The last line follows by a Taylor expansion. To see that $\sup_{ \beta \in \mathcal{B}_\delta}\left\| L_\beta(x,w'\beta) / L(x,w'\beta)^2 \right\| <\infty$, write
\begin{align*}
	\sup_{ \beta \in \mathcal{B}_\delta}\left\|\frac{L_\beta(x,w'\beta)}{L(x,w'\beta)^2}\right\| & =  \sup_{ \beta \in \mathcal{B}_\delta}\left\|\frac{(2x-1) F'(w'\beta)w}{xF(w'\beta)^2 + (1-x)(1-F(w'\beta))^2}\right\|\\
	&\leq \|w\| \sup_{a \in \R} |F'(a)| \frac{1}{x F(\inf_{\beta \in \beta} w'\beta)^2 + (1-x)(1 - F(\sup_{\beta \in \beta} w'\beta))^2}\\
	&<\infty.
\end{align*}
The last line follows since $\{w'\beta: \beta \in \mathcal{B}_\delta\}$ is bounded for fixed $w$, and since $F'(a)$ is uniformly bounded by assumption A\ref{assn:prop score}.3. Thus
\begin{multline*}
	|\overline{\Gamma}_1(x,w,\tau, \widetilde{\theta}) - \overline{\Gamma}_1(x,w,\tau, \theta)| \\
	\leq \|q(x,w)\| \cdot \|\widetilde{\gamma} - \gamma\|_\infty + \|q(x,w)\| B (\tau + c\min\{\tau,1-\tau\}) \sup_{ \beta \in \mathcal{B}_\delta}\left\|\frac{L_\beta(x,w'\beta)}{L(x,w'\beta)^2}\right\| \|\widetilde{\beta} - \beta\|.
\end{multline*}
Hence $\overline\Gamma_1(x,w,\tau, \theta)$ is Lipschitz in $\theta$. Therefore,
\[
	\frac{\overline{\Gamma}_1(x,w,\tau, \theta_0 + t_m h_m) - \overline{\Gamma}_1(x,w,\tau, \theta_0)}{t_m}
\]
is dominated by 
\begin{align*}
	&\|q(x,w)\| \cdot \|h_{2m}\|_\infty + \|q(x,w)\|B(\tau + c\min\{\tau,1-\tau\}) \sup_{ \beta \in \mathcal{B}_\delta}\left\|\frac{L_\beta(x,w'\beta)}{L(x,w'\beta)^2}\right\| \|h_{1m}\|\\
	&\leq \|q(x,w)\| \; (\|h_{2}\|_\infty + \lambda) + \|q(x,w)\|B(\tau + c\min\{\tau,1-\tau\}) \sup_{ \beta \in \mathcal{B}_\delta}\left\|\frac{L_\beta(x,w'\beta)}{L(x,w'\beta)^2}\right\| (\|h_{1}\| + \lambda) \\
	&<\infty.
\end{align*}
In the second line $\lambda > 0$ is a constant that can be made arbitrarily small by choosing $m$ sufficiently large, since $h_m$ converges to $h$. Moreover, note that this dominating function is integrable over $\tau \in (0,1)$. Thus we can apply the dominated convergence theorem to show that
\begin{align*}
		\frac{\overline{\Gamma}_2(x,w, \theta_0 + t_m h_m) - \overline{\Gamma}_2(x,w, \theta_0)}{t_m} &= \int_0^1 \left(\frac{\overline{\Gamma}_1(x,w,\tau, \theta_0 + t_m h_m) - \overline{\Gamma}_1(x,w,\tau, \theta_0)}{t_m}\right) \; d\tau\\
		&\to \int_0^1 \overline{\Gamma}'_{1,\theta_0}(x,w,\tau, h) \; d\tau \\
		&= q(x,w)'\int_0^1 h_{2}(S_2(x,w,\tau, \beta_0)) \; d\tau \\
		&\quad + q(x,w)'\int_0^1\gamma_0'(S_2(x,w,\tau, \beta_0)) T_2(x,w,\tau,\beta_0, h_1, 0) \; d\tau\\
		&\equiv \overline{\Gamma}'_{2,\theta_0}(x,w, h).
\end{align*}

\textbf{Part 2: The lower bound is HDD.} We can similarly show that
\begin{align*}
		\frac{\underline{\Gamma}_2(x,w, \theta_0 + t_m h_m) - \underline{\Gamma}_2(x,w, \theta_0)}{t_m} &\to  q(x,w)'\int_0^1 h_{2}(S_4(x,w,\tau, \beta_0)) \; d\tau\\
		&\qquad + q(x,w)'\int_0^1\gamma_0'(S_4(x,w,\tau, \beta_0))T_4(x,w,\tau, h_1, 0) \; d\tau\\
		&\equiv \underline{\Gamma}'_{2,\theta_0}(x, w, h).
\end{align*}

\textbf{Part 3: Apply the delta method.} The functional delta method for HDD functionals now implies that, uniformly in $x \in \{0,1\}$,
\begin{align*}
	\sqrt{n}
	\begin{pmatrix}
		\overline{\Gamma}_2(x,w, \widehat{\theta}) - \overline{\Gamma}_2(x,w, \theta_0)\\
		\underline{\Gamma}_2(x,w, \widehat{\theta}) - \underline{\Gamma}_2(x,w, \theta_0)
	\end{pmatrix}
	&\xrightarrow{d} 
	\begin{pmatrix}
		\overline{\Gamma}'_{2,\theta_0}(x,w, \textbf{Z}_1)\\
		\underline{\Gamma}'_{2,\theta_0}(x,w, \textbf{Z}_1)
	\end{pmatrix}
	\equiv \mathbf{Z}_3(x,w)
\end{align*}
and hence
\begin{align*}
	\sqrt{n}
	\begin{pmatrix}
		\widehat{\overline{\text{CATE}}}^c(w) - \overline{\text{CATE}}_\varepsilon^c(w)\\
		\widehat{\underline{\text{CATE}}}^c(w) - \underline{\text{CATE}}_\varepsilon^c(w)
	\end{pmatrix}
	&\xrightarrow{d}
	\begin{pmatrix}
		\mathbf{Z}_3^{(1)}(1,w) - \mathbf{Z}_3^{(2)}(0,w)\\
		\mathbf{Z}_3^{(2)}(1,w) - \mathbf{Z}_3^{(1)}(0,w)
	\end{pmatrix} \equiv \mathbf{Z}_{\text{CATE}}(w).
\end{align*}
\end{proof}

\subsection{The ATE Bounds}

\begin{proof}[Proof of theorem \ref{thm:ATE convergence}]
\hfill

\medskip

\textbf{Part 1: The expectation upper bound}. Write
\begin{align*}
		\sqrt{n}(\widehat{\overline{E}}_x^c - \overline{E}_{x,\varepsilon}^c)
		&= \sqrt{n}\left(\avg \overline{\Gamma}_2(x,W_i, \widehat{\theta}) - \int_\mathcal{W} \overline{\Gamma}_2(x,w, \theta_0)\; dF_W(w)\right)\\
		&= \sqrt{n}\left(\avg (\overline{\Gamma}_2(x,W_i, \widehat{\theta}) - \overline{\Gamma}_2(x,W_i, \theta_0))  - \int_\mathcal{W} (\overline{\Gamma}_2(x,w, \widehat{\theta}) - \overline{\Gamma}_2(x,w, \theta_0)) \; dF_W(w)\right)\\
		&\quad + \frac{1}{\sqrt{n}}\sum_{i=1}^n \Big( \overline{\Gamma}_2(x,W_i, \theta_0) - \Exp[\overline{\Gamma}_2(x,W, \theta_0)] \Big) + \sqrt{n} \big( \overline\Gamma_3(x, \widehat{\theta}) - \overline\Gamma_3(x, \theta_0) \big).
\end{align*}
There are three terms here. We'll show that the first is $o_p(1)$ and that the second and third contribute to the asymptotic distribution.

\medskip

\textbf{\emph{Step 1.}} We'll begin by showing that the first term is $o_p(1)$. For some $\delta > 0$, consider the class of functions
\[
	\overline{\mathcal{F}} = \left\{\overline{\Gamma}_2(x,w, \theta): \theta \in  \Theta_\delta \right\}.
\]
As in the proof of proposition \ref{prop:CATE convergence}, we will show that $\overline{\Gamma}_2(x,w, \theta)$ is Lipschitz in $\theta$. Let $\widetilde{\theta},\theta \in \Theta_\delta$. Then
\begin{align*}
	&|\overline{\Gamma}_2(x,w, \widetilde{\theta}) - \overline{\Gamma}_2(x,w, \theta)| \\
	&= \left|\int_0^1 \overline{\Gamma}_1(x,w,\tau, \widetilde{\theta}) \; d\tau - \int_0^1 \overline{\Gamma}_1(x,w,\tau, \theta) \; d\tau\right| \\
	&\leq \|q(x,w)\| \int_0^1 \|\widetilde{\gamma} - \gamma\|_\infty  \; d\tau
	 + \|q(x,w)\|B\int_0^1(\tau + c\min\{\tau,1-\tau\}) \; d\tau \sup_{ \beta \in \mathcal{B}_\delta}\left\|\frac{L_\beta(x,w'\beta)}{L(x,w'\beta)^2}\right\| \|\widetilde{\beta} - \beta\|\\
	&= \|q(x,w)\|\left(1 +  \frac{B (2+c)}{4} \sup_{ \beta \in \mathcal{B}_\delta}\left\|\frac{L_\beta(x,w'\beta)}{L(x,w'\beta)^2}\right\|\right) \left(\|\widetilde{\gamma} - \gamma\|_\infty + \|\widetilde{\beta} - \beta\|\right)\\
	&\equiv K(w) \|\widetilde{\theta} - \theta\|_\Theta.
\end{align*}
The second line follows by our derivations in the proof of proposition \ref{prop:CATE convergence}. In the last line we let $\|\theta\|_\Theta = \|\beta\| + \|\gamma\|_\infty$ and defined
\begin{align*}
	K(w) &= \|q(x,w)\|\left(1 +  \frac{B(2+c)}{4} \sup_{ \beta \in \mathcal{B}_\delta}\left\|\frac{L_\beta(x,w'\beta)}{L(x,w'\beta)^2}\right\|\right).
\end{align*}
Assumption A\ref{assn:propscoremoment} says that
\[
	\Exp \left( \sup_{\beta\in\mathcal{B}_\delta}\left\|\frac{L_\beta(x,W'\beta)}{L(x,W'\beta)^2}\right\|^4\right) <\infty
\]
and A\ref{assn:quant reg}.4 says $\Exp(\| q(x,W) \|^4) < \infty$. These assumptions imply that $K(W)$ has a bounded second moment:
\begin{align*}
	\Exp[K(W)^2] &= \Exp \left[\|q(x,W)\|^2\left(1 + \frac{B(2+c)}{4} \sup_{ \beta \in \mathcal{B}_\delta}\left\|\frac{L_\beta(x,W'\beta)}{L(x,W'\beta)^2}\right\|\right)^2\right]\\
	&\leq \Exp \left[\|q(x,W)\|^4\right]^{1/2} \times \Exp \left[\left(1 +  \frac{B(2+c)}{4} \sup_{ \beta \in \mathcal{B}_\delta}\left\|\frac{L_\beta(x,W'\beta)}{L(x,W'\beta)^2}\right\|\right)^4\right]^{1/2}\\
	&<\infty
\end{align*}
where the second line follows by the Cauchy-Schwarz inequality. Thus $\overline{\Gamma}_2(x,w, \theta)$ is Lipschitz in $\theta$. This lets us apply theorem 2.7.11 in \cite{VaartWellner1996} to see that the bracketing number $N_{[\cdot]}(2 \, \epsilon \, \Exp[K(W)^2]^{1/2},\overline{\mathcal{F}},L_2(\Prob))$ is bounded above by 
\[
	N(\epsilon,\mathcal{B}_\delta \times \mathcal{G}_\delta, \|\cdot\|_\Theta) \leq N(\epsilon,\mathcal{B}_\delta, \|\cdot\|) + N(\epsilon,\mathcal{G}_\delta, \|\cdot\|_\infty).
\]
By example 19.7 of \cite{Vaart2000}, $N(\epsilon,\mathcal{B}_\delta, \|\cdot\|) \lesssim \epsilon^{-d_W}$ for small enough $\epsilon$. By theorem 2.7.1 in \cite{VaartWellner1996},
\[
	N(\epsilon,\mathcal{G}_\delta, \|\cdot\|_\infty) \lesssim \exp\left(\epsilon^{-\frac{1}{m+\nu}}\right).
\]
So
\[
	\int_0^\delta \sqrt{\log N_{[\cdot]}(2 \, \epsilon \, \Exp[K(W)^2]^{1/2},\overline{\mathcal{F}},L_2(\Prob))} \; d\epsilon <\infty.
\]
Hence $\overline{\mathcal{F}}$ is Donsker.

By convergence of $\widehat{\theta}$ to $\theta_0$ (lemma \ref{lemma:prelim estimators}) and the Lipschitz property of $\overline{\Gamma}_2$, we have
\begin{align*}
	\int_\mathcal{W} \left|\overline{\Gamma}_2(x,w, \widehat{\theta}) - \overline{\Gamma}_2(x,w, \theta_0)\right|^2 dF_W(w)
	&\leq \Exp[K(W)^2] \cdot \|\widehat{\theta} - \theta_0\|_\Theta^2 \\
	&= o_p(1).
\end{align*}
Therefore, by lemma 19.24 in \cite{Vaart2000},
$$\sqrt{n}\left(\frac{1}{n}\sum_{i=1}^n (\overline{\Gamma}_2(x,W_i, \widehat{\theta}) - \overline{\Gamma}_2(x,W_i, \theta_0))  - \int_\mathcal{W} (\overline{\Gamma}_2(x,w, \widehat{\theta}) - \overline{\Gamma}_2(x,w, \theta_0)) \; dF_W(w)\right) = o_p(1).$$

\medskip

\textbf{\emph{Step 2.}} Next consider the second term. First note that
\begin{align*}
	\Exp[ \overline{\Gamma}_2(x,W,\theta_0)^2 ]
		&= \Exp \left[ \left| \int_0^1 q(x,W)' \gamma_0(S_2(x,W,\tau,\beta_0)) \; d\tau \right|^2 \right] \\
		&\leq \Exp \left[ \left( \int_0^1 \| q(x,W) \| \cdot \| \gamma_0(S_2(x,W,\tau,\beta_0)) \| \; d\tau \right)^2 \right] \\
		&\leq \Exp \left[ \left( \int_0^1 \| q(x,W) \| \sup_{u \in [\varepsilon,1-\varepsilon]} \| \gamma_0(u) \| \; d\tau \right)^2 \right] \\
		&\leq \Exp( \| q(x,W) \|^2 ) B^2 \\
		&< \infty.
\end{align*}
The first line follows be definition of $\overline{\Gamma}_2$. The second line follows by the Cauchy-Schwarz inequality. The third line follows the fact that $S_2$ lies between $\varepsilon$ and $1-\varepsilon$.  The fourth line follows by  A\ref{assn:quant reg regularity}.  The last line follows by A\ref{assn:quant reg}.4.

This result lets us apply a CLT to the second term. Combining that with the fact that the influence functions for the two first step estimators are Donsker (lemmas \ref{applemma:prop score estimation} and \ref{applemma:quantile estimation}) we get
\begin{align*}
	\sqrt{n}
	\begin{pmatrix}
		\widehat{\theta} - \theta_0\\
		\displaystyle \frac{1}{n}\sum_{i=1}^n \Big(\Gamma_2(x,W_i, \theta_0) - \Exp \left[\Gamma_2(x,W, \theta_0)\right] \Big)
	\end{pmatrix}
	&\rightsquigarrow
	\begin{pmatrix}
		\mathbf{Z}_1\\
		\widetilde{\mathbf{Z}}_4(x)
	\end{pmatrix},
\end{align*}	
a mean-zero Gaussian process in $\R^{d_W}\times \ell^\infty([\varepsilon,1-\varepsilon],\R^{d_q}) \times \R^2$. Notice here we use $\Gamma_2 = (\overline{\Gamma}_2, \underline{\Gamma}_2)$, not just $\overline{\Gamma}_2$, as preparation for parts 2 and 3.

\medskip

\textbf{\emph{Step 3.}} Next, consider the third term. For this step we'll show that the mapping $\overline\Gamma_3(x, \theta)$ is HDD at $\theta_0$ tangentially to $\R^{d_W} \times \mathscr{C}([\varepsilon,1-\varepsilon],\R^{d_q})$. This will let us apply the delta method for HDD functionals in the last step. To see that this functional is HDD, note that
\begin{align*}
	\frac{\overline{\Gamma}_3(x, \theta_0 + t_m h_m) - \overline{\Gamma}_3(x, \theta_0)}{t_m} &= \int_\mathcal{W} \int_0^1 \frac{\overline{\Gamma}_1(x,w,\tau, \theta_0 + t_m h_m) - \overline{\Gamma}_1(x,w,\tau, \theta_0)}{t_m} \; d\tau \; dF_W(w).
\end{align*}
By proposition \ref{prop:CATE convergence}, for $m$ large enough,
$$\frac{\overline{\Gamma}_1(x,w,\tau, \theta_0 + t_m h_m) - \overline{\Gamma}_1(x,w,\tau, \theta_0)}{t_m}$$
is dominated by 
\[
	\|q(x,w)\|(\|h_{2}\|_\infty + \lambda) + \|q(x,w)\|B(\tau + c\min\{\tau,1-\tau\}) \sup_{ \beta \in \mathcal{B}_\delta}\left\|\frac{L_\beta(x,w'\beta)}{L(x,w'\beta)^2}\right\| (\|h_{1}\| + \lambda).
\]
This expression which has finite integral over $(\tau,w) \in (0,1)\times \mathcal{W}$ by $\|h\|_\Theta <\infty$, $\Exp(\|q(x,W)\|^4) < \infty$, and
\[
	\Exp \left(\sup_{ \beta \in \mathcal{B}_\delta}\left\|\frac{L_\beta(x,W, \beta)}{L(x,W'\beta)^2}\right\|^4\right) <\infty.
\]
So we can apply dominated convergence to see that $\overline{\Gamma}_3$ is HDD:
\begin{align*}
	\frac{\overline{\Gamma}_3(x, \theta_0 + t_m h_m) - \overline{\Gamma}_3(x, \theta_0)}{t_m} &\to \int_\mathcal{W} \left[ q(x,w)'\int_0^1 h_{2}(S_2(x,w,\tau, \beta_0)) \; d\tau \right. \\
	&\qquad \left. +  q(x,w)'\int_0^1\gamma_0'(S_2(x,w,\tau, \beta_0))T_2(x,w,\tau,\beta_0, h_1, 0) \; d\tau\right] \; dF_W(w) \\
	&\equiv \overline{\Gamma}_{3,\theta_0}'(x, h).
\end{align*} 

\medskip

\textbf{\emph{Step 4.}} Finally, putting all the previous steps together and applying the delta method for HDD functionals gives that, uniformly in $x \in \{0,1\}$,
\begin{align*}
	\sqrt{n}(\widehat{\overline{E}}_x^c - \overline{E}_{x,\varepsilon}^c)
	&= o_p(1) +  \frac{1}{\sqrt{n}}\sum_{i=1}^n \left(\overline{\Gamma}_2(x,W_i, \theta_0) - \Exp[\overline{\Gamma}_2(x,W, \theta_0)]\right) + \sqrt{n} (\overline\Gamma_3(x, \widehat{\theta}) - \overline\Gamma_3(x, \theta_0))\\[0.5em]
	&\xrightarrow{d} \widetilde{\mathbf{Z}}_4^{(1)}(x) + \overline{\Gamma}_{3,\theta_0}'(x, \mathbf{Z}_{1}) \\[0.6em]
	&\equiv \mathbf{Z}_4^{(1)}(x).
\end{align*}

\medskip

\textbf{Part 2: The expectation lower bound.} An identical argument can be applied to show that
\begin{align*}
	\sqrt{n}(\widehat{\underline{E}}_x^c - \underline{E}_x^c)
	&\xrightarrow{d} \widetilde{\mathbf{Z}}_4^{(2)}(x) + \underline{\Gamma}_{3,\theta_0}'(x, \mathbf{Z}_{1}) \\[0.6em]
	&\equiv \mathbf{Z}_4^{(2)}(x)
\end{align*}
where
\[
	\underline{\Gamma}_{3,\theta_0}'(x, h) = \int_{\mathcal{W}}\underline{\Gamma}_{2,\theta_0}'(x,w,h)\; dF_W(w).
\]

\medskip

\textbf{Part 3: Putting them together.} Finally, the analysis in parts 1 and 2 can be combined to obtain joint convergence:
\begin{align*}
	\sqrt{n}
	\begin{pmatrix}
		\widehat{\overline{\text{ATE}}}^c - \overline{\text{ATE}}_\varepsilon^c \\
		\widehat{\underline{\text{ATE}}}^c - \underline{\text{ATE}}_\varepsilon^c
	\end{pmatrix}
	&\xrightarrow{d}
	\begin{pmatrix}
		\mathbf{Z}_4^{(1)}(1) - \mathbf{Z}_4^{(2)}(0)\\
		\mathbf{Z}_4^{(2)}(1) - \mathbf{Z}_4^{(1)}(0)
	\end{pmatrix} \equiv \mathbf{Z}_{\text{ATE}}.
\end{align*}
\end{proof}

\subsection{The ATT Bounds}

\begin{proof}[Proof of proposition \ref{prop:ATT convergence}]
Note that in this proof we'll use several of the results we derived in the proof of proposition \ref{thm:ATE convergence}. By $\var(Y \indicator(X=x))<\infty$ and $\var(\indicator(X=x)) <\infty$, we have that
\begin{align*}
	\sqrt{n} \Big( \widehat{\Exp}(Y \mid X=x) - \Exp(Y \mid X=x) \Big)
	&= \frac{1}{p_x}\sqrt{n}\left(\frac{1}{n} \sum_{i=1}^n Y_i \indicator(X_i = x)- \Exp(Y\indicator(X=x)) \right)  \\
	&\qquad \qquad - \frac{\Exp(Y \mid X=x)}{p_x}\sqrt{n}\left(\frac{1}{n} \sum_{i=1}^n  \indicator(X_i = x) - p_x \right) + o_p(1)\\	
&	\xrightarrow{d} \textbf{Z}_{\Exp(Y \mid X=x)},
\end{align*}
a mean-zero Gaussian variable.
Also, $\sqrt{n}(\widehat{p}_x - p_x) \xrightarrow{d} \textbf{Z}_{p_x}$, another mean-zero Gaussian variable. As before the influence functions are Donsker and hence we can stack them to obtain
\begin{align*}
	\sqrt{n}
	\begin{pmatrix}
		\widehat{\theta} - \theta_0\\
		\displaystyle \frac{1}{n}\sum_{i=1}^n \left(\Gamma_2(x,W_i, \theta_0) - \Exp \left[\Gamma_2(x,W, \theta_0)\right]\right)	\\[1.5em]
		\widehat{\Exp}(Y \mid X=x) - \Exp(Y \mid X=x) \\[0.5em]
		\widehat{p}_x - p_x
	\end{pmatrix}
	&\rightsquigarrow
	\begin{pmatrix}
		\mathbf{Z}_1\\
		\widetilde{\mathbf{Z}}_4(x)\\
		\textbf{Z}_{\Exp(Y \mid X=x)}\\
		\textbf{Z}_{p_x}
	\end{pmatrix},
\end{align*}	
a mean-zero Gaussian process in $\R^{d_W} \times \ell^{\infty}([\varepsilon,1-\varepsilon],\R^{d_q}) \times \R^4$, for $x \in \{0,1\}$ and whose covariance kernel can be calculated as in the proof of theorem \ref{thm:ATE convergence}. By the delta method,
\begin{align*}
	&\sqrt{n}
	\begin{pmatrix}
		\widehat{\overline{\text{ATT}}}^c - \overline{\text{ATT}}_\varepsilon^c\\[0.5em]
		\widehat{\underline{\text{ATT}}}^c - \underline{\text{ATT}}_\varepsilon^c
	\end{pmatrix} \\[0.5em]
	&\xrightarrow{d}
	\begin{pmatrix}
	\displaystyle
		\textbf{Z}_{\Exp(Y \mid X=1)} - \frac{\textbf{Z}_4^{(2)}(0)}{p_1} + \frac{p_0}{p_1}\textbf{Z}_{\Exp(Y \mid X=0)} + \frac{\Exp(Y \mid X=0)}{p_1}\textbf{Z}_{p_0} + \frac{\underline{E}_{0,\varepsilon}^c - \Exp(Y \mid X=0) p_0}{p_1^2}\mathbf{Z}_{p_1}\\[2em]
	\displaystyle
		\textbf{Z}_{\Exp(Y \mid X=1)} - \frac{\textbf{Z}_4^{(1)}(0)}{p_1} + \frac{p_0}{p_1}\textbf{Z}_{\Exp(Y \mid X=0)} + \frac{\Exp(Y \mid X=0)}{p_1}\textbf{Z}_{p_0} + \frac{\overline{E}_{0,\varepsilon}^c - \Exp(Y \mid X=0) p_0}{p_1^2}\mathbf{Z}_{p_1}
	\end{pmatrix} \label{eq:ATTlimdist} \\
	&\equiv \mathbf{Z}_{\text{ATT}},
\end{align*}
a random vector in $\R^2$.
\end{proof}

\section{Estimating the Analytical Hadamard Directional Derivatives}\label{sec:HDDformulas}

In this section we give formulas for our estimators of the analytical Hadamard directional derivatives (HDD) used in the CQTE, CATE, ATE, and ATT functionals. In these estimators we use a tuning parameter $\kappa_n \geq 0$. This parameter acts as a slackness value, which lets us estimate when certain equalities hold in the population, but which may not hold exactly in finite samples. We assume $\kappa_n \to 0$ and $n \kappa_n^2 \to \infty$ as $n \to \infty$.

We estimate $\Gamma_{1,\theta_0}'(x,w,\tau,h)$ by
\[
	\begin{pmatrix}
	\widehat{\overline{\Gamma}}_{1,\theta_0}'( x,w,\tau,h) \\[0.5em]
	\widehat{\underline{\Gamma}}_{1,\theta_0}'( x,w,\tau,h)
	\end{pmatrix}
	=
	\begin{pmatrix}
		q(x,w)'h_{2}(S_2(x,w,\tau, \widehat\beta)) + q(x,w)'\widehat{\gamma}'(S_2(x,w,\tau, \widehat\beta)) \cdot T_2(x,w,\tau,\widehat\beta, h_1, \kappa_n) \\[0.5em]
		q(x,w)'h_{2}(S_4(x,w,\tau, \widehat\beta)) + q(x,w)'\widehat{\gamma}'(S_4(x,w,\tau, \widehat\beta)) \cdot T_4(x,w,\tau,\widehat\beta, h_1, \kappa_n)
	\end{pmatrix}.
\]
Note that $q(x,w)'$ refers to the transpose of $q(x,w)$ while $\widehat{\gamma}'(\cdot)$ refers to our estimator of the derivative of $\gamma(\cdot)$, defined in equation \eqref{eq:QRderivativeEstimator} on page \pageref{eq:QRderivativeEstimator}.
We defined the functions $S_2$ and $S_4$ in the proof of proposition \ref{prop:CQTE convergence}. We define $T_2$ and $T_4$ below. Estimate $\Gamma'_{2,\theta_0}(x,w,h)$ by
\[
	\begin{pmatrix}
	\widehat{\overline{\Gamma}}_{2,\theta_0}'( x,w,h) \\[0.5em]
	\widehat{\underline{\Gamma}}_{2,\theta_0}'( x,w,h)
	\end{pmatrix}
	=
	\begin{pmatrix}
	\int_0^1 \widehat{\overline{\Gamma}}_{1,\theta_0}'( x,w,\tau,h) \; d\tau \\[0.5em]
	\int_0^1 \widehat{\underline{\Gamma}}_{1,\theta_0}'( x,w,\tau,h) \; d\tau
	\end{pmatrix}.
\]
Estimate $\Gamma_{3,\theta_0}(x,h)$ by
\[
	\begin{pmatrix}
	\widehat{\overline{\Gamma}}_{3,\theta_0}'( x,h) \\[0.5em]
	\widehat{\underline{\Gamma}}_{3,\theta_0}'( x,h)
	\end{pmatrix}
	=
	\begin{pmatrix}
	\avg \widehat{\overline{\Gamma}}_{2,\theta_0}'( x,W_i,h) \\[0.5em]
	\avg \widehat{\underline{\Gamma}}_{2,\theta_0}'( x,W_i,h)
	\end{pmatrix}.
\]
Throughout the paper we will also use the vector notation $\widehat{\Gamma}_{j,\theta_0}' = (\widehat{\overline{\Gamma}}_{j,\theta_0}', \widehat{\underline{\Gamma}}_{j,\theta_0}' )$ and $\Gamma_{j,\theta_0}' = (\overline{\Gamma}_{j,\theta_0}', \underline{\Gamma}_{j,\theta_0}' )$ for $j=1,2,3$.
Next we define the functions $T_2$ and $T_4$. To do this, we'll also define functions $T_1$ and $T_3$.

\subsection*{The function $T_1$}

We first define $T_1(x,w,\tau,\beta, h_1, \kappa_n)$. We do this by splitting it into seven different mutually exclusive cases. This lets us write
\begin{equation}\label{eq:T1summationRepresentation}
	T_1(x,w,\tau,\beta, h_1, \kappa_n) = \sum_{j=1}^7 T_{1,j}(x,w,\tau,\beta,h_1) \cdot \indicator_{1,j}(x,w,\tau,\beta, \kappa_n).
\end{equation}
For these cases, it is helpful to recall our notation
\[
	L(x,w'\beta) = F(w'\beta)^x(1-F(w'\beta))^{1-x}.
\]
These seven cases are defined as follows:
\begin{enumerate}
\item Let 
\[
	\indicator_{1,1}(x,w,\tau,\beta, \kappa_n) = \ind\left\{\tau + \frac{c}{L(x,w'\beta)} \min\{\tau,1-\tau\}
	< \min \left\{ \frac{\tau}{L(x,w'\beta)},1-\varepsilon \right\} - \kappa_n\right\}
\]
and
\[
	T_{1,1}(x,w,\tau,\beta, h_1) = - c\min\{\tau,1-\tau\} \frac{L_\beta(x,w'\beta)'h_1}{L(x,w'\beta)^2}.
\]

\item 
Let 
\[
	\indicator_{1,2}(x,w,\tau,\beta, \kappa_n) = \ind\left\{\frac{\tau}{L(x,w'\beta)} < \min \left\{ \tau + \frac{c}{L(x,w'\beta)}\min\{\tau,1-\tau\},1-\varepsilon \right\} - \kappa_n\right\}
\]
and
\[
	T_{1,2}(x,w,\tau,\beta, h_1) = - \tau \frac{L_\beta(x,w'\beta)'h_1}{L(x,w'\beta)^2}.
\]

\item 
Let 
\[
	\indicator_{1,3}(x,w,\tau,\beta, \kappa_n) = \ind\left\{ 1-\varepsilon 
	< \min \left\{ \tau + \frac{c}{L(x,w'\beta)}\min\{\tau,1-\tau\}, \frac{\tau}{L(x,w'\beta)} \right\} - \kappa_n \right\}
\]
and
\[
	T_{1,3}(x,w,\tau,\beta, h_1) = 0.
\]

\item Let 
\begin{align*}
	\indicator_{1,4}(x,w,\tau,\beta, \kappa_n) &= \ind\left\{ \left|\tau + \frac{c}{L(x,w'\beta)}\min\{\tau,1-\tau\} - \frac{\tau}{L(x,w'\beta)}\right| \leq \kappa_n \right\}\\
	&\cdot \ind\left\{\max\left\{\tau + \frac{c}{L(x,w'\beta)}\min\{\tau,1-\tau\},\frac{\tau}{L(x,w'\beta)}\right\} < 1-\varepsilon - \kappa_n\right\}
\end{align*}
and
\[
	T_{1,4}(x,w,\tau,\beta, h_1) = \min\left\{- c\min\{\tau,1-\tau\} \frac{L_\beta(x,w'\beta)'h_1}{L(x,w'\beta)^2}, - \tau \frac{L_\beta(x,w'\beta)'h_1}{L(x,w'\beta)^2} \right\} .
\]

\item Let
\begin{align*}
	\indicator_{1,5}(x,w,\tau,\beta, \kappa_n) &= \ind\left\{ \left|\tau + \frac{c}{L(x,w'\beta)}\min\{\tau,1-\tau\} -(1-\varepsilon)\right| \leq \kappa_n \right\}\\
	&\cdot \ind\left\{\max\left\{\tau + \frac{c}{L(x,w'\beta)}\min\{\tau,1-\tau\},1-\varepsilon\right\} < \frac{\tau}{L(x,w'\beta)} - \kappa_n \right\}
\end{align*}
and
\[
	T_{1,5}(x,w,\tau,\beta, h_1) = \min\left\{- c\min\{\tau,1-\tau\} \frac{L_\beta(x,w'\beta)'h_1}{L(x,w'\beta)^2}, 0 \right\}.
\]

\item Let 
\begin{align*}
	\indicator_{1,6}(x,w,\tau,\beta, \kappa_n) &= \ind\left\{ \left|\frac{\tau}{L(x,w'\beta)} -(1-\varepsilon)\right| \leq \kappa_n \right\}\\
	&\cdot \ind\left\{\max\left\{\frac{\tau}{L(x,w'\beta)} , 1-\varepsilon \right\} < \tau + \frac{c}{L(x,w'\beta)}\min\{\tau,1-\tau\} - \kappa_n \right\}
\end{align*}
and
\[
	T_{1,6}(x,w,\tau,\beta, h_1) = \min\left\{- \tau \frac{L_\beta(x,w'\beta)'h_1}{L(x,w'\beta)^2},0\right\}.
\]

\item Let $\indicator_{1,7}(x,w,\tau,\beta, \kappa_n)$ equal 1 if at least two of the three following hold
\begin{align*}
		\left|\tau + \frac{c}{L(x,w'\beta)}\min\{\tau,1-\tau\} -(1-\varepsilon)\right| &\leq \kappa_n \\
		\left|\frac{\tau}{L(x,w'\beta)} -(1-\varepsilon)\right| &\leq \kappa_n\\
		\left|\tau + \frac{c}{L(x,w'\beta)}\min\{\tau,1-\tau\} -\frac{\tau}{L(x,w'\beta)}\right| &\leq \kappa_n
\end{align*}
and zero otherwise. Let
\[
	T_{1,7}(x,w,\tau,\beta, h_1) = \min\left\{- c\min\{\tau,1-\tau\} \frac{L_\beta(x,w'\beta)'h_1}{L(x,w'\beta)^2}, - \tau \frac{L_\beta(x,w'\beta)'h_1}{L(x,w'\beta)^2},0 \right\}.
\]

\end{enumerate}

\subsection*{The function $T_2$}

Next we define $T_2(x,w,\tau,\beta,h_1,\kappa_n)$ in terms of $T_1(x,w,\tau,\beta,h_1,\kappa_n)$ and two indicator functions. Recall from the proof of proposition \ref{prop:CQTE convergence} that
\[
	S_1(x,w,\tau, \beta) = \min\left\{\tau + \frac{c}{L(x,w'\beta)}\min\{\tau,1-\tau\},\frac{\tau}{L(x,w'\beta)},1-\varepsilon\right\}.
\]
Then
\begin{align*}
	T_2(x,w,\tau,\beta,h_1,\kappa_n)
	 &= T_1(x,w,\tau,\beta,h_1,\kappa_n) \cdot \ind \Big( S_1(x,w,\tau,\beta) > \varepsilon  + \kappa_n \Big) \\
	&\qquad + \max\left\{T_1(x,w,\tau,\beta, h_1, \kappa_n),0\right\} \cdot \ind \Big( \left|S_1(x,w,\tau,\beta) - \varepsilon\right| \leq \kappa_n \Big).
\end{align*}

\subsection*{The function $T_3$}

Next we define $T_3(x,w,\tau,\beta, h_1, \kappa_n)$. Again we split this into seven cases, which lets us write this function as
\begin{equation}\label{eq:T3summationRepresentation}
	T_3(x,w,\tau,\beta, h_1, \kappa_n) = \sum_{j=1}^7 T_{3,j}(x,w,\tau,\beta,h_1) \cdot \indicator_{3,j}(x,w,\tau,\beta, \kappa_n).
\end{equation}
where the seven cases are defined as follows:
\begin{enumerate}
\item Let 
\[
	\indicator_{3,1}(x,w,\tau,\beta, \kappa_n) = \ind\left\{\tau - \frac{c}{L(x,w'\beta)}\min\{\tau,1-\tau\}
	> 
	\max \left\{ \frac{\tau-1}{L(x,w'\beta)}+1,\varepsilon \right\} + \kappa_n\right\}
\]
and
\[
	T_{3,1}(x,w,\tau,\beta, h_1) = c \min\{\tau,1-\tau\} \frac{L_\beta(x,w'\beta)'h_1}{L(x,w'\beta)^2} .
\]

\item 
Let 
\[
	\indicator_{3,2}(x,w,\tau,\beta, \kappa_n) = \ind\left\{\frac{\tau-1}{L(x,w'\beta)}+1
	> 
	\max \left\{ \tau - \frac{c}{L(x,w'\beta)}\min\{\tau,1-\tau\},\varepsilon \right\} + \kappa_n \right\}
\]
and
\[
	T_{3,2}(x,w,\tau,\beta, h_1) = (1- \tau) \frac{L_\beta(x,w'\beta)'h_1}{L(x,w'\beta)^2}.
\]

\item 
Let 
\[
	\indicator_{3,3}(x,w,\tau,\beta, \kappa_n) = \ind\left\{ \varepsilon > 
	\max \left\{ \tau - \frac{c}{L(x,w'\beta)}\min\{\tau,1-\tau\}, \frac{\tau-1}{L(x,w'\beta)}+1 \right\} + \kappa_n\right\}
\]
and
\[
	T_{3,3}(x,w,\tau,\beta, h_1) = 0.
\]

\item Let 
\begin{align*}
	\indicator_{3,4}(x,w,\tau,\beta, \kappa_n) &= \ind\left\{ \left|\tau - \frac{c}{L(x,w'\beta)}\min\{\tau,1-\tau\} - \left(\frac{\tau-1}{L(x,w'\beta)}+1\right)\right| \leq \kappa_n \right\}\\
	&\cdot \ind\left\{\min\left\{\tau - \frac{c}{L(x,w'\beta)}\min\{\tau,1-\tau\},\frac{\tau-1}{L(x,w'\beta)}+1\right\} > \varepsilon + \kappa_n \right\}
\end{align*}
and
\[
	T_{3,4}(x,w,\tau,\beta, h_1) = \max\left\{ c\min\{\tau,1-\tau\} \frac{L_\beta(x,w'\beta)'h_1}{L(x,w'\beta)^2}, (1- \tau) \frac{L_\beta(x,w'\beta)'h_1}{L(x,w'\beta)^2} \right\}.
\]

\item Let
\begin{align*}
	\indicator_{3,5}(x,w,\tau,\beta, \kappa_n) &= \ind\left\{ \left|\tau - \frac{c}{L(x,w'\beta)}\min\{\tau,1-\tau\} - \varepsilon\right| \leq \kappa_n \right\}\\
	&\cdot \ind\left\{\min\left\{\tau - \frac{c}{L(x,w'\beta)}\min\{\tau,1-\tau\},\varepsilon\right\} > \frac{\tau-1}{L(x,w'\beta)} +1 + \kappa_n \right\}
\end{align*}
and
\[
	T_{3,5}(x,w,\tau,\beta, h_1) = \max\left\{ c\min\{\tau,1-\tau\} \frac{L_\beta(x,w'\beta)'h_1}{L(x,w'\beta)^2}, 0 \right\}.
\]

\item Let 
\begin{align*}
	\indicator_{3,6}(x,w,\tau,\beta, \kappa_n) &= \ind\left\{ \left|\frac{\tau-1}{L(x,w'\beta)}+1- \varepsilon\right| \leq \kappa_n \right\}\\
	&\cdot \ind\left\{\min\left\{\frac{\tau-1}{L(x,w'\beta)}+1 , \varepsilon \right\} > \tau - \frac{c}{L(x,w'\beta)}\min\{\tau,1-\tau\} + \kappa_n \right\}
\end{align*}
and
\[
	T_{3,6}(x,w,\tau,\beta, h_1) = \max\left\{(1- \tau) \frac{L_\beta(x,w'\beta)'h_1}{L(x,w'\beta)^2},0\right\}.
\]

\item Let $\indicator_{3,7}(x,w,\tau,\beta, \kappa_n)$ equal 1 if at least 2 of the 3 following hold
\begin{align*}
	\left|\tau - \frac{c}{L(x,w'\beta)}\min\{\tau,1-\tau\} - \left(\frac{\tau-1}{L(x,w'\beta)}+1\right)\right|
	&\leq \kappa_n \\
	\left|\tau - \frac{c}{L(x,w'\beta)}\min\{\tau,1-\tau\} - \varepsilon\right| &\leq \kappa_n \\
	\left|\frac{\tau-1}{L(x,w'\beta)}+1- \varepsilon\right| &\leq \kappa_n
\end{align*}
and zero otherwise. Let
\[
	T_{3,7}(x,w,\tau,\beta, h_1) = \max\left\{ c\min\{\tau,1-\tau\} \frac{L_\beta(x,w'\beta)'h_1}{L(x,w'\beta)^2}, (1- \tau) \frac{L_\beta(x,w'\beta)'h_1}{L(x,w'\beta)^2},0 \right\}.
\]

\end{enumerate}

\subsection*{The function $T_4$}

Finally we define $T_4(x,w,\tau,\beta,h_1,\kappa_n)$ in terms of $T_3(x,w,\tau,\beta,h_1,\kappa_n)$ and two indicator functions. Recall from the proof of proposition \ref{prop:CQTE convergence} that
\[
	S_3(x,w,\tau, \beta) = \max\left\{\tau - \frac{c}{L(x,w'\beta)}\min\{\tau,1-\tau\},\frac{\tau-1}{L(x,w'\beta)}+1,\varepsilon\right\}.
\]
Then,
\begin{align*}
	T_4(x,w,\tau,\beta,h_1,\kappa_n)
	&= T_3(x,w,\tau,\beta,h_1,\kappa_n) \cdot \indicator \Big( S_3(x,w,\tau,\beta) < 1-\varepsilon  - \kappa_n \Big) \\
	&\qquad + \min\left\{T_3(x,w,\tau,\beta, h_1, \kappa_n),0\right\} \cdot \indicator \Big( \left|S_3(x,w,\tau,\beta) - (1-\varepsilon) \right| \leq \kappa_n \Big).
\end{align*}

\section{Proofs for Section \ref{sec:bootstrap}}\label{sec:BootstrapProofs}

In this appendix we give the proofs for section \ref{sec:bootstrap}. We start with a lemma about the HDD of $\Gamma_1$.

\begin{lemma}[$\Gamma_{1,\theta_0}'(x,w,\tau,h)$ is Lipschitz in $h$] \label{applemma:Gamma1 Lipschitz}
Suppose the assumptions of proposition \ref{prop:CQTE convergence} hold. Let $\kappa_n \to 0$, $n\kappa_n^2  \to \infty$, $\eta_n \to 0$, and $n\eta_n^2 \to \infty$ as $n\to\infty$. Then
\[
	\left\|\widehat{\Gamma}_{1,\theta_0}'(x,w,\tau,\widetilde{h}) - \widehat{\Gamma}_{1,\theta_0}'(x,w,\tau,h)\right\|
	\leq K(x,w,\widehat{\theta}) \cdot \|\widetilde{h} - h\|_\Theta
\]
for some $K(x,w,\widehat{\theta}) = O_p(1)$ defined below.
\end{lemma}

\begin{proof}[Proof of lemma \ref{applemma:Gamma1 Lipschitz}]
We will show that $\widehat{\overline{\Gamma}}_{1,\theta_0}'(x,w,\tau,h)$ is Lipschitz in $h$. The proof for the lower bound is similar. Let $h, \widetilde{h} \in \R^{d_W}\times \mathscr{C}([\varepsilon,1-\varepsilon],\R^{d_q})$. Then
\begin{align}
	&\Big| \widehat{\overline{\Gamma}}_{1,\theta_0}'(x,w,\tau, \widetilde{h}) - \widehat{\overline{\Gamma}}_{1,\theta_0}'(x,w,\tau, h)\Big| \notag \\
	&\leq \|q(x,w)\| \cdot \|\widetilde{h}_2(S_2(x,w,\tau, \widehat{\beta})) - h_{2}(S_2(x,w,\tau, \widehat{\beta}))\| \nonumber \\
	& \quad + |q(x,w)'\widehat{\gamma}'(S_2(x,w,\tau, \widehat{\beta}))| \cdot |T_2(x,w,\tau,\widehat{\beta}, \widetilde{h}_1, \kappa_n) - T_2(x,w,\tau,\widehat{\beta}, h_1, \kappa_n)|\nonumber \\
	&\leq \|q(x,w)\| \cdot \|\widetilde{h}_2 - h_2\|_\infty \nonumber \\
	& \quad + |q(x,w)'\widehat{\gamma}'(S_2(x,w,\tau, \widehat{\beta}))| \cdot |T_2(x,w,\tau,\widehat{\beta}, \widetilde{h}_1, \kappa_n) - T_2(x,w,\tau,\widehat{\beta}, h_1, \kappa_n)|. \label{eq:tau1hat_lipschitz}
\end{align}
The first line follows by the definition of $\widehat{\overline{\Gamma}}_1$, the triangle inequality, and the Cauchy-Schwarz inequality. Next,
\begin{align*}
	|q(x,w)'\widehat{\gamma}'(S_2(x,w,\tau, \widehat{\beta}))|
		&\leq \|q(x,w)\| \cdot \|\widehat{\gamma}'\|_\infty \\
		&\leq \|q(x,w)\|\left(\|\gamma_0'\|_\infty + \|\widehat{\gamma}' - \gamma_0\|_\infty\right) \\
		&= O_p(1).
\end{align*}
The last line follows since $\|\gamma_0'\|_\infty \leq B$ (by $\gamma_0 \in \mathcal{G}$) and $\|\widehat{\gamma}' - \gamma_0'\|_\infty = o_p(1)$ (by lemma \ref{applemma:QRderivatives}). Thus it suffices to show that $T_2(x,w,\tau,\widehat{\beta}, h_1, \kappa_n)$ is Lipschitz in $h_1$. To see this, write
\begin{align*}
	&|T_2(x,w,\tau,\widehat{\beta}, \widetilde{h}_1, \kappa_n) - T_2(x,w,\tau,\widehat{\beta}, h_1, \kappa_n)|\\
	&\leq |T_1(x,w,\tau,\widehat{\beta}, \widetilde{h}_1, \kappa_n) - T_1(x,w,\tau,\widehat{\beta}, h_1, \kappa_n)|\cdot\ind\left(S_1(x,w,\tau,\widehat{\beta}) > \varepsilon  + \kappa_n\right)\\
	&+ |\max\{T_1(x,w,\tau,\widehat{\beta}, \widetilde{h}_1, \kappa_n),0\} - \max\{T_1(x,w,\tau,\widehat{\beta}, h_1, \kappa_n),0\}|\cdot \ind\left(\left|S_1(x,w,\tau,\widehat{\beta}) - \varepsilon\right| \leq \kappa_n \right)\\
	&\leq |T_1(x,w,\tau,\widehat{\beta}, \widetilde{h}_1, \kappa_n) - T_1(x,w,\tau,\widehat{\beta}, h_1, \kappa_n)|\\
	&+ |\max\{T_1(x,w,\tau,\widehat{\beta}, \widetilde{h}_1, \kappa_n),0\} - \max\{T_1(x,w,\tau,\widehat{\beta}, h_1, \kappa_n),0\}|\\
	&\leq 2 \cdot |T_1(x,w,\tau,\widehat{\beta}, \widetilde{h}_1, \kappa_n) - T_1(x,w,\tau,\widehat{\beta}, h_1, \kappa_n)|
\end{align*}
The first inequality follows by the definition of $T_2$ and the triangle inequality. The last inequality follows from lemma \ref{applemma:minmaxlipschitz}. Thus it suffices to show that $T_1$ is Lipschitz in $h_1$.

To see this, consider
\begin{align}
	&|T_1(x,w,\tau,\widehat{\beta}, \widetilde{h}_1,\kappa_n) - T_1(x,w,\tau,\widehat{\beta}, h_1,\kappa_n)| \notag \\
	&\leq \sum_{j=1}^7|T_{1,j}(x,w,\tau,\widehat{\beta}, \widetilde{h}_1) - T_{1,j}(x,w,\tau,\widehat{\beta}, h_1)| \nonumber \\
&\leq 4 \left| c \min\{\tau,1-\tau\} \frac{L_\beta(x,w'\widehat{\beta})'(\widetilde{h}_1-h_1)}{L(x,w'\widehat{\beta})^2} \right| + 4 \left| \tau \frac{L_\beta(x,w'\widehat{\beta})'(\widetilde{h}_1 - h_1)}{L(x,w'\widehat{\beta})^2} \right| \nonumber \\
	&\leq 8 \left\|\frac{L_\beta(x,w'\widehat{\beta})}{L(x,w'\widehat{\beta})^2} \right\| \cdot \|\widetilde{h}_1 - h_1\|.\label{eq:t1hat_lipschitz}
\end{align}
The second line uses the definition of $T_1$ and repeated applications of lemma \ref{applemma:minmaxlipschitz}. The last line follows from $\tau \leq 1$, $c \min\{\tau,1-\tau\} \leq 1$, and the Cauchy-Schwarz inequality. We have
\[
	\left\|\frac{L_\beta(x,w'\widehat{\beta})}{L(x,w'\widehat{\beta})^2} \right\| = O_p(1)
\]
since
\[
	\sup_{ \beta \in \mathcal{B}_\delta} \left\|\frac{L_\beta(x,w'\beta)}{L(x,w'\beta)^2}\right\| < \infty
\]
and $\Prob(\widehat{\beta} \in \mathcal{B}_\delta) \to 1$. Thus $T_1$ is Lipschitz in $h_1$. 

Overall, we have shown that $\widehat{\overline{\Gamma}}_{1,\theta_0}'(x,w,\tau, h)$ is Lipschitz in $h$ with Lipschitz constant equal to
\[
\overline{K}(x,w,\widehat{\theta}) = \|q(x,w)\|\left(1 + 16 \cdot \|\widehat{\gamma}'\|_\infty \left\|\frac{L_\beta(x,w'\widehat{\beta})}{L(x,w'\widehat{\beta})^2} \right\|\right) = O_p(1).
\]
A similar argument can be used to show that $\widehat{\underline{\Gamma}}_{1,\theta_0}'(x,w,\tau, h)$ is Lipschitz in $h$ with the same constant.
Setting $K(x,w,\widehat{\theta})  = \left(\overline{K}(x,w,\widehat{\theta})^2 + \overline{K}(x,w,\widehat{\theta})^2\right)^{1/2} = \sqrt{2} \cdot \overline{K}(x,w,\widehat{\theta})$ concludes the proof.
\end{proof}

Next we prove proposition \ref{prop:ATE analytical boot}. This is our main result on the analytical bootstrap for mean potential outcomes. In section \ref{sec:bootstrap} we use this result to do bootstrap inference on our ATE bounds. There we discussed how the asymptotic distribution of our mean potential outcome bounds comes from two terms. The first term requires using HDDs while the second term is standard. We consider each term one at a time in the following two lemmas.

\begin{lemma}[Non-standard component]\label{applemma:ATEcomp1 analytical boot}
Suppose the assumptions of proposition \ref{thm:ATE convergence} hold. Suppose $\kappa_n \to 0$, $n \kappa_n^2 \to \infty$, $\eta_n \to 0$, and $n\eta_n^2 \to \infty$ as $n\to\infty$. Then
\[
	\widehat{\Gamma}_{3,\theta_0}'(x, \sqrt{n}( \widehat{\theta}^* - \widehat{\theta})) \overset{P}{\rightsquigarrow} \overline{\Gamma}_{3,\theta_0}'(x,\mathbf{Z}_1).
\]
\end{lemma}

\begin{proof}[Proof of lemma \ref{applemma:ATEcomp1 analytical boot}]
We prove this by applying theorem 3.2 in \cite{FangSantos2014}. To do this we must verify their assumptions 1--4.
\begin{enumerate}
\item Their assumption 1 requires $\Gamma_3(x,\theta)$ to be HDD. We showed this in the proof of theorem \ref{thm:ATE convergence}.

\item Their assumption 2 is about the asymptotic distribution of the first step estimator $\widehat{\theta}$. This holds by our lemma \ref{lemma:prelim estimators} in appendix \ref{sec:prelimest}.

\item Their assumption 3 is about validity of the bootstrap for $\widehat{\theta}$. This holds by theorem 3.6.1 in \cite{VaartWellner1996}.
\end{enumerate}
Finally, in their remark 3.4, they note that sufficient conditions for their assumption 4 are
\begin{enumerate}
\item (Smoothness) $\widehat{\Gamma}_{3,\theta_0}'(x,h)$ is Lipschitz in $h$.

\item (Consistency) $\| \widehat{\Gamma}_{3,\theta_0}'(x,h) - \Gamma_{3,\theta_0}(x,h) \| = o_p(1)$ for any $h$.
\end{enumerate}
These are properties of the HDD estimator $\widehat{\Gamma}_{3,\theta_0}'(x,h)$. We finish this proof by verifying that these properties hold in our setting.

\bigskip

\textbf{Part 1: (Smoothness) $\widehat{\Gamma}_{3,\theta_0}'(x,h)$ is Lipschitz in $h$.} Recall that
\[
	\widehat{\overline{\Gamma}}_{3,\theta_0}'(x, h) = \avg \int_0^1 \widehat{\overline\Gamma}_{1,\theta_0}'(x, W_i,\tau, h) \; d\tau.
\]
So
\begin{align*}
		\left|\widehat{\overline\Gamma}_{3,\theta_0}'(x,\widetilde{h}) - \widehat{\overline\Gamma}_{3,\theta_0}'(x,h)\right| 
		&\leq \avg \int_0^1 \left|\widehat{\overline\Gamma}_{1,\theta_0}'(x,W_i,\tau, \widetilde{h}) - \widehat{\overline\Gamma}_{1,\theta_0}'(x,W_i,\tau, h)\right| \; d\tau\\
		&\leq \left(\avg \int_0^1 \overline{K}_1(x,W_i,\widehat{\theta})  \; d\tau\right)\|\widetilde{h} - h\|_\Theta \\
		&= \left(\avg \overline{K}_1(x,W_i,\widehat{\theta}) \right)\|\widetilde{h} - h\|_\Theta.
\end{align*}
The second line follows by lemma \ref{applemma:Gamma1 Lipschitz}. Next we'll show that
\[
	\avg \overline{K}_1(x,W_i,\widehat{\theta}) = O_p(1).
\]
We have
\begin{align*}
	&\avg \overline{K}_1(x,W_i,\widehat{\theta}) \\
	 &= \avg \|q(x,W_i)\|\left(1 + 16 \cdot \|\widehat{\gamma}'\|_\infty \left\|\frac{L_\beta(x,W_i'\widehat{\beta})}{L(x,W_i'\widehat{\beta})^2} \right\|\right)\\
	&\leq \left( \avg \|q(x,W_i)\| \right) + 16 \cdot \|\widehat{\gamma}'\|_\infty \left(\avg \|q(x,W_i)\|^2\right)^{1/2} \left(  \avg\left\|\frac{L_\beta(x,W_i'\widehat{\beta})}{L(x,W_i'\widehat{\beta})^2} \right\|^2 \right)^{1/2}.
\end{align*}
The first line follows by the definition of $\overline{K}_1$. The second line follows by the Cauchy-Schwarz inequality. By $\Exp(\|q(x,W)\|^2)<\infty$, by $\Prob(\widehat{\beta} \in \mathcal{B}_\delta) \to 1$, and by
\[
	\Exp \left( \sup_{\beta\in\mathcal{B}_\delta} \left\|\frac{L_\beta(x,W'\beta)}{L(x,W'\beta)^2} \right\|^2\right) <\infty
\]
we have
\[
	\avg \int_0^1 \overline{K}_1(x,W_i,\widehat{\theta}) \; d\tau = O_p(1).
\] 
Thus $\widehat{\overline{\Gamma}}_{3,\theta_0}'(x,h)$ is Lipschitz in $h$. It can be similarly shown that $\widehat{\underline{\Gamma}}_{3,\theta_0}'(x,h)$ is Lipschitz in $h$.

\bigskip

\textbf{Part 2: Consistency of $\widehat{\Gamma}_{3,\theta_0}'(x,h)$.} Next we show that
\[
	\widehat{\overline{\Gamma}}_{3,\theta_0}'(x,h) \xrightarrow{p} \overline{\Gamma}_{3,\theta_0}'(x,h).
\]
To do this we use the triangle inequality to decompose their difference into four different terms:
\[
	\left|\widehat{\overline{\Gamma}}_{3,\theta_0}'(x, h) - \overline{\Gamma}_{3,\theta_0}'(x, h)\right|
	\leq R_1 + R_2 + R_3 + R_4
\]
where
\begin{align*}
	R_1
		&= \left|\avg  q(x,W_i)'\int_0^1 h_2(S_2(x,W_i,\tau, \widehat{\beta})) \; d\tau - \Exp \left[q(x,W)' \int_0^1 h_{2}(S_2(x,W,\tau, \beta_0)) \; d\tau \right]\right| \\[1em]
	R_2
		&= \Bigg| \avg \left[q(x,W_i)'\int_0^1\widehat{\gamma}'(S_2(x,W_i,\tau, \widehat{\beta}))T_2(x,W_i,\tau,\widehat{\beta}, h_1, \kappa_n) \; d\tau \right. \\
		&\hspace{20mm} \left. - q(x,W_i)'\int_0^1\gamma_0'(S_2(x,W_i,\tau, \widehat{\beta}))T_2(x,W_i,\tau,\widehat{\beta}, h_1, \kappa_n) \; d\tau\right] \Bigg| \\[1em]
	R_3
		&= \Bigg| \avg q(x,W_i)'\int_0^1\gamma_0'(S_2(x,W_i,\tau, \widehat{\beta}))T_2(x,W_i,\tau,\widehat{\beta}, h_1, \kappa_n) \; d\tau \\
		&\hspace{20mm} - \Exp \left[q(x,W)'\int_0^1\gamma_0'(S_2(x,W,\tau, \widehat\beta))T_2(x,W,\tau,\widehat\beta, h_1, \kappa_n) \; d\tau\right] \Bigg| \\[1em]
	R_4
		&= \Bigg| \Exp \left[q(x,W)'\int_0^1\gamma_0'(S_2(x,W,\tau, \widehat\beta))T_2(x,W,\tau,\widehat\beta, h_1, \kappa_n) \; d\tau\right] \\
		&\hspace{20mm} - \Exp \left[q(x,W)'\int_0^1\gamma_0'(S_2(x,W,\tau, \beta_0))T_2(x,W,\tau,\beta_0, h_1,0) \; d\tau\right] \Bigg|.
\end{align*}

\bigskip

\emph{\textbf{Convergence of $R_1$.}} $h_2$ is continuous on compact domain $[\varepsilon,1-\varepsilon]$. Hence $h_2$ is bounded on $[\varepsilon,1-\varepsilon]$. Since $S_2$ lies between $[\varepsilon, 1-\varepsilon]$, the composite function $h_2(S_2(x,w,\tau, \beta))$ is thus bounded uniformly over $(x,w,\tau) \in \{0,1\}\times\mathcal{W}\times(0,1)$. This composite function is also continuous in $\beta$ for any $(x,w,\tau)$, since $S_2$ is continuous in $\beta$. Therefore, by the dominated convergence theorem,
\[
	\int_0^1 h_{2}(S_2(x,w,\tau, \beta)) \; d\tau
\]
is continuous in $\beta$ for all $(x,w) \in \{0,1\}\times\mathcal{W}$. Moreover, it has a bounded envelope:
\[
	\Exp \left( \sup_{\beta \in \mathcal{B}} \int_0^1 h_2(S_2(x,W,\tau,\beta)) \; d\tau \right) < \infty.
\]
These properties plus compactness of $\mathcal{B}$ imply that 
\[
	\left\{\int_0^1 h_{2}(S_2(x,W,\tau, \beta)) \; d\tau: x\in\{0,1\}, \beta \in \mathcal{B}\right\}
\]
is Glivenko-Cantelli, by example 19.8 in \cite{Vaart2000}. By $\Exp ( \|q(x,W)\|^2) < \infty$ (A\ref{assn:quant reg}.4) and by corollary 9.27 part (ii) in \cite{Kosorok2007}, the class of functions
\[
	\left\{q(x,W)'\int_0^1 h_{2}(S_2(x,W,\tau, \beta)) \; d\tau: x\in\{0,1\}, \beta \in \mathcal{B}\right\}
\]
is also Glivenko-Cantelli. Hence
\begin{align*}
	R_1 &= \left|\avg  q(x,W_i)'\int_0^1 h_2(S_2(x,W_i,\tau, \widehat{\beta})) \; d\tau - \Exp \left[q(x,W)' \int_0^1 h_{2}(S_2(x,W,\tau, \beta_0)) \; d\tau \right]\right|\\
	&\leq \sup_{\beta\in\mathcal{B}}\left|\avg  q(x,W_i)'\int_0^1 h_{2}(S_2(x,W_i,\tau, \beta)) \; d\tau - \Exp \left[q(x,W)'\int_0^1 h_{2}(S_2(x,W,\tau, \beta)) \; d\tau\right]\right|\\
	& + \left|\int_\mathcal{W} q(x,w)'\int_0^1 h_{2}(S_2(x,w,\tau, \widehat{\beta})) \; d\tau  \; dF_W(w) - \int_\mathcal{W} q(x,w)'\int_0^1 h_{2}(S_2(x,w,\tau, \beta_0)) \; d\tau  \; dF_W(w)\right| \\
	&= o_p(1) + o_p(1)  \\
	&= o_p(1).
\end{align*}
The second line follows by the triangle inequality. The first term in that line is $o_p(1)$ by the Glivenko-Cantelli property. The second term is $o_p(1)$ by its continuity in $\beta$ and $\widehat{\beta} \xrightarrow{p} \beta_0$.

\bigskip

\emph{\textbf{Convergence of $R_2$}}. We have
\begin{align*}
	R_2 &= \Bigg|\avg \left[q(x,W_i)'\int_0^1\widehat{\gamma}'(S_2(x,W_i,\tau, \widehat{\beta}))T_2(x,W_i,\tau,\widehat{\beta}, h_1, \kappa_n) \; d\tau \right. \\
	&\hspace{60mm} \left. - q(x,W_i)'\int_0^1\gamma_0'(S_2(x,W_i,\tau, \widehat{\beta}))T_2(x,W_i,\tau,\widehat{\beta}, h_1, \kappa_n) \; d\tau\right]\Bigg| \\
	&\leq \avg \left\|q(x,W_i)\right\|\int_0^1 \left\|\widehat{\gamma}'(S_2(x,W_i,\tau, \widehat{\beta}))-\gamma_0'(S_2(x,W_i,\tau, \widehat{\beta}))\right\| \left|T_2(x,W_i,\tau,\widehat{\beta}, h_1, \kappa_n) \right| \; d\tau\\
	&\leq \left\|\widehat{\gamma}'-\gamma_0'\right\|_\infty \avg \left\|q(x,W_i)\right\|\int_0^1 \left|T_2(x,W_i,\tau,\widehat{\beta}, h_1, \kappa_n)\right| \; d\tau\\
	&\leq  o_p(1) \times \left(\avg \left\|q(x,W_i)\right\|^2\right)^{1/2} \times \left(\avg \left(\int_0^1 \left|T_2(x,W_i,\tau,\widehat{\beta}, h_1, \kappa_n)\right| \; d\tau\right)^2\right)^{1/2}.
\end{align*}
The first line is the definition of $R_2$. The second line follows by the triangle inequality and the Cauchy-Schwarz inequality. The last line follows by uniform convergence of $\widehat{\gamma}'$ to $\gamma_0'$ (lemma \ref{applemma:QRderivatives}) and the Cauchy-Schwarz inequality. By A\ref{assn:quant reg}.4,
\[
	\left(\avg \left\|q(x,W_i)\right\|^2\right)^{1/2} = O_p(1).
\]
Also, 
\begin{align*}
	\int_0^1 \left|T_2(x,W_i,\tau,\widehat{\beta}, h_1, \kappa_n)\right| \; d\tau &\leq \int_0^1 \left|T_1(x,W_i,\tau,\widehat{\beta}, h_1, \kappa_n)\right| \; d\tau\\
	&\leq \int_0^1 \left|\max\left\{- c\min\{\tau,1-\tau\} \frac{L_\beta(x,W_i'\widehat{\beta})'h_1}{L(x,W_i'\widehat{\beta})^2}, - \tau \frac{L_\beta(x,W_i'\widehat{\beta})'h_1}{L(x,W_i'\widehat{\beta})^2}, 0 \right\}\right| \; d\tau\\
	&\leq \int_0^1 (c\min\{\tau,1-\tau\} + \tau) \; d\tau \left|\frac{L_\beta(x,W_i'\widehat{\beta})'h_1}{L(x,W_i'\widehat{\beta})^2}\right|\\
	&= \frac{c+2}{4} \left|\frac{L_\beta(x,W_i'\widehat{\beta})'h_1}{L(x,W_i'\widehat{\beta})^2}\right|.
\end{align*}
The first line follows by the definition of $T_2$. The second line follows by the definition of $T_1$ (notice the maximum of the three values $T_1$ can take is an upper bound for it). By A\ref{assn:prop score},
\[
	\left(\frac{L_\beta(x,w'\beta)'h_1}{L(x,w'\beta)^2}\right)^2
\]
is continuous in $\beta$ for any $w \in \mathcal{W}$. Moreover, this term has a bounded envelope in a neighborhood of $\beta_0$ by our assumption that
\[
	\Exp \left( \sup_{\beta\in\mathcal{B}_\delta}\left\|\frac{L_\beta(x,W'\beta)}{L(x,W'\beta)^2}\right\|^2\right) <\infty.
\]
Hence
\begin{align*}
	\avg \left(\int_0^1 \left|T_2(x,W_i,\tau,\widehat{\beta}, h_1, \kappa_n)\right| \; d\tau\right)^2 &\leq \frac{(c+2)^2}{16}\avg \left\|\frac{L_\beta(x,W_i'\widehat\beta)}{L(x,W_i'\widehat\beta)^2}\right\|^2 \|h_1\|^2 \\
	&= O_p(1)
\end{align*}
where the last line follows by the uniform law of large numbers, as in example 19.8 in \cite{Vaart2000}. Thus $R_2 = o_p(1)$.

\bigskip

\emph{\textbf{Convergence of $R_3$}}. For fixed $h_1$ and $x$, let
\[
	g_n(w,\beta) = q(x,w)'\int_0^1\gamma_0'(S_2(x,w,\tau, \beta))T_2(x,w,\tau,\beta, h_1, \kappa_n) \; d\tau.
\]
Define
\[
	R_3(\beta) = \left| \frac{1}{n} \sum_{i=1}^n g_n(W_i,\beta) - \Exp[g_n(W,\beta)] \right|.
\]
Then $R_3 = R_3(\widehat{\beta})$. We want to show that $R_3 = o_p(1)$. For any $\epsilon > 0$,
\begin{align*}
	\Prob(| R_3 | \geq \epsilon)
		&= \Prob(| R_3 | \geq \epsilon, \widehat{\beta} \in \mathcal{B}_\delta) + \Prob(| R_3 | \geq \epsilon, \widehat{\beta} \notin \mathcal{B}_\delta) \\
		&\leq \Prob \left( \sup_{\beta \in \mathcal{B}_\delta} | R_3(\beta) | \geq \epsilon, \widehat{\beta} \in \mathcal{B}_\delta \right) + \Prob(\widehat{\beta} \notin \mathcal{B}_\delta) \\
		&\leq \Prob \left( \sup_{\beta \in \mathcal{B}_\delta} | R_3(\beta) | \geq \epsilon \right) + \Prob(\widehat{\beta} \notin \mathcal{B}_\delta).
\end{align*}
The second term converges to zero by consistency of $\widehat{\beta}$. Thus it suffices to show that the first term converges to zero. That is, we want to show that
\[
	\sup_{\beta \in \mathcal{B}_\delta} \left| \frac{1}{n} \sum_{i=1}^n g_n(W_i,\beta) - \Exp[g_n(W,\beta)] \right| 
	= o_p(1).
\]
This follows from a uniform law of large numbers. Specifically, we use theorem 4.2.2 in \cite{Amemiya1985}.
There are two main properties required to apply this theorem:
\begin{enumerate}
\item A dominance condition:
\[
	\Exp \left(\sup_{\beta\in\mathcal{B}_\delta}|g_n(W,\beta)|^2\right) < \infty.
\]

\item A continuity condition: $g_n(w,\beta)$ is continuous at any $\beta \in \mathcal{B}_\delta$ for all $w \in \mathcal{W}$. 
\end{enumerate}
So we conclude the proof by verifying these two properties.

\bigskip

\emph{The dominance condition}. We have
\begin{align*}
	\Exp \left(\sup_{\beta\in\mathcal{B}_\delta}|g_n(W,\beta)|^2\right)
	&= \Exp \left( \sup_{\beta \in \mathcal{B}_\delta} \left|q(x,W)'\int_0^1\gamma_0'(S_2(x,W,\tau, \beta))T_2(x,W,\tau,\beta, h_1, \kappa_n) \; d\tau\right|^2\right)\\
	&\leq \Exp \left(\|q(x,W)\|^2 \sup_{\beta\in\mathcal{B}_\delta} \int_0^1 \|\gamma_0'(S_2(x,W,\tau, \beta))\|^2 | T_2(x,W,\tau,\beta, h_1, \kappa_n)|^2 \; d\tau\right)\\
	&\leq \Exp \left(\|q(x,W)\|^2 B^2 \left(\frac{c+2}{4}\right)^2 \sup_{\beta\in\mathcal{B}_\delta} \left|\frac{L_\beta(x,W'\beta)'h_1}{L(x,W'\beta)^2}\right|^2\right)\\
	&\leq B^2 \left(\frac{c+2}{4}\right)^2 \Exp \left(\|q(x,W)\|^4\right)^{1/2}\Exp \left(\sup_{\beta\in\mathcal{B}_\delta} \left\|\frac{L_\beta(x,W'\beta)}{L(x,W'\beta)^2}\right\|^4\right)^{1/2} \|h_1\|^2\\
	&<\infty.
\end{align*}
The second and fourth lines follow by the Cauchy-Schwarz inequality.

\bigskip

\emph{The continuity condition.} Fix $w \in \mathcal{W}$. We next show continuity of
\[
	g_n(w,\beta) = q(x,w)'\int_0^1\gamma_0'(S_2(x,w,\tau, \beta))T_2(x,w,\tau,\beta, h_1, \kappa_n) \; d\tau
\]
in $\beta$. First note that
\[
	\lim_{m \rightarrow \infty} g_n(w,\beta_m) = 
	q(x,w)' \left( \lim_{m \rightarrow \infty} \int_0^1 \gamma_0'(S_2(x,w,\tau, \beta_m))T_2(x,w,\tau,\beta_m, h_1, \kappa_n) \; d\tau \right).
\]
We now show that we can bring the limit inside the integral by applying the dominated convergence theorem. First note that the integrand satisfies a dominance condition, similar to our analysis above. The other condition is pointwise convergence of the integrand for all $\tau \in (0,1)$ except possibly on a set of Lebesgue measure zero. Thus we need to show that
\[
	\lim_{m \rightarrow \infty} \gamma_0'(S_2(x,w,\tau, \beta_m))T_2(x,w,\tau,\beta_m, h_1, \kappa_n) =
	\gamma_0'(S_2(x,w,\tau, \beta))T_2(x,w,\tau,\beta, h_1, \kappa_n)
\]
for all $\tau \in (0,1)$ except possibly a set of Lebesgue measure zero. We do this by showing that
\[
	\gamma_0'(S_2(x,w,\tau, \beta))T_2(x,w,\tau,\beta, h_1, \kappa_n)
\]
is continuous in $\beta$, for all $\tau \in (0,1)$ except possibly a set of Lebesgue measure zero. This term is the product of two pieces, so it suffices to show that each piece separately is continuous. After we do this, the overall proof will be complete.

\bigskip

\emph{Piece 1}: By $\gamma_0 \in \mathcal{G}$ and by continuity of $S_2(x,w,\tau,\beta)$ in $\beta$, the function $\gamma_0'(S_2(x,w,\tau, \beta))$ is continuous. %

\bigskip

\emph{Piece 2}: Next we'll show that $T_2(x,w,\tau,\beta,h_1,\kappa_n)$ is continuous in $\beta$ for all $\tau \in (0,1)$ except a set of Lebesgue measure zero. Recall that, omitting the arguments of the functions,
\[
	T_2 = T_1 \cdot \indicator( S_1 - \varepsilon  > \kappa_n) + \max \{ T_1, 0 \} \cdot \indicator( | S_1 - \varepsilon | \leq \kappa_n ).
\]
So we'll start by studying continuity of $T_1(x,w,\tau,\beta, h_1, \kappa_n)$ in $\beta$. Recall equation \eqref{eq:T1summationRepresentation}:
\[
	T_1(x,w,\tau,\beta,h_1,\kappa_n) = \sum_{j=1}^7 T_{1,j}(x,w,\tau,\beta,h_1) \cdot \ind_{1,j}(x,w,\tau,\beta,\kappa_n).
\]
Given our assumptions on the propensity score (A\ref{assn:prop score}), $T_{1,j}(x,w,\tau,\beta,h_1)$ is a composition of functions that are continuous in $\beta$. Thus it is also continuous in $\beta$. This holds for all $j=1,\ldots,7$. For example,
\[
	T_{1,1}(x,w,\tau,\beta, h_1) = -c\min\{\tau,1-\tau\} \frac{L_\beta(x,w'\beta)'h_1}{L(x,w'\beta)^2},
\]
which is continuous in $\beta$ by continuity of $L(x,\cdot)$. This holds for all $\tau \in (0,1)$.

Next consider the indicator functions $\ind_{1,j}(x,w,\tau,\beta, \kappa_n)$. Fix $\beta \in \mathcal{B}_\delta$. We will show that these functions are constant, and therefore continuous at $\beta$, except on a set of $\tau$'s of Lebesgue measure zero.

First consider
\[
	\indicator_{1,1}(x,w,\tau,\beta, \kappa_n) = \ind\left\{\tau + \frac{c}{L(x,w'\beta)} \min\{\tau,1-\tau\}
	< \min \left\{ \frac{\tau}{L(x,w'\beta)},1-\varepsilon \right\} - \kappa_n\right\}
\]
Recall that $w \in \mathcal{W}$ is fixed, along with $x$ and $\kappa_n >0$.
Let
\[
	\mathcal{T}_1 = \left\{ \tau \in (0,1) : \tau + \frac{c}{L(x,w'\beta)}\min\{\tau,1-\tau\} = \min\left\{\frac{\tau}{L(x,w'\beta)},1-\varepsilon \right\} - \kappa_n \right\}.
\]
Then $\indicator_{1,1}(x,w,\tau,\beta, \kappa_n)$ is continuous at $\beta$ for any $\tau \notin \mathcal{T}_1$. Moreover, $\mathcal{T}_1$ has Lebesgue measure zero. To see this, let $\bar\tau = (1-\varepsilon) L(x,w'\beta)$. Suppose $\bar{\tau} \leq 1/2$. Then
\begin{align*}
	\mathcal{T}_1
		&= \left\{ \tau \in (0, \bar{\tau}] : \tau\left(1 - \frac{1-c}{L(x,w'\beta)}\right) = -\kappa_n\right\}\\
		&\qquad \cup \left\{ \tau \in (\bar{\tau},1/2] : \tau\left(1 + \frac{c}{L(x,w'\beta)}\right) = 1-\varepsilon-\kappa_n\right\}\\
		&\qquad \cup \left\{ \tau \in (1/2,1) : \tau\left(1 - \frac{c}{L(x,w'\beta)}\right) = 1-\varepsilon-\frac{c}{L(x,w'\beta)}-\kappa_n \right\}.
\end{align*}
Since $\kappa_n \neq 0$ the first set in this union contains at most one point. Since %
\[
	1 + \frac{c}{L(x,w'\beta)} \neq 0,
\]
the second set in this union also contains at most one point. Finally, the last set contains more than one point only if
\[
	L(x,w'\beta) = c
	\qquad \text{and} \qquad
	1-\varepsilon - \frac{c}{L(x,w'\beta)} - \kappa_n = 0,
\]
which yields a contradiction since $-\varepsilon - \kappa_n <0$. Therefore, this is the union of at most three points and hence has measure zero. If $\bar{\tau} \geq 1/2$ then
\begin{align*}
	\mathcal{T}_1
	&= \left\{ \tau \in (0,1/2] :  \tau\left(1 - \frac{1-c}{L(x,w'\beta)}\right) = -\kappa_n\right\}\\
	&\qquad \cup \left\{ \tau \in (1/2,\bar{\tau}] : \tau\left(1 - \frac{1+c}{L(x,w'\beta)}\right) =- \frac{c}{L(x,w'\beta)} -\kappa_n\right\}\\
	&\qquad \cup \left\{\tau \in (\bar{\tau},1) : \tau\left(1 - \frac{c}{L(x,w'\beta)}\right) = 1-\varepsilon-\frac{c}{L(x,w'\beta)}-\kappa_n\right\}.
\end{align*}
Once again this is the union of at most three points and hence has measure zero. Thus we've shown that for any $\beta \in \mathcal{B}_\delta$ and for all $\tau$ except a set of Lebesgue measure zero,
\[
	T_{1,1}(x,w,\tau,\beta, h_1) \cdot\ind_{1,1}(x,w,\tau,\beta, \kappa_n)
\]
is continuous at $\beta$.

A similar argument holds for the other indicators. Specifically, let $\mathcal{T}_j$ denote the set of $\tau$'s at which $\ind_{1,j}(x,w,\tau,\beta, \kappa_n)$ is discontinuous at $\beta$. Then
\begin{align*}
	\mathcal{T}_2 &= \left\{\tau \in (0,1) : \frac{\tau}{L(x,w'\beta)} = \min\left\{\tau + \frac{c}{L(x,w'\beta)}\min\{\tau,1-\tau\},1-\varepsilon \right\} - \kappa_n\right\}\\[0.5em]
	\mathcal{T}_3
	&= \left\{ \tau \in (0,1) : 1 - \varepsilon = \min \left\{ \tau + \frac{c}{L(x,w'\beta)} \min \{\tau,1-\tau\}, \frac{\tau}{L(x,w'\beta)} \right\} - \kappa_n \right\}
	\\[0.5em]
	\mathcal{T}_4 &\subseteq \left\{\tau \in (0,1) : \left|\tau + \frac{c}{L(x,w'\beta)}\min\{\tau,1-\tau\} - \frac{\tau}{L(x,w'\beta)}\right| = \kappa_n\right\}\\
	&\qquad \cup \left\{\tau \in (0,1) : \tau + \frac{c}{L(x,w'\beta)}\min\{\tau,1-\tau\}  = 1-\varepsilon-\kappa_n\right\} \\
	&\qquad \cup \left\{\tau \in (0,1) : \frac{\tau}{L(x,w'\beta)} = 1-\varepsilon-\kappa_n\right\}\\[0.5em]
	\mathcal{T}_5 &\subseteq \left\{\tau \in (0,1) : \left|\tau + \frac{c}{L(x,w'\beta)}\min\{\tau,1-\tau\} -(1-\varepsilon)\right| = \kappa_n\right\}\\
	&\qquad \cup \left\{\tau \in (0,1) : \tau + \frac{c}{L(x,w'\beta)}\min\{\tau,1-\tau\} = \frac{\tau}{L(x,w'\beta)} - \kappa_n\right\} \\
	&\qquad \cup \left\{\tau \in (0,1) : 1-\varepsilon = \frac{\tau}{L(x,w'\beta)} - \kappa_n\right\}\\[0.5em]
	\mathcal{T}_6 &\subseteq \left\{\tau \in (0,1) : \left|\frac{\tau}{L(x,w'\beta)} -(1-\varepsilon)\right| = \kappa_n\right\}\\
	&\qquad \cup \left\{\tau \in (0,1) : \frac{\tau}{L(x,w'\beta)} = \tau + \frac{c}{L(x,w'\beta)}\min\{\tau,1-\tau\} - \kappa_n\right\}\\
	&\qquad \cup \left\{\tau \in (0,1) : 1-\varepsilon = \tau + \frac{c}{L(x,w'\beta)}\min\{\tau,1-\tau\} - \kappa_n\right\}\\[0.5em]
	\mathcal{T}_7 &\subseteq \left\{\tau \in (0,1) : \left|\tau + \frac{c}{L(x,w'\beta)}\min\{\tau,1-\tau\} -(1-\varepsilon)\right| = \kappa_n\right\}\\
	&\qquad \cup \left\{\tau \in (0,1) : \left|\frac{\tau}{L(x,w'\beta)} -(1-\varepsilon)\right| = \kappa_n\right\}\\
	&\qquad \cup \left\{\tau \in (0,1) : \left|\tau + \frac{c}{L(x,w'\beta)}\min\{\tau,1-\tau\} -\frac{\tau}{L(x,w'\beta)}\right|= \kappa_n\right\}.
\end{align*}
Using similar arguments to the $j=1$ case, we see that all of these sets have Lebesgue measure zero. So $\bigcup_{j=1}^7 \mathcal{T}_j$ has Lebesgue measure zero. Hence the function $T_{1}(x,w,\tau,\beta, h_1, \kappa_n)$ is continuous at all $\beta \in \mathcal{B}_\delta$ for all $\tau \in (0,1)$ except a set of Lebesgue measure zero.

\bigskip

Now let's return to $T_2$. This function is continuous in $\beta$ for all $\tau \in (0,1)$ except possibly on the set
\[
	\mathcal{T} = \left( \bigcup_{j=1}^7 \mathcal{T}_j \right) \cup
	\left\{ \tau \in (0,1) : \left|\min\left\{\tau + \frac{c}{L(x,w'\beta)}\min\{\tau,1-\tau\},\frac{\tau}{L(x,w'\beta)},1-\varepsilon\right\} - \varepsilon\right| = \kappa_n\right\}.
\]
The second term here comes from the indicators $\indicator( |S_1 - \varepsilon| \leq \kappa_n)$ and $\indicator( S_1 - \varepsilon  > \kappa_n)$. We can see that this set has Lebesgue measure zero using similar arguments as above. Thus the overall set $\mathcal{T}$ has Lebesgue measure zero. Hence we've shown that, for a fixed $(x,w,h_1,\kappa_n)$, and for any $\beta \in \mathcal{B}_\delta$, $T_2(x,w,\tau,\beta,h_1,\kappa_n)$ is continuous at $\beta$ for all $\tau \in (0,1)$ except a set of Lebesgue measure zero. As noted earlier, this is sufficient to complete the proof that $R_3 = o_p(1)$.

\bigskip

\emph{\textbf{Convergence of $R_4$.}} This part is the difference between two expectations, one evaluated at $\widehat{\beta}$ and the other at $\beta_0$. Note that the expectations are over $W$, not $\widehat{\beta}$. We'll show that $R_4 = o_p(1)$ by applying the dominated convergence theorem and then using the fact that $\widehat{\beta} \xrightarrow{p} \beta_0$. The $R_4$ term is similar to $R_1$, and so our proof here will use some of our derivations from our proof that $R_1 = o_p(1)$. The main difference is that $R_4$ is a function of $T_2$. So we'll spend most of our time on that. $T_2$, in turn, depends on $T_1$. So we'll begin by showing that
\[
	T_1(x,w,\tau,\widehat{\beta},h_1,\kappa_n)
	\xrightarrow{p}
	T_1(x,w,\tau,\beta_0,h_1,0).
\]
By the definition of $T_1$, this convergence holds if
\[
	T_{1,j}(x,w,\tau,\widehat\beta,h_1) \xrightarrow{p} T_{1,j}(x,w,\tau,\beta_0,h_1)
	\qquad \text{and} \qquad
	\ind_{1,j}(x,w,\tau,\widehat\beta,\kappa_n) \xrightarrow{p} \ind_{1,j}(x,w,\tau,\beta_0,0).
\]
for all $j=1,\ldots,7$. By continuity of $T_{1,j}(x,w,\tau,\beta,h_1)$ in $\beta$ for all $j=1,\ldots,7$, and by the consistency of $\widehat{\beta}$, the $T_{1,j}(x,w,\tau,\widehat\beta,h_1)$ terms are consistent. The indicator functions are slightly trickier. Given $(x,w,\tau)$, the value of $\beta_0$ determines which of the seven cases we are in. That is, which of the indicators $\ind_{1,j}(x,w,\tau,\beta_0,0)$ is 1; the other six are all zero.

We'll consider each case separately. First suppose that
\[
	\tau + \frac{c}{L(x,w'\beta_0)} \min\{\tau,1-\tau\}
	< \min \left\{ \frac{\tau}{L(x,w'\beta_0)},1-\varepsilon \right\}.
\]
Thus we have $\ind_{1,1}(x,w,\tau,\beta_0,0) = 1$. By $\widehat{\beta} \xrightarrow{p} \beta_0$ and $\kappa_n \to 0$, 
\begin{align*}
	\ind_{1,1}(x,w,\tau,\widehat\beta,\kappa_n) &= \ind\left\{\tau + \frac{c}{L(x,w'\widehat\beta)} \min\{\tau,1-\tau\}
	< \min \left\{ \frac{\tau}{L(x,w'\widehat\beta)},1-\varepsilon \right\} - \kappa_n\right\}\\
	& \xrightarrow{p} \ind\left\{\tau + \frac{c}{L(x,w'\beta_0)} \min\{\tau,1-\tau\}
	< \min \left\{ \frac{\tau}{L(x,w'\beta_0)},1-\varepsilon \right\}\right\}\\
	&= \ind_{1,1}(x,w,\tau,\beta_0,0) \\
	&= 1.
\end{align*}
Moreover, by taking complements, we see that for these values of $(x,w,\tau,\beta_0)$ all the other indicators converge (to zero) as well.

Next suppose $\ind_{1,2}(x,w,\tau,\beta_0,0) = 1$ or $\ind_{1,3}(x,w,\tau,\beta_0,0) = 1$. In either of these cases, we can similarly show that
\[
	\ind_{1,j}(x,w,\tau,\widehat\beta,\kappa_n) \xrightarrow{p} \ind_{1,j}(x,w,\tau,\beta_0,0)
\]
for $j=2,3$. Next suppose
\[
	\tau + \frac{c}{L(x,w'\beta_0)}\min\{\tau,1-\tau\} = \frac{\tau}{L(x,w'\beta_0)} < 1-\varepsilon,
\]
which puts us in the $\ind_{1,4}(x,w,\tau,\beta_0,0) = 1$ case. This case is more delicate, and shows the second place where the $\kappa_n$'s are important. $\ind_{1,4}(x,w,\tau,\widehat{\beta},\kappa_n)$ can be viewed as the product of three indicator functions. Two of them are handled like the $j=1,2,3$ cases:
\[
	\ind\left\{\tau + \frac{c}{L(x,w'\widehat\beta)}\min\{\tau,1-\tau\} < 1-\varepsilon - \kappa_n\right\} \xrightarrow{p} \ind\left\{\tau + \frac{c}{L(x,w'\beta_0)}\min\{\tau,1-\tau\} < 1-\varepsilon \right\}
\]
and
\[
	\ind\left\{\frac{\tau}{L(x,w'\widehat\beta)} < 1-\varepsilon - \kappa_n\right\} \xrightarrow{p} \ind\left\{\frac{\tau}{L(x,w'\beta_0)} < 1-\varepsilon \right\}
\]
by $\widehat{\beta}\xrightarrow{p} \beta_0$ and $\kappa_n \to 0$. The third indicator requires a different argument:
\begin{align*}
	&\indicator\left\{\left|\tau + \frac{c}{L(x,w'\widehat{\beta})}\min\{\tau,1-\tau\} - \frac{\tau}{L(x,w'\widehat{\beta})}\right| \leq \kappa_n \right\}\\
		&= \indicator\left\{ \frac{1}{\sqrt{n \kappa_n^2}} \cdot \sqrt{n}\left(\tau + \frac{c}{L(x,w'\widehat{\beta})}\min\{\tau,1-\tau\} - \frac{\tau}{L(x,w'\widehat{\beta})}\right) \in [-1,1]\right\}\\
		&\to 1
\end{align*}
since $n\kappa_n^2 \to \infty$ and
\[
	\sqrt{n}\left(\tau + \frac{c}{L(x,w'\widehat{\beta})}\min\{\tau,1-\tau\} - \frac{\tau}{L(x,w'\widehat{\beta})}\right) = O_p(1).
\]
This term is $O_p(1)$ since we're looking at the case where
\[
	\tau + \frac{c}{L(x,w'\widehat{\beta})}\min\{\tau,1-\tau\} - \frac{\tau}{L(x,w'\widehat{\beta})} \xrightarrow{p} \tau + \frac{c}{L(x,w'\beta_0)}\min\{\tau,1-\tau\} - \frac{\tau}{L(x,w'\beta_0)} = 0,
\]
and by the delta method. Combining the consistency of these three indicator functions, we have that
\[
	\ind_{1,4}(x,w,\tau,\widehat\beta,\kappa_n) \xrightarrow{p} \ind_{1,4}(x,w,\tau,\beta_0,0).
\]
As before, by taking complements we see that all of the other indicator functions are all also consistent in this case. Notice that in this case $\kappa_n$ is a slackness parameter that we introduced to allow the indicator $\ind_{1,4}(x,w,\tau,\widehat\beta,\kappa_n)$ to be 1 even if the inequality 
\[
	\tau + \frac{c}{L(x,w' \widehat{\beta})}\min\{\tau,1-\tau\} = \frac{\tau}{L(x,w'\widehat{\beta})} 
\]
does not hold exactly in finite samples. 

The last three cases are all similar to the case $\ind_{1,4}(x,w,\tau,\beta_0,0) = 1$ that we just studied. Thus, putting all of these cases together gives
\begin{align*}
	T_1(x,w,\tau,\widehat{\beta}, h_1, \kappa_n) &\xrightarrow{p} T_1(x,w,\tau,\beta_0, h_1,0).
\end{align*}
By similar arguments, we can also show that
\[
	T_2(x,w,\tau,\widehat{\beta}, h_1, \kappa_n) \xrightarrow{p} T_2(x,w,\tau,\beta_0, h_1, 0).
\]
By continuity of $\gamma_0'(\cdot)$ and of $S_2(x,w,\tau,\beta)$ in $\beta$, this implies that
\[
	q(x,w)'\gamma_0'(S_2(x,w,\tau,\widehat{\beta}))T_2(x,w,\tau,\widehat{\beta},h_1,\kappa_n) \xrightarrow{p} q(x,w)'\gamma_0'(S_2(x,w,\tau,\beta_0))T_2(x,w,\tau,\beta_0,h_1,0).
\]
Finally, note that 
\[
	q(x,w)'\int_0^1\gamma_0'(S_2(x,w,\tau, \beta))T_2(x,w,\tau,\beta, h_1, \kappa_n) \;  \; d\tau
\]
is continuous in $\beta$ and has a bounded envelope (which can be seen using arguments similar to that in our proof for $R_1$). Thus we can apply the dominated convergence theorem, which gives $R_4 = o_p(1)$.

\bigskip

\textbf{\emph{Putting the four pieces together}}. We've shown that $R_1,\ldots,R_4$ are all $o_p(1)$. Thus we've shown that
\[
	|\widehat{\overline{\Gamma}}_{3,\theta_0}'(x,h) - \overline{\Gamma}_{3,\theta_0}'(x,h)| = o_p(1).
\]
A similar argument can be used to show that
\[
	|\widehat{\underline{\Gamma}}_{3,\theta_0}'(x,h) - \underline{\Gamma}_{3,\theta_0}'(x,h)| = o_p(1).
\]
As discussed at the beginning of the proof this consistency of $\widehat{\Gamma}_{3,\theta_0}'(x,h)$ was all that we had left to show, so we are done.
\end{proof}

\begin{lemma}[Standard component]\label{applemma:ATEcomp2 analytical boot}
Suppose the assumptions of theorem \ref{thm:ATE convergence} hold. Then
\begin{align*}
	\mathbb{G}_n^* \Gamma_2(x,W,\widehat{\theta})
	&\overset{P}{\rightsquigarrow} \mathbb{G} \Gamma_2(x,W,\theta_0) \\
	&\equiv \widetilde{\mathbf{Z}}_4(x),
\end{align*}
a mean-zero Gaussian vector in $\R^2$.
\end{lemma}

\begin{proof}[Proof of lemma \ref{applemma:ATEcomp2 analytical boot}]
Write
\begin{align*}
	\mathbb{G}_n^* \Gamma_2(x,W,\widehat{\theta})
	&= \mathbb{G}_n^* \Gamma_2(x,W,\theta_0) + 
	\mathbb{G}_n^* \Big( \Gamma_2(x,W,\widehat{\theta}) - \Gamma_2(x,W,\theta_0) \Big).
\end{align*}
The first term converges to $\mathbb{G} \Gamma_2(x,W,\theta_0)$ by consistency of the standard nonparametric bootstrap. So it suffices to show that the second term converges to zero. We do this using an argument similar to that in the proof of lemma 19.24 in \cite{Vaart2000}. We'll give the proof for the upper bound, the first component of $\Gamma_2$. The proof for the lower bound is analogous.
\begin{enumerate}
\item By the proof of theorem \ref{thm:ATE convergence},
\[
	\overline{\mathcal{F}} = \{\overline{\Gamma}_{2}(x,W, \theta): \theta \in \Theta_\delta\}
\]
is Donsker with finite envelope function. So theorem 23.7 in \cite{Vaart2000} gives
\[
	\mathbb{G}_n^* \overline{\Gamma}_{2}(x,W, \cdot)
	\overset{P}{\rightsquigarrow}
	\mathbb{G} \overline{\Gamma}_{2}(x,W, \cdot),
\]
where $\mathbb{G}$ is a Gaussian process indexed by $\overline{\mathcal{F}}$.

\item Endow $\overline{\mathcal{F}}$ with the $L_2(\Prob)$ semi-metric. Note that 
\[
	\overline{\Gamma}_{2}(x,\cdot, \widehat{\theta}) \xrightarrow{p} \overline{\Gamma}_{2}(x,\cdot, \theta_0)
\]
in this semi-metric. This follows since, from the proof of theorem 1,
\[
	\left| \overline{\Gamma}_2(x,w,\widehat{\theta}) - \overline{\Gamma}_2(x,w,\theta_0) \right|
	\leq K(w) \| \widehat{\theta} - \theta_0 \|_{\Theta}
\]
where $\Exp(K(W)^2) < \infty$. Thus
\begin{align*}
	\int_\mathcal{W} \left| \overline{\Gamma}_2(x,w,\widehat{\theta}) - \overline{\Gamma}_2(x,w,\theta_0) \right|^2 \; dF_W(w)
	&\leq \Exp( K(W)^2) \| \widehat{\theta} - \theta_0 \|^2_\Theta,
\end{align*}
which converges to zero in probability by consistency of $\widehat{\theta}$ for $\theta_0$.
\end{enumerate}
These two points imply that $(\mathbb{G}_n^*, \overline{\Gamma}_{2}(x,\cdot, \widehat{\theta})) \overset{P}{\rightsquigarrow} (\mathbb{G}, \overline{\Gamma}_{2}(x,\cdot, \theta_0))$ in the space $\ell^\infty(\overline{\mathcal{F}}) \times \overline{\mathcal{F}}$ by Slutsky's theorem. Define the function $\phi :\ell^\infty(\overline {\mathcal{F}}) \times \overline{\mathcal{F}} \to \R$ by
\[
	\phi(g, \overline{\Gamma}_{2}(x,\cdot, \theta)) = g(\overline{\Gamma}_{2}(x,\cdot, \theta)) - g(\overline{\Gamma}_{2}(x,\cdot, \theta_0)).
\]
Since $\mathbb{G}$ has continuous paths (lemma 18.15 in \citealt{Vaart2000}) almost surely, the function $\phi$ is continuous at almost every $(\mathbb{G}, \overline{\Gamma}_{2}(x,\cdot, \theta_0))$. So the continuous mapping theorem (e.g., theorem 10.8 in \citealt{Kosorok2007}) implies that
\begin{align*}
 	\mathbb{G}_n^*(\overline{\Gamma}_{2}(x,W, \widehat{\theta}) - \overline{\Gamma}_{2}(x,W, \theta_0))
	&= \phi(\mathbb{G}_n^*, \overline{\Gamma}_{2}(x,\cdot, \widehat{\theta})) \\
	&\overset{P}{\rightsquigarrow} \phi(\mathbb{G}, \overline{\Gamma}_{2}(x,\cdot, \theta_0)) \\
	&= 0.
\end{align*}
\end{proof}

\begin{proof}[Proof of proposition \ref{prop:ATE analytical boot}]
First, since the influence functions for the first step estimators are Donsker (lemmas \ref{applemma:prop score estimation} and \ref{applemma:quantile estimation}), and since the standard nonparametric bootstrap for those estimators is valid (theorem 3.6.1 in \citealt{VaartWellner1996}), and along with our analysis in the proof of lemma \ref{applemma:ATEcomp2 analytical boot}, we have
\begin{align*}
	\begin{pmatrix}
			\sqrt{n}(\widehat{\theta}^* - \widehat{\theta})\\
			\mathbb{G}_n^* \Gamma_2(x,W,\widehat{\theta}) 
	\end{pmatrix} \overset{P}{\rightsquigarrow} \begin{pmatrix}
		\Z_1\\
		\widetilde{\Z}_{4}(x) 
	\end{pmatrix},
\end{align*}
a mean-zero process in $\R^{d_W} \times \ell^\infty([\varepsilon,1-\varepsilon],\R^{d_q}) \times \R^2$.

Next, define $\Lambda: \R^{d_W} \times \ell^\infty([\varepsilon,1-\varepsilon],\R^{d_q}) \times \R^2 \to \R^2$ by
\begin{align*}
	\Lambda(\theta, u) = \Gamma_3(x,\theta) + u
\end{align*}
where $\theta \in \R^{d_W} \times \ell^\infty([\varepsilon,1-\varepsilon],\R^{d_q})$ and $u \in \R^2$. 
By the proof of theorem \ref{thm:ATE convergence}, the mapping $\Gamma_3(x,\theta)$ is HDD at $\theta_0$ tangentially to $\R^{d_W} \times \mathscr{C}([\varepsilon,1-\varepsilon],\R^{d_q})$. So, for any $u_0 \in \R^2$, $\Lambda(\theta, u)$ is HDD at $(\theta_0,u_0)$ tangentially to $\R^{d_W} \times \mathscr{C}([\varepsilon,1-\varepsilon],\R^{d_q}) \times \R^2$. Its HDD is
\begin{align*}
	\Lambda'_{\theta_0} (h_1,h_2,h_3) &= \Gamma_{3,\theta_0}'(x,(h_1,h_2)) + h_3,
\end{align*}
where $h_3 \in \R^2$. Estimate it by
\[
	\widehat{\Lambda}'_{\theta_0}(h) = \widehat\Gamma_{3,\theta_0}'(x,(h_1,h_2)) + h_3.
\]
By the proof of lemma \ref{applemma:ATEcomp1 analytical boot}, this HDD estimator satisfies assumption 4 of \cite{FangSantos2014}. Thus we can apply their theorem 3.2 to get
\begin{align*}
	\widehat{\Gamma}_{3,\theta_0}'(x, \sqrt{n}(\widehat{\theta}^* - \widehat{\theta})) + \mathbb{G}_n^* \Gamma_2(x,W,\widehat{\theta}) &= \widehat{\Lambda}'_{\theta_0}(\sqrt{n}(\widehat{\theta}^* - \widehat{\theta}),\mathbb{G}_n^* \Gamma_2(x,W,\widehat{\theta}))\\
	 &\overset{P}{\rightsquigarrow} \Lambda'_{\theta_0}(\Z_1, \widetilde{\Z}_4(x))\\	
	&= \Gamma_{3,\theta_0}'(x, \Z_1) + \widetilde{\Z}_4(x)\\
	&\equiv \Z_4(x).
\end{align*}
\end{proof}

\section{Analytical Bootstrap Results for the CQTE and CATE}\label{sec:CQTEandCATEbootstrapResults}

In this section we formally derive bootstrap consistency for the CQTE and CATE. 

\begin{proposition}[CQTE Boostrap]\label{prop:CQTE analytical boot}
Suppose the assumptions of proposition \ref{prop:CQTE convergence} hold. Let $\kappa_n \to 0$, $n \kappa_n^2 \to \infty$, $\eta_n \to 0$, and $n\eta_n^2 \to \infty$ as $n\to\infty$. Then
\[
	\widehat{\Gamma}_{1,\theta_0}'(x,w,\tau,\sqrt{n}(\widehat\theta^* - \widehat\theta))
	\overset{P}{\rightsquigarrow}
	\Gamma_{1,\theta_0}'(x,w,\tau,\mathbf{Z}_1).
\]
\end{proposition}

This proposition implies that the asymptotic distribution of the CQTE bounds can be approximated by the bootstrap distribution of
\begin{align*}
	\begin{pmatrix}
		\widehat{\overline{\Gamma}}_{1,\theta_0}'(1,w,\tau,\sqrt{n}(\widehat\theta^* - \widehat\theta)) - \widehat{\underline{\Gamma}}_{1,\theta_0}'(0,w,\tau,\sqrt{n}(\widehat\theta^* - \widehat\theta)) \\[1em]
		\widehat{\underline{\Gamma}}_{1,\theta_0}'(1,w,\tau,\sqrt{n}(\widehat\theta^* - \widehat\theta)) - \widehat{\overline{\Gamma}}_{1,\theta_0}'(0,w,\tau,\sqrt{n}(\widehat\theta^* - \widehat\theta))
	\end{pmatrix}.
\end{align*}

We also show the bootstrap consistency for the CATE.

\begin{proposition}[CATE Bootstrap]\label{prop:CATE analytical boot}
Suppose the assumptions of proposition \ref{prop:CATE convergence} hold. Let $\kappa_n \to 0$, $n \kappa_n^2 \to \infty$, $\eta_n \to 0$, and $n\eta_n^2 \to \infty$ as $n\to\infty$. Then
\[
	\widehat{\Gamma}_{2,\theta_0}'(x,w,\sqrt{n}(\widehat\theta^* - \widehat\theta))
	\overset{P}{\rightsquigarrow}
	\Gamma_{2,\theta_0}'(x,w,\mathbf{Z}_1).
\]
\end{proposition}

Like with the CQTE, the asymptotic distribution of the CATE bounds can be approximated by the bootstrap distribution of
\begin{align*}
	\begin{pmatrix}
		\widehat{\overline{\Gamma}}_{2,\theta_0}'(1,w,\sqrt{n}(\widehat\theta^* - \widehat\theta)) - \widehat{\underline{\Gamma}}_{2,\theta_0}'(0,w,\sqrt{n}(\widehat\theta^* - \widehat\theta)) \\[1em]
		\widehat{\underline{\Gamma}}_{2,\theta_0}'(1,w,\sqrt{n}(\widehat\theta^* - \widehat\theta)) - \widehat{\overline{\Gamma}}_{2,\theta_0}'(0,w,\sqrt{n}(\widehat\theta^* - \widehat\theta))
	\end{pmatrix}.
\end{align*}

\subsection{Proofs} \label{sec:CQTEandCATEanalyticalboot}

\begin{proof}[Proof of proposition \ref{prop:CQTE analytical boot}]
As in the proof of lemma \ref{applemma:ATEcomp1 analytical boot} we'll use theorem 3.2 in \cite{FangSantos2014}. To do this we must verify their assumptions 1--4. Their assumptions 1--3 hold as in the proof of lemma \ref{applemma:ATEcomp1 analytical boot}. By their remark 3.4, sufficient conditions for their assumption 4 are:
\begin{enumerate}
\item \emph{A smoothness condition}: $\widehat{\overline{\Gamma}}_{1,\theta_0}'(x,w,\tau, h)$ is Lipschitz in $h$. This holds by lemma \ref{applemma:Gamma1 Lipschitz}.

\item \emph{A consistency condition}: $\widehat{\overline{\Gamma}}_{1,\theta_0}'(x,w,\tau, h)$ converges in probability to $\overline{\Gamma}_{1,\theta_0}'(x,w,\tau, h)$ for any $h \in \R^{d_W}\times \mathscr{C}([\varepsilon,1-\varepsilon],\R^{d_q})$. To see that this holds, recall that in the proof of lemma \ref{applemma:ATEcomp1 analytical boot} we showed that
\[q(x,w)' h_2(S_2(x,w,\tau, \widehat{\beta})) \xrightarrow{p} q(x,w)'  h_{2}(S_2(x,\tau, \beta_0)) \]
and
\[
	q(x,w)'\gamma_0'(S_2(x,w,\tau,\widehat{\beta}))T_2(x,w,\tau,\widehat{\beta},h_1,\kappa_n) \xrightarrow{p} q(x,w)'\gamma_0'(S_2(x,w,\tau,\beta_0))T_2(x,w,\tau,\beta_0,h_1,0).
\]
By the definition of $\widehat{\overline{\Gamma}}_{1,\theta_0}'$, these two results imply $\widehat{\overline{\Gamma}}_{1,\theta_0}'(x,w,\tau, h) \xrightarrow{p} \overline{\Gamma}_{1,\theta_0}'(x,w,\tau,h)$.
\end{enumerate}
Similar arguments can be used for the lower bound, to show that $\widehat{\underline{\Gamma}}_{1,\theta_0}'(x,w,\tau, h)$ is Lipschitz in $h$, and that $\widehat{\underline{\Gamma}}_{1,\theta_0}'(x,w,\tau, h) \xrightarrow{p} \underline{\Gamma}_{1,\theta_0}'(x,w,\tau, h)$.
\end{proof}

\begin{proof}[Proof of proposition \ref{prop:CATE analytical boot}]
The proof of this result is similar to the proof of proposition \ref{prop:CQTE analytical boot}. Like there, we show that the two sufficient conditions for assumption 4 in theorem 3.2 of \cite{FangSantos2014} hold.
\begin{enumerate}
\item \emph{A smoothness condition}: $\widehat{\overline{\Gamma}}_{2,\theta_0}'(x,w,h)$ is Lipschitz in $h$. To see this, write
\begin{align*}
		\left|\widehat{\overline\Gamma}_{2,\theta_0}'(x,w,\widetilde{h}) - \widehat{\overline\Gamma}_{2,\theta_0}'(x,w,h)\right|
		&\leq \int_0^1 \left|\widehat{\overline\Gamma}_{1,\theta_0}'(x,w,\tau, \widetilde{h}) - \widehat{\overline\Gamma}_{1,\theta_0}'(x,w,\tau, h)\right|d\tau\\
		&\leq \int_0^1 \overline{K}_1(x,w,\widehat{\theta}) \; d\tau \|\widetilde{h} - h\|_\Theta \\
		&= \overline{K}_1(x,w,\widehat{\theta}) \cdot \| \widetilde{h} - h \|_\Theta
	\end{align*}
where the second line follows by the proof of lemma \ref{applemma:Gamma1 Lipschitz}, and where $\overline{K}_1(x,w,\widehat{\theta})$ is defined in lemma \ref{applemma:Gamma1 Lipschitz} and is shown to be $O_p(1)$. So $\widehat{\overline{\Gamma}}_{2,\theta_0}'(x,w,h)$ is Lipschitz in $h$.

\item \emph{A consistency condition}: $\widehat{\overline{\Gamma}}_{2,\theta_0}'(x,w,h) \xrightarrow{p} \overline{\Gamma}_{2,\theta_0}'(x,w,h)$. This result follows from arguments similar to those in the proof of proposition \ref{prop:CATE convergence} and the dominated convergence theorem.
\end{enumerate}
Similar arguments can be used for the lower bound, to show that $\widehat{\underline\Gamma}_{2,\theta_0}'(x,w,h)$ is Lipschitz in $h$, and that $\widehat{\underline\Gamma}_{2,\theta_0}'(x,w,h) \xrightarrow{p} \underline\Gamma_{2,\theta_0}'(x,w,h)$.
\end{proof}

\section{Proofs for Section \ref{sec:standardboot}}\label{sec:standardbootProof}

\begin{proof}[Proof of theorem \ref{thm:standardbootATE}]
Recall that our ATE bounds depend on the functionals $\Gamma_3(x,\theta)$. We will show that $\Prob(p_{1|W} \in \{c,1-c\}) = 0$ implies that $\Gamma_{3,\theta_0}'(x,h)$ is linear in $h$. This, in turn, implies that $\Gamma_3(x,\theta)$ is Hadamard differentiable at $\theta_0$. Consistency of the standard bootstrap then follows from the delta method for the bootstrap: see theorem 3.9.11 in \cite{VaartWellner1996}.

\bigskip

Thus it suffices to show that $\Gamma_{3,\theta_0}'(x,h)$ is linear in $h$. We'll show this in three steps.

\medskip

\textbf{Step 1}. First we show that
\begin{align*}
	\overline\Gamma_{1,\theta_0}'(x,w,\tau,h) &= q(x,w)'\gamma_0'(S_2(x,w,\tau, \beta_0))T_2(x,w,\tau,\beta_0, h_1, 0) + q(x,w)'h_2(S_2(x,w,\tau, \beta_0)).
\end{align*}
is linear in $h$. First note that it is trivially linear in $h_2$. It is linear in $h_1$ if and only if $T_2(x,w,\tau,\beta_0,h_1,0)$ is linear in $h_1$. Recall the definition of $T_2$:
\begin{align*}
	T_2(x,w,\tau_0,\beta_0,h_1,0) &=T_1(x,w,\tau,\beta,h_1,0) \cdot \ind \Big( S_1(x,w,\tau,\beta_0) > \varepsilon   \Big) \\
	&\qquad + \max\left\{T_1(x,w,\tau,\beta_0, h_1, 0),0\right\} \cdot \ind \Big( S_1(x,w,\tau,\beta_0) = \varepsilon \Big)
\end{align*}
where
\[
	S_1(x,w,\tau, \beta) = \min\left\{\tau + \frac{c}{L(x,w'\beta)}\min\{\tau,1-\tau\},\frac{\tau}{L(x,w'\beta)},1-\varepsilon\right\}.
\]
There are two parts of $T_2$. We'll consider each of them separately.

\bigskip

\emph{The second part}. For given $(x,w,\tau,\beta_0,c,\varepsilon)$, the set
\[
	\{\tau \in (0,1): S_1(x,w,\tau,\beta_0) = \varepsilon\}
\]
has Lebesgue measure zero. This is the case since $S_1(x,w,\tau,\beta_0)$ is strictly increasing in $\tau$ whenever $S_1(x,w,\tau,\beta_0) < 1-\varepsilon$, and $\varepsilon < 1-\varepsilon$ by assuming $\varepsilon < 1/2$. Therefore, although
\[
	\max\left\{T_1(x,w,\tau,\beta_0, h_1, 0),0\right\}
\]
may be nonlinear in $h_1$,
\[
	\max\left\{T_1(x,w,\tau,\beta_0, h_1, 0),0\right\} \cdot \ind \Big( S_1(x,w,\tau,\beta_0) = \varepsilon \Big)
\]
is nonlinear for a measure zero set of $\tau$.

\bigskip

\emph{The first part}. Next we study linearity of $T_1(x,w,\tau,\beta_0,h_1,0)$ in $h_1$. Recall from appendix \ref{sec:HDDformulas} that it can be written as
\[
	T_1(x,w,\tau,\beta_0, h_1, 0) = \sum_{j=1}^7 T_{1,j}(x,w,\tau,\beta_0,h_1) \cdot \indicator_{1,j}(x,w,\tau,\beta_0, 0).
\]
By examining the specific functional forms given in appendix \ref{sec:HDDformulas}, we immediately see that the functions  $T_{1,j}(x,w,\tau,\beta_0,h_1)$ are linear for $j \in\{1,2,3\}$ and nonlinear for $j \in \{4,5,6,7\}$. The main question is for how many values of $(\tau,w)$ are the indicators $\indicator_{1,j}$ equal to 1 for $j\in\{4,5,6,7\}$. We'll show that these indicators are 1 only on a set of $\tau$'s and $w$'s that have measure zero under the product measure with the Lebesgue measure and $F_W$ as the marginals.

Define
\[
	\mathcal{S}_j(p_{x|w},c) \equiv \{\tau \in (0,1): \indicator_{1,j}(x,w,\tau,\beta_0, 0) = 1\}
\]
for $j \in \{ 4,5,6,7 \}$. For $j=4$, by the definition of $\indicator_{1,4}$ in appendix \ref{sec:HDDformulas} we have
\begin{align*}
	\mathcal{S}_4(p_{x|w},c) &= \left\{\tau \in (0,1): \tau + \frac{c}{p_{x|w}}\min\{\tau,1-\tau\} - \frac{\tau}{p_{x|w}}= 0\right\}\\
	&\qquad\qquad \cap \left\{\tau \in (0,1): \max\left\{\tau + \frac{c}{p_{x|w}}\min\{\tau,1-\tau\},\frac{\tau}{p_{x|w}}\right\} < 1-\varepsilon\right\}\\
	&\equiv \mathcal{S}_{4,a}(p_{x|w},c) \cap \mathcal{S}_{4,b}(p_{x|w},c).
\end{align*}
We write the first set as
\begin{align*}
	&\mathcal{S}_{4,a}(p_{x|w},c) \\
	&= \left\{\tau \in (0,1/2]: \tau\left( 1 + \frac{c}{p_{x|w}} - \frac{1}{p_{x|w}}\right) = 0\right\} \cup  \left\{\tau \in (1/2,1): \tau\left( 1 - \frac{c}{p_{x|w}} - \frac{1}{p_{x|w}}\right) = -\frac{c}{p_{x|w}}\right\}\\
	&= \begin{cases}
	\emptyset &\text{ if } c < 1-p_{x|w}\\
	(0,1/2] &\text{ if } c = 1-p_{x|w}\\
	\left\{\frac{c}{1+c-p_{x|w}}\right\} &\text{ if } c > 1-p_{x|w}.
	\end{cases}
\end{align*}
Thus if $c \neq 1 - p_{x \mid w}$ then $\mathcal{S}_{4,a}(p_{x|w},c)$ has Lebesgue measure zero. Consequently $\mathcal{S}_4(p_{x|w},c)$ also has Lebesgue measure zero in this case.

Next consider $j=5$. As before, by the definition of $\indicator_{1,5}$ we have
\begin{align*}
	\mathcal{S}_{5}(p_{x|w},c)
		&= \left\{\tau \in (0,1): \tau + \frac{c}{p_{x|w}}\min\{\tau,1-\tau\} -(1-\varepsilon) = 0\right\}\\
		&\qquad\qquad \cap \left\{\tau \in (0,1): \max\left\{\tau + \frac{c}{p_{x|w}}\min\{\tau,1-\tau\},1-\varepsilon\right\} < \frac{\tau}{p_{x|w}}\right\}\\
		&\equiv \mathcal{S}_{5,a}(p_{x|w},c) \cap \mathcal{S}_{5,b}(p_{x|w},c).
\end{align*}
Write the first set as
\begin{multline*}
	\mathcal{S}_{5,a}(p_{x|w},c) = \\
	\left\{\tau \in (0,1/2]: \tau\left( 1 + \frac{c}{p_{x|w}}\right) = 1-\varepsilon \right\} \cup  \left\{\tau \in (1/2,1): \tau\left( 1 - \frac{c}{p_{x|w}} \right) = 1-\varepsilon -\frac{c}{p_{x|w}}\right\}.
\end{multline*}
Since
\[
	1 + \frac{c}{p_{x|w}} \neq 0,
\]
the set
\[
	\left\{\tau \in (0,1/2]: \tau\left( 1 + \frac{c}{p_{x|w}}\right) = 1-\varepsilon \right\}
\]
contains at most one point. Likewise,
\[
	\left\{\tau \in (1/2,1):\tau\left( 1 - \frac{c}{p_{x|w}} \right) = 1-\varepsilon -\frac{c}{p_{x|w}}\right\}
\]
contains at most one point whenever $c \neq p_{x|w}$. When $c = p_{x|w}$ this set equals $\{ \tau \in (1/2,1) : 0 = -\varepsilon \}$, which is empty since $\varepsilon > 0$. Thus $\mathcal{S}_{5,a}(p_{x|w},c)$ has Lebesgue measure zero. Consequently, $\mathcal{S}_5(p_{x|w},c)$ also has Lebesgue measure zero.

Next consider $j=6$. We have
\begin{align*}
	\mathcal{S}_{6}(p_{x|w},c)
		&= \left\{\tau \in (0,1): \frac{\tau}{p_{x|w}} -(1-\varepsilon) = 0\right\}\\
		&\qquad \qquad \cap \left\{\tau \in (0,1): \max\left\{\frac{\tau}{p_{x|w}} , 1-\varepsilon \right\} < \tau + \frac{c}{p_{x|w}}\min\{\tau,1-\tau\}\right\}\\
		&\subseteq  \left\{\tau \in (0,1):\tau = (1-\varepsilon)p_{x|w} \right\}.
\end{align*}
The first line follows from the definition of $\indicator_{1,6}$. The last line follows from looking at the first set in the intersection in the first line. This last set is a singleton and hence has Lebesgue measure zero. Thus $\mathcal{S}_6(p_{x|w}, c)$ has Lebesgue measure zero.

Finally, consider $j=7$. This case is a combination of the above cases. Hence we can show that $\mathcal{S}_7(p_{x|w},c)$ has Lebesgue measure zero by repeating some of the above steps.

\bigskip

\textbf{Step 2}: From step 1 we see that for any $w \in \mathcal{W}$ such that $p_{x|w} \neq 1-c$ the mapping $\overline\Gamma_{1,\theta_0}'(x,w,\tau,h)$ is linear for all $\tau \in (0,1)$ except for a Lebesgue measure zero set. Denote this set by $\mathcal{T}$. Then
\begin{align*}
	\overline\Gamma_{2,\theta_0}'(x,w,h)
	&= \int_0^1 \overline\Gamma_{1,\theta_0}'(x,w,\tau,h) \; d\tau \\
	&= \int_{\tau \notin \mathcal{T}} \overline\Gamma_{1,\theta_0}'(x,w,\tau,h) \; d\tau + \int_{\tau \in \mathcal{T}} \overline\Gamma_{1,\theta_0}'(x,w,\tau,h) \; d\tau \\
	&= \int_{\tau \notin \mathcal{T}} \overline\Gamma_{1,\theta_0}'(x,w,\tau,h) \; d\tau.
\end{align*}
The first line follows by definition of $\overline{\Gamma}_2$. The last line follows since $\mathcal{T}$ has Lebesgue measure zero. Since integrals are linear operators, we see that $\overline\Gamma_{2,\theta_0}'(x,w,h)$ is linear in $h$  for any $w \in \mathcal{W}$ such that $p_{x|w} \neq 1-c$.

\bigskip

\textbf{Step 3}: We have
\begin{align*}
	\overline\Gamma_{3,\theta_0}'(x,h) &= \int_{\mathcal{W}} \overline\Gamma_{2,\theta_0}'(x,w,h) \; dF_W(w)\\
	&= \int_{\{ w \in \mathcal{W}: p_{x|w} \notin \{c,1-c\} \}} \overline\Gamma_{2,\theta_0}'(x,w,h) \; dF_W(w) + \int_{\{ w \in \mathcal{W}: p_{x|w} \in \{c,1-c\} \}} \overline\Gamma_{2,\theta_0}'(x,w,h) \; dF_W(w)\\
	&= \int_{\{w \in \mathcal{W}: p_{x|w} \notin \{c,1-c\} \}} \overline\Gamma_{2,\theta_0}'(x,w,h) \; dF_W(w).
\end{align*}
The first line follows by definition of $\overline{\Gamma}_3$. The third line follows since we assumed $p_{x|W} \in \{c,1-c\} $ occurs with probability zero. By step 2, $\overline{\Gamma}_{2,\theta_0}'(x,w,h)$ is linear on the set over which we are integrating in the last line. Since integrals are linear operators, this implies $\overline\Gamma_{3,\theta_0}'(x,h)$ is linear in $h$. 

\bigskip

Similar calculations for the lower bound show that $\underline\Gamma_{3,\theta_0}'(x,h)$ is linear when $\Prob(p_{x|W} \in \{c,1-c\}) = 0$. Thus we have shown that $\Gamma_{3,\theta_0}'(x,h)$ is linear in $h$, as desired.
\end{proof}

\begin{proof}[Proof of proposition \ref{prop:standardbootATT}]
By the proof of theorem \ref{thm:standardbootATE}, the mapping $\Gamma_3(0,\theta_0)$ is Hadamard differentiable when $\Prob(p_{0|W} = 1-c) = \Prob(p_{1|W} = c) = 0$. The conclusion then follows by the delta method for the bootstrap, theorem 3.9.11 in \cite{VaartWellner1996}.
\end{proof}

\end{document}